\documentclass[acmsmall,screen]{acmart}

\AtBeginDocument{%
  }

\setcopyright{cc}
\setcctype{by}
\acmDOI{10.1145/3763130}
\acmYear{2025}
\acmJournal{PACMPL}
\acmVolume{9}
\acmNumber{OOPSLA2}
\acmArticle{352}
\acmMonth{10}
\received{2025-03-25}
\received[accepted]{2025-08-12}

\usepackage{amsthm}

\usepackage{algorithm}

\usepackage{dsfont}
\usepackage{thm-restate}

\usepackage{hyperref}
\usepackage{cleveref}
\usepackage{tikz}
\usetikzlibrary{shapes,calc,positioning,patterns}
\usepackage{bussproofs}
\usepackage{stackengine}
\usepackage{wrapfig}

\usepackage{booktabs}
\usepackage{enumitem}

\usepackage{subcaption}
\newcommand{\nat}{\mathbb{N}}
\newcommand{\bool}{\mathbb{B}}
\newcommand{\ldot}{\mathpunct{.}}

\newcommand{\traceVars}{\mathcal{V}}
\newcommand{\trajVars}{\mathcal{B}}
\newcommand{\traceSet}{\mathbb{T}}
\newcommand{\traces}[1]{\mathit{Traces}(#1)}
\newcommand{\paths}[1]{\mathit{Paths}(#1)}

\newcommand{\calT}{\mathcal{T}}
\newcommand{\calG}{\mathcal{G}}
\newcommand{\calA}{\mathcal{A}}

\newcommand{\calL}{\mathcal{L}}

\newcommand{\calX}{\mathcal{X}}
\newcommand{\calF}{\mathcal{F}}
\newcommand{\calW}{\mathcal{W}}

\newcommand{\quant}{\mathds{Q}}

\usepackage{pifont}
\newcommand{\cmark}{\ding{51}}%
\newcommand{\xmark}{\ding{55}}%

\newcommand{\moved}{\mathit{moved}}

\newcommand{\game}[3]{\mathcal{G}_{#1, #2, #3}}
\newcommand{\refuter}{\mathfrak{R}}
\newcommand{\verifier}{\mathfrak{V}}

\newcommand{\trajE}{\text{\normalfont \sffamily E}}
\newcommand{\trajA}{\text{\normalfont\sffamily A}}
\newcommand{\trajQ}{\text{\normalfont\sffamily Q}}

\newcommand*\circled[1]{\tikz[baseline=(char.base)]{
		\node[shape=circle,draw,inner sep=0.7pt] (char) {#1};}}

\newcommand{\stage}{\mathit{stage}}

\newcommand{\fstage}{\protect\circled{\scriptsize$\forall$}}
\newcommand{\estage}{\protect\circled{\scriptsize$\exists$}}
\newcommand{\ustage}{\protect\circled{\scriptsize \normalfont \textsf{U}}}

\newcommand{\tool}[0]{\texttt{HyMCA}}
\newcommand{\hyperqb}[0]{\texttt{HyperQB}}

\providecommand{\ltlN}{\operatorname{%
		\tikz[baseline]{
			\draw[line width=.12ex]
			(0,.6ex) circle (.8ex);
}}}{}

\providecommand{\ltlF}{\operatorname{%
		\tikz[baseline]{
			\draw[line width=.12ex,line join=round]
			(0ex,.6ex) -- (.95ex,1.55ex) -- (1.9ex,.6ex) -- (.95ex,-.35ex) -- cycle;
}}}{}

\providecommand{\ltlG}{\operatorname{%
		\tikz[baseline]{
			\draw[line width=.12ex,line join=round]
			(0ex,-.2ex) -- (0ex,1.3ex) -- (1.5ex,1.3ex) -- (1.5ex,.-.2ex) -- cycle;
}}}{}

\DeclareMathOperator{\ltlU}{\mathcal{U}}
\DeclareMathOperator{\stutter}{\blacktriangleright\!}

\newcommand\xqed[1]{%
	\leavevmode\unskip\penalty9999 \hbox{}\nobreak\hfill
	\quad\hbox{#1}}
\newcommand\demo{\xqed{$\triangle$}}

\newcommand{\theory}{\mathfrak{T}}

\newcommand{\values}{\mathbb{V}}

\usepackage{listings}
\usepackage{relsize}

\definecolor{myOrange}{RGB}{151, 64,7}
\definecolor{myGrey}{RGB}{80, 80,80}
\definecolor{dkblue}{rgb}{0,0.1,0.5}

\colorlet{comment-color}{black!50}

\newtheorem{remark}{Remark}

\lstdefinelanguage{custom-lang}{
	keywords={while, if, then, else,readInput,repeat},
	keywordstyle=[1]\color{myOrange},
	morekeywords=[2]{true,false},
	keywordstyle=[2]\color{myGrey},
	morekeywords=[3]{product,nestedStateFormulas,LTLtoAPA},
	keywordstyle=[3]\color{dkblue},
	comment=[l][\color{comment-color}]{//},
	literate=%
	{=}{{\color{myOrange}=}}1
	{+}{{\color{myOrange}+}}1
	{/}{{\color{myOrange}/}}1
	{!}{{\color{myOrange}!}}1
	{xor}{{\color{myOrange}xor}}1
	{not}{{\color{myOrange}not}}1
	{*}{{\color{myOrange}*}}1
	{&}{{\color{myOrange}\&}}1
	{//}{{\color{black!50!white}/\!/}}1
	{|}{{{|}}}1
	{:}{{{:}}}1
	{:=}{{{:=}}}1
	{@}{ }1
}

\lstdefinestyle{default}{
	escapeinside={(*}{*)},
	basicstyle=\ttfamily\fontsize{10}{12},
	columns=fullflexible,
	commentstyle=\rmfamily\color{black!50!white},
	framexleftmargin=1em,
	framexrightmargin=1ex,
	keepspaces=true,
	keywordstyle=,
	mathescape,
	numbers=left,
	numberblanklines=false,
	numbersep=0.5em,
	numberstyle=\relscale{0.75}\color{gray}\ttfamily,
	showstringspaces=true,
	stepnumber=1,
	xleftmargin=1.2em,
	xrightmargin=-1em,
}

\lstnewenvironment{code}[1][]
{\small
	\lstset{
		style=default, 
		language=custom-lang,
		#1
	}
}
{}

\newif\iffullversion
\fullversiontrue

\newcommand{\ifFull}[2]{\iffullversion#1\else#2\fi}

\begin{document}

\title{Verifying Asynchronous Hyperproperties in Reactive Systems}

%%
%% The "author" command and its associated commands are used to define
%% the authors and their affiliations.
%% Of note is the shared affiliation of the first two authors, and the
%% "authornote" and "authornotemark" commands
%% used to denote shared contribution to the research.
\author{Raven Beutner}
\email{raven.beutner@cispa.de}
\orcid{0000-0001-6234-5651}
\affiliation{%
	\institution{CISPA Helmholtz Center for Information Security}
	\country{Germany}
}

\author{Bernd Finkbeiner}
\email{finkbeiner@cispa.de}
\orcid{0000-0002-4280-8441}
\affiliation{%
  \institution{CISPA Helmholtz Center for Information Security}
  \country{Germany}
}

%%
%% By default, the full list of authors will be used in the page
%% headers. Often, this list is too long, and will overlap
%% other information printed in the page headers. This command allows
%% the author to define a more concise list
%% of authors' names for this purpose.
%\renewcommand{\shortauthors}{Raven Beutner and Bernd Finkbeiner}

\begin{abstract}
Hyperproperties are system properties that relate multiple execution traces and commonly occur when specifying information-flow and security policies.
Logics like HyperLTL utilize explicit quantification over execution traces to express temporal hyperproperties in reactive systems, i.e., hyperproperties that reason about the temporal behavior along infinite executions.
An often unwanted side-effect of such logics is that they compare the quantified traces \emph{synchronously}.
This prohibits the logics from expressing properties that compare multiple traces asynchronously, such as Zdancewic and Myers's \emph{observational determinism}, McLean's \emph{non-inference}, or \emph{stuttering refinement}. 
We study the model-checking problem for a variant of \emph{asynchronous HyperLTL} (A-HLTL), a temporal logic that can express hyperproperties where multiple traces are compared across timesteps. 
In addition to quantifying over system traces, A-HLTL features secondary quantification over stutterings  of these traces. 
Consequently, A-HLTL allows for a succinct specification of many widely used asynchronous hyperproperties. 
Model-checking A-HLTL requires finding suitable stutterings, which, thus far, has been only possible for very restricted fragments or \emph{terminating} systems. 
In this paper, we propose a novel game-based approach for the verification of arbitrary $\forall^*\exists^*$  A-HLTL formulas in \emph{reactive} systems. 
In our method, we consider the verification as a game played between a verifier and a refuter, who challenge each other by controlling parts of the underlying traces and stutterings. 
A winning strategy for the verifier then corresponds to concrete witnesses for existentially quantified traces and asynchronous alignments for existentially quantified stutterings.
We identify fragments for which our game-based interpretation is complete and thus constitutes a finite-state decision procedure. 
We contribute a prototype implementation for finite-state systems and report on encouraging experimental results. 
\end{abstract}

%%
%% The code below is generated by the tool at http://dl.acm.org/ccs.cfm.
%% Please copy and paste the code instead of the example below.
%%



%\received{20 February 2007}
%\received[revised]{12 March 2009}
%\received[accepted]{5 June 2009}

\maketitle

\section{Introduction}\label{sec:intro}

In 2008, \citet{ClarksonS08} coined the term \emph{hyperproperties} for the rich class of system requirements that relate multiple computations. 
Contrary to traditional trace properties (i.e., properties that reason about individual executions, formally defined as \emph{sets of traces}), hyperproperties (formally defined as \emph{sets of sets of traces}) capture the interaction of multiple computations.
This covers a wide range of requirements, including information-flow policies \cite{GoguenM82a,GuarnieriKMRS20}, robustness \cite{BiewerDFGHHM22}, continuity \cite{ChaudhuriGL12}, knowledge \cite{BozzelliMP15,BeutnerFFM23}, and linearizability \cite{HerlihyW90}.
For example, \citet{ZdancewicM03}'s seminal definition of \emph{observational determinism} (OD) requires that all \emph{pairs} of executions with the same low-security input have the same sequence of low-security observations.
Likewise, \citet{McLean94}'s \emph{non-inference} requires that for every execution of the system, there exists a matching execution that has the same (sequence of) low-security observations despite having a fixed ``dummy'' high-security input; both are hyperproperties.

Missing from \citet{ClarksonS08}'s general definition was, however, a concrete specification language for hyperproperties. 
The introduction of HyperLTL \cite{ClarksonFKMRS14} sparked an extensive development of \emph{temporal} logics for expressing and reasoning about hyperproperties in reactive systems, i.e., systems that continuously interact with an environment and thus produce \emph{infinite} execution traces. 
HyperLTL extends linear-time temporal logic (LTL) \cite{Pnueli77} with explicit quantification over traces in a system.
For example, we can express a simplified form of observational determinism as $\varphi^\mathit{syn}_{\mathit{OD}} := \forall \pi_1. \forall \pi_2. (l_{\pi_1} = l_{\pi_2}) \to \ltlG (o_{\pi_1} = o_{\pi_2})$, stating that all pairs of executions $\pi_1, \pi_2$ with initially the same low-security input (modeled via variable $l$), should globally (expressed using LTL's globally operator $\ltlG$) produce the same output (modeled by $o$). 

Crucially, the semantics of HyperLTL -- and that of most other logics for temporal hyperproperties \cite{GutsfeldMO20,FinkbeinerHHT20,GiacomoFMP21,CoenenFHH19,Rabe16,BajwaZCV23,BeutnerF23,BeutnerFFM23,DimitrovaFT20,AbrahamBBD20,AbrahamB18} -- is \emph{synchronous}.
That is, the logic can relate multiple traces in a system, but -- during the evaluation of the LTL body -- time progresses in lock-step on all traces. 
This is too restrictive for many properties.
As a simple example, consider the program in \Cref{fig:insec}.
The program branches on the (Boolean) high-security input $h$ and updates $o$ in either branch to \lstinline[style=default, language=custom-lang]|$o$ xor $l$|.
Initially equal values of $l$ thus imply that the sequence of outputs is the same; the program satisfies \citet{ZdancewicM03}'s original definition of OD. 
The program does, however, \emph{violate} the synchronous HyperLTL property expressed in $\varphi^\mathit{syn}_{\mathit{OD}}$; the update of $o$ is delayed by one step in the first branch (due to the intermediate write to $t$ in line \ref{line:true1}), so the LTL body is violated during its synchronous evaluation. 
\Cref{fig:traces} depicts two example executions of \Cref{fig:insec}, illustrating how different executions delay the time point where the output changes.

\begin{figure}[!t]
	
	\begin{subfigure}{0.3\linewidth}
		\newsavebox{\myboxi}
		\begin{lrbox}{\myboxi}
\begin{code}
$o$ = $\bot$ (*\label{line:init}*)
repeat (*\label{line:loop}*)
@@$h$ = readInput$_H$() (*\label{line:read}*)
@@if $h$ then (*\label{line:branch}*)
@@@@$\mathit{t}$ = $o$ xor $l$ (*\label{line:true1}*)
@@@@$o$ = $\mathit{t}$ (*\label{line:true2}*)
@@else 
@@@@$o$ = $o$ xor $l$ (*\label{line:false}*)
\end{code}
		\end{lrbox}
		\centering
		\scalebox{0.9}{\usebox{\myboxi}}
		\subcaption{}\label{fig:insec}
	\end{subfigure}%
	\begin{subfigure}{0.7\linewidth}
		\centering
		\small
		\scalebox{0.8}{
			\begin{tikzpicture}
				\node[] at (-2.5, 0) () {\large $\pi_1$:};
				\node[] at (-2.5, -2.75) () {\large $\pi_2$:};
				
				\node[rectangle,align=center, draw, thick,minimum width=9mm] at (-1.5,0) (n0) {\ref{line:init}\\[1mm] $\begin{aligned}
						&l \mapsto \top \\
						&o \mapsto \_\\ 
						&h \mapsto \_ \\
						&t \mapsto \_
					\end{aligned}$};
				
				\node[rectangle,align=center, draw, thick,minimum width=9mm] at (0,0) (n1) {\ref{line:loop}\\[1mm] $\begin{aligned}
						&l \mapsto \top \\
						&o \mapsto \bot\\ 
						&h \mapsto \_ \\
						&t \mapsto \_
					\end{aligned}$};
				
				\node[rectangle,align=center, draw, thick,minimum width=9mm] at (1.5,0) (n2) {\ref{line:read}\\[1mm] $\begin{aligned}
						&l \mapsto \top\\
						&o \mapsto \bot\\ 
						&h \mapsto \_\\
						&t \mapsto \_
					\end{aligned}$};
				
				\node[rectangle,align=center, draw, thick,minimum width=9mm] at (3,0) (n3) {\ref{line:branch}\\[1mm] $\begin{aligned}
						&l \mapsto \top\\
						 &o \mapsto \bot\\ 
						&h \mapsto \bot\\
						&t \mapsto \_
					\end{aligned}$};
				
				\node[rectangle,align=center, draw, thick,minimum width=9mm] at (4.5,0) (n4) {\ref{line:false}\\[1mm] $\begin{aligned}
						&l \mapsto \top\\
						&o \mapsto \top\\ 
						&h \mapsto \bot\\
						&t \mapsto \_
					\end{aligned}$};
				
				\node[rectangle,align=center, draw, thick,minimum width=9mm] at (6,0) (n5) {\ref{line:loop}\\[1mm] $\begin{aligned}
						&l \mapsto \top\\
						&o \mapsto \top\\ 
						&h \mapsto \bot\\
						&t \mapsto \_
					\end{aligned}$};
				
				\node[rectangle,align=center, draw, thick,minimum width=9mm] at (-1.5,-2.75) (m0) {\ref{line:init}\\[1mm] $\begin{aligned}
						&l \mapsto \top\\
						&o \mapsto \_\\ 
						&h \mapsto \_\\
						&t \mapsto \_
					\end{aligned}$};

				\node[rectangle,align=center, draw, thick,minimum width=9mm] at (0,-2.75) (m1) {\ref{line:loop}\\[1mm] $\begin{aligned}
						&l \mapsto \top\\
						&o \mapsto \bot\\ 
						&h \mapsto \_\\
						&t \mapsto \_
					\end{aligned}$};
				
				\node[rectangle,align=center, draw, thick,minimum width=9mm] at (1.5,-2.75) (m2) {\ref{line:read}\\[1mm] $\begin{aligned}
						&l \mapsto \top\\
						&o \mapsto \bot\\ 
						&h \mapsto \_\\
						&t \mapsto \_
					\end{aligned}$};
				
				\node[rectangle,align=center, draw, thick,minimum width=9mm] at (3,-2.75) (m3) {\ref{line:branch}\\[1mm] $\begin{aligned}
						&l \mapsto \top\\
						&o \mapsto \bot,\\ 
						&h \mapsto \top\\
						&t \mapsto \_
					\end{aligned}$};
				
				\node[rectangle,align=center, draw, thick,minimum width=9mm] at (4.5,-2.75) (m4) {\ref{line:true1}\\[1mm] $\begin{aligned}
						&l \mapsto \top\\
						&o \mapsto \top,\\ 
						&h \mapsto \top\\
						&t \mapsto \_
					\end{aligned}$};
				
				\node[rectangle,align=center, draw, thick,minimum width=9mm] at (6,-2.75) (m5) {\ref{line:true2}\\[1mm] $\begin{aligned}
						&l \mapsto \top\\
						&o \mapsto \top,\\ 
						&h \mapsto \top\\
						&t \mapsto \bot
					\end{aligned}$};
				
				\node[rectangle,align=center, draw, thick,minimum width=9mm] at (7.5,-2.75) (m6) {\ref{line:loop}\\[1mm] $\begin{aligned}
						&l \mapsto \top\\
						&o \mapsto \bot,\\ 
						&h \mapsto \top\\
						&t \mapsto \bot
					\end{aligned}$};
				
				\draw[-, thick] (n0) -- (n1);
				\draw[-, thick] (n1) -- (n2);
				\draw[-, thick] (n2) -- (n3);
				\draw[-, thick] (n3) -- (n4);
				\draw[-, thick] (n4) -- (n5);
				\draw[-, thick,dotted] (n5) -- (7,0);
				
				\draw[-, thick] (m0) -- (m1);
				\draw[-, thick] (m1) -- (m2);
				\draw[-, thick] (m2) -- (m3);
				\draw[-, thick] (m3) -- (m4);
				\draw[-, thick] (m4) -- (m5);
				\draw[-, thick] (m5) -- (m6);
				\draw[-, thick,dotted] (m6) -- (8.5,-2.75);
				
				\draw[-, very thick, dashed,black!50] (n0) -- (m0);
				\draw[-, very thick, dashed,black!50] (n1) -- (m1);
				\draw[-, very thick, dashed, black!50] (n2) -- (m2);
				\draw[-, very thick, dashed, black!50] (n3) -- (m3);
				\draw[-, very thick, dashed, black!50] (n4) -- (m4);
				\draw[-, very thick, dashed, black!50] (n4) -- (m5);
				\draw[-, very thick, dashed, black!50] (n5) -- (m6);
				
			\end{tikzpicture}
		}
		\subcaption{}\label{fig:traces}
	\end{subfigure}
	\vspace{-4mm}
	\caption{
		\Cref{fig:insec} depicts a Boolean program over variables $o, l$, $h$, and $t$.
		In \Cref{fig:traces}, we depict two executions $\pi_1, \pi_2$.
		Each state contains the current program line and the current value of all variables (initially, we set the low-security input $l$ to $\top$). 
		On $\pi_1$, the read in line \ref{line:read} sets $h$ to $\bot$, and on $\pi_2$ it assigns $h$ to $\top$. 
	}
\end{figure}
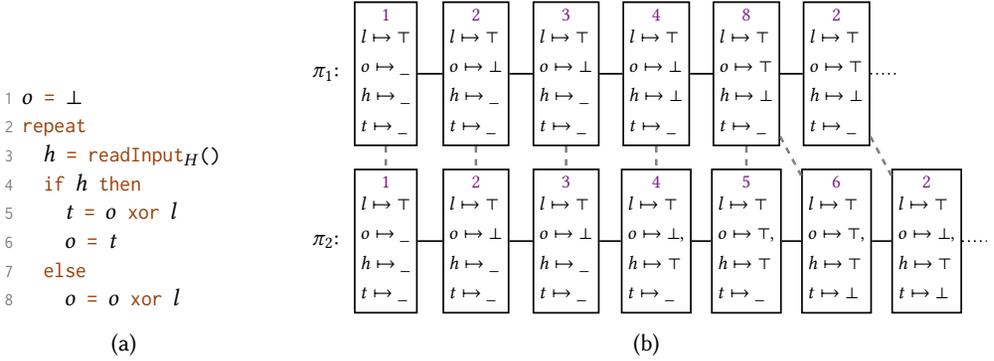

\paragraph{Asynchronous Hyperproperties}

Many security and information-flow properties -- particularly in the study of \emph{distributed} programs -- thus cannot be expressed in HyperLTL's rigid synchronous semantics.
This points to a sharp contrast: Hyperproperties, per \citet{ClarksonS08}'s definition, were never ``synchronous'' and, instead, loosely defined as \emph{sets of sets of traces} and prominent properties \cite{ZdancewicM03,McLean94,McCullough88} only reason about sequences of events, without any form of synchronous timesteps. 
Yet, most existing \emph{logics} for expressing temporal hyperproperties enforce a synchronous traversal of all traces.

\paragraph{Asynchronous HyperLTL}

It turns out that we can significantly extend the expressiveness of HyperLTL -- allowing us to precisely capture many security properties out of reach for synchronous hyperlogics -- using a simple idea: \emph{stuttering}. 
Formally, Asynchronous HyperLTL (A-HLTL for short) extends HyperLTL with explicit quantification over stutterings of system traces \cite{BaumeisterCBFS21}.
In this paper, we use a novel (yet equivalent) variant of A-HLTL that directly quantifies over stuttered traces instead of using the so-called trajectories used by \citet{BaumeisterCBFS21} (see \Cref{sec:ahltl} for details).
Given a trace $\pi$ in the system, we write $\beta \stutter \pi$ to denote that the trace $\beta$ is a fair stuttering of trace $\pi$.
In A-HLTL, we can then state
\begin{align}\label{eq:asyn-od}
	\forall \pi_1. \forall \pi_2. \exists \beta_1 \stutter \pi_1. \exists \beta_2 \stutter \pi_2 \ldot (l_{\beta_1} = l_{\beta_2}) \to  \ltlG (o_{\beta_1} = o_{\beta_2}) \tag{$\varphi_{\mathit{OD}}$}
\end{align}
requiring that for \emph{every} pair of traces $\pi_1, \pi_2$ in the system, there \emph{exist} stutterings $\beta_1$ of $\pi_1$ and $\beta_2$ of $\pi_2$, such that the stuttered traces $\beta_1, \beta_2$ satisfy $(l_{\beta_1} = l_{\beta_2}) \to  \ltlG (o_{\beta_1} = o_{\beta_2})$.
\ref{eq:asyn-od} thus requires that all pairs of traces $\pi_1, \pi_2$ with an initially equal value of $l$, are \emph{stutter-equivalent} on $o$, \emph{precisely} capturing \citet{ZdancewicM03}'s original definition of observational determinism. 
The program in \Cref{fig:insec} satisfies \ref{eq:asyn-od}: any two traces traverse the same sequence of output values and can, therefore, be aligned by stuttering appropriately. 
In \Cref{fig:traces}, we visualize a possible stuttering that aligns the outputs using the dashed lines. 

\paragraph{Verifying Asynchronous HyperLTL}

\citet{BaumeisterCBFS21} and \citet{HsuBFS23} demonstrate that A-HLTL is an expressive logic that allows for succinct high-level specifications of, e.g.,  many commonly used information-flow policies and asynchronous \emph{refinement} properties (see \Cref{sec:related-work} for more details).
In particular, we can use A-HLTL to directly \emph{quantify} over stutterings without manually identifying alignment points (such as positions where the output of a system changes, or non-visible steps in a refinement property); see \Cref{sec:related-work} for an extensive discussion.
While A-HLTL seems like a simple extension of HyperLTL, the quantification over stutterings significantly complicates model-checking, so we cannot apply existing algorithms developed for HyperLTL (which often heavily rely on HyperLTL's synchronous evaluation by, e.g., using automata to summarize trace combinations \cite{FinkbeinerRS15,BeutnerF23b}).
Instead, existing approaches for A-HLTL either employ a bounded unrolling of the system and are thus limited to \emph{terminating} systems (i.e., systems that reach a final state after a \emph{fixed} number of steps) \cite{HsuBFS23} or consider restricted fragments of A-HLTL that can be manually reduced to HyperLTL \cite{BaumeisterCBFS21}.

\paragraph{Game-Based Verification for A-HLTL}

In this paper, we propose a novel method for the verification of A-HLTL in reactive systems that can verify relevant properties well beyond the fragments supported by previous methods. 
In our approach, we interpret the verification of an A-HLTL formula as a game between a verifier and a refuter; similar to successful approximations for synchronous HyperLTL \cite{CoenenFST19,BeutnerF22CSF}.
In our game, we construct the traces and stutterings \emph{step-wise} and yield the control of traces -- \emph{and their stutterings} -- to the players. 
Intuitively, in each step, the refuter can extend all universally quantified traces and decide on whether universally quantified stutterings should be progressed or stuttered.
The verifier can make similar decisions for all existentially quantified traces and stutterings. 
We provide an overview in \Cref{sec:overview}.
We show that our game-based approximation is sound: If the verifier wins the game, the A-HLTL formula is satisfied.
As model-checking of A-HLTL is, in general, undecidable (even in finite-state systems!) \cite{BaumeisterCBFS21}, our approach is necessarily incomplete.
We identify fragments of A-HLTL for which our game-based interpretation is complete and thus constitutes a finite-state \emph{decision procedure}.
In particular, we prove that our method  is complete for the fragments supported by previous approaches \cite{BaumeisterCBFS21,HsuBFS23}. 

\begin{wrapfigure}{R}{0.25\linewidth}
	\vspace{-2mm}
	\centering
		\begin{tikzpicture}
			\node[circle,draw, thick,label=below:{\small$\{a = 0\}$}] at (0, 0) (n0) { $s_0$};
			
			\node[circle,draw,thick,label=above:{\small$\{a = 1\}$}] at (1.2, 0.75) (n1) { $s_1$};
			
			\node[circle,draw,thick,label=below:{\small$\{a = 1\}$}] at (1.2, -0.75) (n2) { $s_2$};
			
			\draw[->, thick] (n0) -- (n1); 
			
			\draw[->, thick] (-0.6,0) -- (n0); 
			
			\draw[->, thick] (n1) edge[bend right] (n2);
			\draw[->, thick] (n2) edge[bend right] (n1);
			
			\draw[->, thick] (n2) -- (n0); 
			
			\draw [->, thick] (n0) edge[loop above] (n0);
		\end{tikzpicture}
	
	\caption{Example system}\label{fig:system-asyn}
\end{wrapfigure}
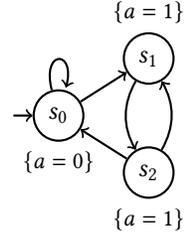

\paragraph{Implementation.}
Our game-based verification method applies to finite and infinite-state systems (see \Cref{sec:related-work}).
To compare with existing approaches and tools, we implement our verification approach for finite-state transition systems in a prototype tool and compare it against a bounded model-checking approach \cite{HsuBFS23}.

\paragraph{Structure and Contributions.}

In short, our contributions include the following: 
\textbf{(1)} We introduce a simpler (yet equivalent) variant of A-HLTL that directly quantifies over stuttering (\Cref{sec:ahltl}); \textbf{(2)} We propose a novel game-based semantics for the verification of asynchronous hyperproperties specified in A-HLTL (\Cref{sec:verification}); \textbf{(3)} We identify fragments for which our game-based interpretation is complete and thus yields a decision procedure for finite-state systems. Most notably, this includes all admissible formulas \cite{BaumeisterCBFS21}, for which we prove that the verifier can follow a canonical (``maximal'') stuttering strategy (\Cref{sec:complete});  \textbf{(4)} We implement our approach for finite-state systems and evaluate it on existing benchmarks (\Cref{sec:eval}).

\section{Overview}\label{sec:overview}

In this section, we illustrate how we can use games to verify hyperproperties that quantify over stutterings. 
Consider the $3$-state transition system $\calT$ in \Cref{fig:system-asyn}, where each state assigns variable $a$ to some integer value, and the following A-HLTL formula
\begin{align}\label{eq:fair}
	\forall \pi_1. \exists \pi_2. \exists\beta_1 \stutter \pi_1. \exists \beta_2 \stutter \pi_2\ldot \ltlN \ltlG (a_{\beta_1} \neq a_{\beta_2}).\tag{$\varphi_{\mathit{fair}}$}
\end{align}
This formula  requires that for any execution $\pi_1$, there exists some execution $\pi_2$ and some stutterings $\beta_1, \beta_2$ (of $\pi_1$ and $\pi_2$, respectively) such that, starting from the second step (using LTL's next operator $\ltlN$), the value of $a$ differs between $\beta_1$ and $\beta_2$.\footnote{The A-HLTL formula we consider here is meant as a simple example that demonstrates the key ideas underlying our verification method. 
Still, if we interpret $a = v$ (for $v \in \{0,1\}$) as ``a grant was given to process $v$'', we can view \ref{eq:fair} as a simplified \emph{fairness} property: it requires that for each execution, some other execution results in the exact opposite grant assignment (up to stuttering), i.e., both processes can receive the grant \emph{equally often}.
}
This property holds on $\calT$. 
For example, take the trace $\pi_1 = \{a=0\}\{a=0\}\big(\{a=1\}\big)^\omega$. 
We can match this execution by choosing $\pi_2 = \{a=0\}\{a=1\}\{a=1\}\big(\{a=0\}\big)^\omega$, and align both traces by picking stutterings $\beta_1 = \{a=0\}\{a=0\}\{a=0\}\big(\{a=1\}\big)^\omega$, and $\beta_2 = \pi_2$ (note that the ``obvious'' choice $\pi_2 = \{a=0\}\{a=1\}\big(\{a=0\}\big)^\omega$ is \emph{not} a trace in $\calT$).

When verifying this property directly, we would need to, for each choice of $\pi_1$ (where quantification ranges over infinite traces in $\calT$), provide a witnessing trace $\pi_2$ \emph{and} find an appropriate stuttering for each combination of $\pi_1, \pi_2$.
Instead of tackling this problem directly, we \emph{approximate} the verification as a game between a verifier and a refuter.\footnote{Such approximations have proven successful in the verification of synchronous HyperLTL \cite{BeutnerF22CSF,CoenenFST19,BeutnerF25b,BeutnerF24a}, where finite-state model-checking is decidable, and games offer a computationally cheaper (incomplete) approximation.  
	In our setting of A-HLTL, model-checking is \emph{undecidable}, and we design a novel game as a sound \emph{approximation} of A-HLTL.  }
The main idea of our game is to view paths for $\pi_1$ and $\pi_2$ and the stutterings $\beta_1, \beta_2$ as the outcome of an infinite-duration two-player game between a refuter (who controls \emph{universally} quantified traces and stutterings) and a verifier (who controls \emph{existentially} quantified traces and stutterings). 
For our example formula \ref{eq:fair}, the refuter constructs the (universally quantified) trace $\pi_1$ step-wise by moving along the transitions of $\calT$. 
The verifier can react by moving through a separate copy of $\calT$, thereby generating a trace for $\pi_2$, and controls the stutterings $\beta_1, \beta_2$.
Each infinite game-play thus generates infinite traces for $\pi_1$ and $\pi_2$ and stutterings $\beta_1, \beta_2$ of $\pi_1, \pi_2$, respectively; the objective of the verifier is then to ensure that the stutterings $\beta_1, \beta_2$ (together) satisfy $\ltlN \ltlG (a_{\beta_1} \neq a_{\beta_2})$, the LTL body of \ref{eq:fair}.

\paragraph{Graph-based Games.}
To formalize the interaction of refuter and verifier, we construct a graph-based game. 
The game consists of a set of vertices connected via edges, where each vertex is assigned to either the refuter or the verifier.
The game progresses by moving a single token from vertex to vertex, and the player controlling the current vertex can decide which outgoing edge the token should take.
The outgoing edges of a given vertex thus determine the possible decision a player can make in the given situation (we give more details on graph-based games in \Cref{sec:prelim}).

\paragraph{Verification Game.}

There are two principled ideas underlying our verification game for A-HLTL:  Windows of states, and relative pointers.
The refuter (resp.~verifier) controls trace $\pi_1$ (resp.~$\pi_2$) by moving through $\calT$.
To accommodate the fact that $\beta_1$ and $\beta_2$ are stutterings of $\pi_1$ and $\pi_2$ (i.e., in the $i$th step of stuttering $\beta_1$ can equal the $j$th step of $\pi_1$ for any $j \leq i$, depending on the speed of the stuttering), we do not maintain a single state but a \emph{window} of states, which we can think of a \emph{sliding window} of a trace.
We then identify the stutterings $\beta_1, \beta_2$ as \emph{relative pointers} into the state-windows of $\pi_1$ and $\pi_2$.
Intuitively, this pointer identifies the current position of the stuttering, i.e., $\beta_1$ (resp.~$\beta_2$) points to one position of (the state-window of) $\pi_1$ (resp.~$\pi_2$).
Whenever we require the current state of $\beta_1$, i.e., we fetch the state in $\pi_1$'s window that is pointed to by the $\beta_1$ pointer.

Each vertex then belongs to one of three game states: the $\ustage$-stage (update-stage), the $\fstage$-stage, or $\estage$-stage.
The verifier controls all vertices in the $\estage$-stage, and the refuter those in the $\fstage$-stage (vertices in the $\ustage$-stage have a unique successor so it does not matter which player controls them).
Whenever the game is in the $\ustage$-stage, we consider the current states of $\beta_1, \beta_2$, and evaluate one step of the LTL formula $\ltlN \ltlG (a_{\beta_1} \neq a_{\beta_2})$. 
In order to win, the verifier thus needs to make sure that, whenever the game is in the $\ustage$-stage (except on the first visit), $a_{\beta_1} \neq a_{\beta_2}$ holds.\footnote{
	To track this temporal property, we translate the LTL formula to a deterministic automaton and track the state of this automaton within each $\ustage$-stage vertex.
	For simplicity, we omit this automaton in this overview section.}
After the $\ustage$-stage, we progress to the $\fstage$-stage.
In the $\fstage$-stage, the refuter can progress the (universally quantified) $\pi_1$; the state window for $\pi_2$ remains unchanged.
As we track a window of states, this corresponds to \emph{appending} a state to that window along $\calT$'s transitions. 
In the $\estage$-stage, the verifier can respond by appending a state to $\pi_2$'s window.
Additionally, the verifier can decide if the (\emph{existentially quantified}) stutterings $\beta_1$ and $\beta_2$ should progress to the next state in the state window of $\pi_1$ and $\pi_2$, respectively.
The verifier can thus (implicitly) control the stutterings $\beta_1, \beta_2$ by controlling the relative pointers.
By advancing the $\beta_1$-pointer on $\pi_1$ (i.e., move the $\beta_1$-pointer to the next state in the state window assigned to $\pi_1$), the stuttering $\beta_1$ does a proper (non-stuttering) step; if the verifier leaves the $\beta_1$-pointer unchanged, it effectively performs a stuttering step by repeating the same state it pointed to in the previous round. 

\begin{figure}[!t]
	\centering
	\scalebox{0.8}{
		\begin{tikzpicture}
			
			%%%%%%%%%%%%%%%%%%%%%%%%%%%%%%%%
			\coordinate (c) at (0, 0);
			\node[draw, minimum height=22mm,minimum width=24mm, rounded corners=1mm,thick,anchor=west] at (c) (b0) {};
			\node[anchor=east] at ([xshift=13mm,yshift=3mm]c)  {\small$\pi_1$:};
			\node[draw, minimum height=5mm,minimum width=5mm] (n) at ([xshift=15mm,yshift=3mm]c)  {\footnotesize $s_0$};
			\node[inner sep=1pt] (m) at ([xshift=15mm,yshift=9mm]c) {\tiny $\beta_1$};
			\draw[->, thick] (m) -- (n); 
			\node[anchor=east] at ([xshift=13mm,yshift=-3mm]c)  {\small$\pi_2$:};
			\node[draw, minimum height=5mm,minimum width=5mm] (n) at ([xshift=15mm,yshift=-3mm]c)  {\footnotesize$s_0$};
			\node[inner sep=1pt] (m) at ([xshift=15mm,yshift=-9mm]c) {\tiny $\beta_2$};
			\draw[->, thick] (m) -- (n); 
			
			\node[] at ([xshift=3mm,yshift=4mm]c) {$v_0$};
			\node[] at ([xshift=3mm,yshift=-4mm]c) {$\ustage$};
			\draw[-, dashed,thick] (c) --  ([xshift=6mm]c);
			\draw[-, dashed,thick] ([xshift=6mm,yshift=11mm]c) --  ([xshift=6mm,yshift=-11mm]c);
			%%%%%%%%%%%%%%%%%%%%%%%%%%%%%%%%
			
			%%%%%%%%%%%%%%%%%%%%%%%%%%%%%%%%
			\coordinate (c) at (2.7, 0);
			\node[draw, minimum height=22mm,minimum width=24mm, rounded corners=1mm,thick,anchor=west] at (c) (b1) {};
			\node[anchor=east] at ([xshift=13mm,yshift=3mm]c)  {$\pi_1$:};
			\node[draw, minimum height=5mm,minimum width=5mm] (n) at ([xshift=15mm,yshift=3mm]c)  {\footnotesize $s_0$};
			\node[inner sep=1pt] (m) at ([xshift=15mm,yshift=9mm]c) {\tiny $\beta_1$};
			\draw[->, thick] (m) -- (n); 
			\node[anchor=east] at ([xshift=13mm,yshift=-3mm]c)  {$\pi_2$:};
			\node[draw, minimum height=5mm,minimum width=5mm] (n) at ([xshift=15mm,yshift=-3mm]c)  {\footnotesize$s_0$};
			\node[inner sep=1pt] (m) at ([xshift=15mm,yshift=-9mm]c) {\tiny $\beta_2$};
			\draw[->, thick] (m) -- (n); 
			
			\node[] at ([xshift=3mm,yshift=4mm]c) {$v_1$};
			\node[] at ([xshift=3mm,yshift=-4mm]c) {$\fstage$};
			\draw[-, dashed,thick] (c) --  ([xshift=6mm]c);
			\draw[-, dashed,thick] ([xshift=6mm,yshift=11mm]c) --  ([xshift=6mm,yshift=-11mm]c);
			%%%%%%%%%%%%%%%%%%%%%%%%%%%%%%%%

			%%%%%%%%%%%%%%%%%%%%%%%%%%%%%%%%
			\coordinate (c) at (5.4, 0);
			\node[draw, minimum height=22mm,minimum width=24mm, rounded corners=1mm,line width=2pt,anchor=west] at (c) (b2) {};
			\node[anchor=east] at ([xshift=13mm,yshift=3mm]c)  {$\pi_1$:};
			\node[draw, minimum height=5mm,minimum width=5mm] (n) at ([xshift=15mm,yshift=3mm]c)  {\footnotesize $s_0$};
			\node[draw, minimum height=5mm,minimum width=5mm] () at ([xshift=20mm,yshift=3mm]c)  {\footnotesize $s_0$};
			\node[inner sep=1pt] (m) at ([xshift=15mm,yshift=9mm]c) {\tiny $\beta_1$};
			\draw[->, thick] (m) -- (n); 
			\node[anchor=east] at ([xshift=13mm,yshift=-3mm]c)  {$\pi_2$:};
			\node[draw, minimum height=5mm,minimum width=5mm] (n) at ([xshift=15mm,yshift=-3mm]c)  {\footnotesize$s_0$};
			\node[inner sep=1pt] (m) at ([xshift=15mm,yshift=-9mm]c) {\tiny $\beta_2$};
			\draw[->, thick] (m) -- (n); 
			
			\node[] at ([xshift=3mm,yshift=4mm]c) {$v_2$};
			\node[] at ([xshift=3mm,yshift=-4mm]c) {$\estage$};
			\draw[-, dashed,thick] (c) --  ([xshift=6mm]c);
			\draw[-, dashed,thick] ([xshift=6mm,yshift=11mm]c) --  ([xshift=6mm,yshift=-11mm]c);
			%%%%%%%%%%%%%%%%%%%%%%%%%%%%%%%%

			%%%%%%%%%%%%%%%%%%%%%%%%%%%%%%%%
			\coordinate (c) at (8.1, 0);
			\node[draw, minimum height=22mm,minimum width=24mm, rounded corners=1mm, thick,anchor=west] at (c) (b3) {};
			\node[anchor=east] at ([xshift=13mm,yshift=3mm]c)  {$\pi_1$:};
			\node[draw, minimum height=5mm,minimum width=5mm] () at ([xshift=15mm,yshift=3mm]c)  {\footnotesize $s_0$};
			\node[draw, minimum height=5mm,minimum width=5mm] (n) at ([xshift=20mm,yshift=3mm]c)  {\footnotesize $s_0$};
			\node[inner sep=1pt] (m) at ([xshift=20mm,yshift=9mm]c) {\tiny $\beta_1$};
			\draw[->, line width=2pt] (m) -- (n); 
			\node[anchor=east] at ([xshift=13mm,yshift=-3mm]c)  {$\pi_2$:};
			\node[draw, minimum height=5mm,minimum width=5mm] () at ([xshift=15mm,yshift=-3mm]c)  {\footnotesize$s_0$};
			\node[draw, minimum height=5mm,minimum width=5mm] (n) at ([xshift=20mm,yshift=-3mm]c)  {\footnotesize$s_1$};
			\node[inner sep=1pt] (m) at ([xshift=20mm,yshift=-9mm]c) {\tiny $\beta_2$};
			\draw[->, line width=2pt] (m) -- (n); 
			
			\node[] at ([xshift=3mm,yshift=4mm]c) {$v_3$};
			\node[] at ([xshift=3mm,yshift=-4mm]c) {$\ustage$};
			\draw[-, dashed,thick] (c) --  ([xshift=6mm]c);
			\draw[-, dashed,thick] ([xshift=6mm,yshift=11mm]c) --  ([xshift=6mm,yshift=-11mm]c);
			%%%%%%%%%%%%%%%%%%%%%%%%%%%%%%%%

			%%%%%%%%%%%%%%%%%%%%%%%%%%%%%%%%
			\coordinate (c) at (10.8, 0);
			\node[draw, minimum height=22mm,minimum width=24mm, rounded corners=1mm, thick,anchor=west] at (c) (b4) {};
			\node[anchor=east] at ([xshift=13mm,yshift=3mm]c)  {$\pi_1$:};
			\node[draw, minimum height=5mm,minimum width=5mm] (n) at ([xshift=15mm,yshift=3mm]c)  {\footnotesize $s_0$};
			\node[inner sep=1pt] (m) at ([xshift=15mm,yshift=9mm]c) {\tiny $\beta_1$};
			\draw[->, thick] (m) -- (n); 
			\node[anchor=east] at ([xshift=13mm,yshift=-3mm]c)  {$\pi_2$:};
			\node[draw, minimum height=5mm,minimum width=5mm] (n) at ([xshift=15mm,yshift=-3mm]c)  {\footnotesize$s_1$};
			\node[inner sep=1pt] (m) at ([xshift=15mm,yshift=-9mm]c) {\tiny $\beta_2$};
			\draw[->, thick] (m) -- (n); 
			
			\node[] at ([xshift=3mm,yshift=4mm]c) {$v_4$};
			\node[] at ([xshift=3mm,yshift=-4mm]c) {$\fstage$};
			\draw[-, dashed,thick] (c) --  ([xshift=6mm]c);
			\draw[-, dashed,thick] ([xshift=6mm,yshift=11mm]c) --  ([xshift=6mm,yshift=-11mm]c);
			%%%%%%%%%%%%%%%%%%%%%%%%%%%%%%%%

			%%%%%%%%%%%%%%%%%%%%%%%%%%%%%%%%
			\coordinate (c) at (13.5, 0);
			\node[draw, minimum height=22mm,minimum width=24mm, rounded corners=1mm, line width=2pt,anchor=west] at (c) (b5) {};
			\node[anchor=east] at ([xshift=13mm,yshift=3mm]c)  {$\pi_1$:};
			\node[draw, minimum height=5mm,minimum width=5mm] (n) at ([xshift=15mm,yshift=3mm]c)  {\footnotesize $s_0$};
			\node[draw, minimum height=5mm,minimum width=5mm] () at ([xshift=20mm,yshift=3mm]c)  {\footnotesize $s_0$};
			\node[inner sep=1pt] (m) at ([xshift=15mm,yshift=9mm]c) {\tiny $\beta_1$};
			\draw[->, thick] (m) -- (n); 
			\node[anchor=east] at ([xshift=13mm,yshift=-3mm]c)  {$\pi_2$:};
			\node[draw, minimum height=5mm,minimum width=5mm] (n) at ([xshift=15mm,yshift=-3mm]c)  {\footnotesize$s_1$};
			\node[inner sep=1pt] (m) at ([xshift=15mm,yshift=-9mm]c) {\tiny $\beta_2$};
			\draw[->, thick] (m) -- (n); 
			
			\node[] at ([xshift=3mm,yshift=4mm]c) {$v_5$};
			\node[] at ([xshift=3mm,yshift=-4mm]c) {$\estage$};
			\draw[-, dashed,thick] (c) --  ([xshift=6mm]c);
			\draw[-, dashed,thick] ([xshift=6mm,yshift=11mm]c) --  ([xshift=6mm,yshift=-11mm]c);
			%%%%%%%%%%%%%%%%%%%%%%%%%%%%%%%%

			%%%% New line
			
			%%%%%%%%%%%%%%%%%%%%%%%%%%%%%%%%
			\coordinate (c) at (0, -2.4);
			\node[draw, minimum height=22mm,minimum width=24mm, rounded corners=1mm, thick,anchor=west] at (c) (b6) {};
			\node[anchor=east] at ([xshift=13mm,yshift=3mm]c)  {$\pi_1$:};
			\node[draw, minimum height=5mm,minimum width=5mm] (n) at ([xshift=15mm,yshift=3mm]c)  {\footnotesize $s_1$};
			\node[inner sep=1pt] (m) at ([xshift=15mm,yshift=9mm]c) {\tiny $\beta_1$};
			\draw[->, thick] (m) -- (n); 
			\node[anchor=east] at ([xshift=13mm,yshift=-3mm]c)  {$\pi_2$:};
			\node[draw, minimum height=5mm,minimum width=5mm] (n) at ([xshift=15mm,yshift=-3mm]c)  {\footnotesize$s_0$};
			\node[inner sep=1pt] (m) at ([xshift=15mm,yshift=-9mm]c) {\tiny $\beta_2$};
			\draw[->, thick] (m) -- (n); 
			
			\node[] at ([xshift=3mm,yshift=4mm]c) {$v_{6}$};
			\node[] at ([xshift=3mm,yshift=-4mm]c) {$\fstage$};
			\draw[-, dashed,thick] (c) --  ([xshift=6mm]c);
			\draw[-, dashed,thick] ([xshift=6mm,yshift=11mm]c) --  ([xshift=6mm,yshift=-11mm]c);
			%%%%%%%%%%%%%%%%%%%%%%%%%%%%%%%%
			
			%%%%%%%%%%%%%%%%%%%%%%%%%%%%%%%%
			\coordinate (c) at (2.7, -2.4);
			\node[draw, minimum height=22mm,minimum width=24mm, rounded corners=1mm, thick,anchor=west] at (c) (b7) {};
			\node[anchor=east] at ([xshift=13mm,yshift=3mm]c)  {$\pi_1$:};
			\node[draw, minimum height=5mm,minimum width=5mm] () at ([xshift=15mm,yshift=3mm]c)  {\footnotesize $s_0$};
			\node[draw, minimum height=5mm,minimum width=5mm] (n) at ([xshift=20mm,yshift=3mm]c)  {\footnotesize $s_1$};
			\node[inner sep=1pt] (m) at ([xshift=20mm,yshift=9mm]c) {\tiny $\beta_1$};
			\draw[->, line width=2pt] (m) -- (n); 
			\node[anchor=east] at ([xshift=13mm,yshift=-3mm]c)  {$\pi_2$:};
			\node[draw, minimum height=5mm,minimum width=5mm] () at ([xshift=15mm,yshift=-3mm]c)  {\footnotesize$s_0$};
			\node[draw, minimum height=5mm,minimum width=5mm] (n) at ([xshift=20mm,yshift=-3mm]c)  {\footnotesize$s_0$};
			\node[inner sep=1pt] (m) at ([xshift=20mm,yshift=-9mm]c) {\tiny $\beta_2$};
			\draw[->, line width=2pt] (m) -- (n); 
			
			\node[] at ([xshift=3mm,yshift=4mm]c) {$v_{7}$};
			\node[] at ([xshift=3mm,yshift=-4mm]c) {$\ustage$};
			\draw[-, dashed,thick] (c) --  ([xshift=6mm]c);
			\draw[-, dashed,thick] ([xshift=6mm,yshift=11mm]c) --  ([xshift=6mm,yshift=-11mm]c);
			%%%%%%%%%%%%%%%%%%%%%%%%%%%%%%%%

			%%%%%%%%%%%%%%%%%%%%%%%%%%%%%%%%
			\coordinate (c) at (5.4, -2.4);
			\node[draw, minimum height=22mm,minimum width=24mm, rounded corners=1mm, line width=2pt,anchor=west] at (c) (b8) {};
			\node[anchor=east] at ([xshift=13mm,yshift=3mm]c)  {$\pi_1$:};
			\node[draw, minimum height=5mm,minimum width=5mm] (n) at ([xshift=15mm,yshift=3mm]c)  {\footnotesize $s_0$};
			\node[draw, minimum height=5mm,minimum width=5mm] () at ([xshift=20mm,yshift=3mm]c)  {\footnotesize $s_1$};
			\node[inner sep=1pt] (m) at ([xshift=15mm,yshift=9mm]c) {\tiny $\beta_1$};
			\draw[->, thick] (m) -- (n); 
			\node[anchor=east] at ([xshift=13mm,yshift=-3mm]c)  {$\pi_2$:};
			\node[draw, minimum height=5mm,minimum width=5mm] (n) at ([xshift=15mm,yshift=-3mm]c)  {\footnotesize$s_0$};
			\node[inner sep=1pt] (m) at ([xshift=15mm,yshift=-9mm]c) {\tiny $\beta_2$};
			
			\draw[->, thick] (m) -- (n); 
			
			\node[] at ([xshift=3mm,yshift=4mm]c) {$v_{8}$};
			\node[] at ([xshift=3mm,yshift=-4mm]c) {$\estage$};
			\draw[-, dashed,thick] (c) --  ([xshift=6mm]c);
			\draw[-, dashed,thick] ([xshift=6mm,yshift=11mm]c) --  ([xshift=6mm,yshift=-11mm]c);
			%%%%%%%%%%%%%%%%%%%%%%%%%%%%%%%%

			%%%%%%%%%%%%%%%%%%%%%%%%%%%%%%%%
			\coordinate (c) at (10.8, -2.4);
			\node[draw, minimum height=22mm,minimum width=24mm, rounded corners=1mm, line width=2pt,anchor=west] at (c) (b9) {};
			\node[anchor=east] at ([xshift=13mm,yshift=3mm]c)  {$\pi_1$:};
			\node[draw, minimum height=5mm,minimum width=5mm] (n) at ([xshift=15mm,yshift=3mm]c)  {\footnotesize $s_0$};
			\node[draw, minimum height=5mm,minimum width=5mm] () at ([xshift=20mm,yshift=3mm]c)  {\footnotesize $s_1$};
			\node[inner sep=1pt] (m) at ([xshift=15mm,yshift=9mm]c) {\tiny $\beta_1$};
			\draw[->, thick] (m) -- (n); 
			\node[anchor=east] at ([xshift=13mm,yshift=-3mm]c)  {$\pi_2$:};
			\node[draw, minimum height=5mm,minimum width=5mm] (n) at ([xshift=15mm,yshift=-3mm]c)  {\footnotesize$s_1$};
			\node[inner sep=1pt] (m) at ([xshift=15mm,yshift=-9mm]c) {\tiny $\beta_2$};
			\draw[->, thick] (m) -- (n); 
			
			\node[] at ([xshift=3mm,yshift=4mm]c) {$v_9$};
			\node[] at ([xshift=3mm,yshift=-4mm]c) {$\estage$};
			\draw[-, dashed,thick] (c) --  ([xshift=6mm]c);
			\draw[-, dashed,thick] ([xshift=6mm,yshift=11mm]c) --  ([xshift=6mm,yshift=-11mm]c);
			%%%%%%%%%%%%%%%%%%%%%%%%%%%%%%%%

			%%%%%%%%%%%%%%%%%%%%%%%%%%%%%%%%
			\coordinate (c) at (13.5, -2.4);
			\node[draw, minimum height=22mm,minimum width=24mm, rounded corners=1mm, thick,anchor=west] at (c) (b10) {};
			\node[anchor=east] at ([xshift=13mm,yshift=3mm]c)  {$\pi_1$:};
			\node[draw, minimum height=5mm,minimum width=5mm] () at ([xshift=15mm,yshift=3mm]c)  {\footnotesize $s_0$};
			\node[draw, minimum height=5mm,minimum width=5mm] (n) at ([xshift=20mm,yshift=3mm]c)  {\footnotesize $s_0$};
			\node[inner sep=1pt] (m) at ([xshift=20mm,yshift=9mm]c) {\tiny $\beta_1$};
			\draw[->, line width=2pt] (m) -- (n); 
			\node[anchor=east] at ([xshift=13mm,yshift=-3mm]c)  {$\pi_2$:};
			\node[draw, minimum height=5mm,minimum width=5mm] () at ([xshift=15mm,yshift=-3mm]c)  {\footnotesize$s_1$};
			\node[draw, minimum height=5mm,minimum width=5mm] (n) at ([xshift=20mm,yshift=-3mm]c)  {\footnotesize$s_2$};
			\node[inner sep=1pt] (m) at ([xshift=20mm,yshift=-9mm]c) {\tiny $\beta_2$};
			\draw[->, line width=2pt] (m) -- (n); 
			
			\node[] at ([xshift=3mm,yshift=4mm]c) {$v_{10}$};
			\node[] at ([xshift=3mm,yshift=-4mm]c) {$\ustage$};
			\draw[-, dashed,thick] (c) --  ([xshift=6mm]c);
			\draw[-, dashed,thick] ([xshift=6mm,yshift=11mm]c) --  ([xshift=6mm,yshift=-11mm]c);
			%%%%%%%%%%%%%%%%%%%%%%%%%%%%%%%%

			%%% New line 
			
			%%%%%%%%%%%%%%%%%%%%%%%%%%%%%%%%
			\coordinate (c) at (0, -4.8);
			\node[draw, minimum height=22mm,minimum width=24mm, rounded corners=1mm, line width=2pt,anchor=west] at (c) (b11) {};
			\node[anchor=east] at ([xshift=13mm,yshift=3mm]c)  {$\pi_1$:};
			\node[draw, minimum height=5mm,minimum width=5mm] (n) at ([xshift=15mm,yshift=3mm]c)  {\footnotesize $s_1$};
			\node[draw, minimum height=5mm,minimum width=5mm] () at ([xshift=20mm,yshift=3mm]c)  {\footnotesize $s_2$};
			\node[inner sep=1pt] (m) at ([xshift=15mm,yshift=9mm]c) {\tiny $\beta_1$};
			\draw[->, thick] (m) -- (n); 
			\node[anchor=east] at ([xshift=13mm,yshift=-3mm]c)  {$\pi_2$:};
			\node[draw, minimum height=5mm,minimum width=5mm] (n) at ([xshift=15mm,yshift=-3mm]c)  {\footnotesize$s_0$};
			\node[inner sep=1pt] (m) at ([xshift=15mm,yshift=-9mm]c) {\tiny $\beta_2$};
			\draw[->, thick] (m) -- (n); 
			
			\node[] at ([xshift=3mm,yshift=4mm]c) {$v_{11}$};
			\node[] at ([xshift=3mm,yshift=-4mm]c) {$\estage$};
			\draw[-, dashed,thick] (c) --  ([xshift=6mm]c);
			\draw[-, dashed,thick] ([xshift=6mm,yshift=11mm]c) --  ([xshift=6mm,yshift=-11mm]c);
			%%%%%%%%%%%%%%%%%%%%%%%%%%%%%%%%

			%%%%%%%%%%%%%%%%%%%%%%%%%%%%%%%%
			\coordinate (c) at (2.7, -4.8);
			\node[draw, minimum height=22mm,minimum width=24mm, rounded corners=1mm, thick,anchor=west] at (c) (b12) {};
			\node[anchor=east] at ([xshift=13mm,yshift=3mm]c)  {$\pi_1$:};
			\node[draw, minimum height=5mm,minimum width=5mm] () at ([xshift=15mm,yshift=3mm]c)  {\footnotesize $s_1$};
			\node[draw, minimum height=5mm,minimum width=5mm] (n) at ([xshift=20mm,yshift=3mm]c)  {\footnotesize $s_2$};
			\node[inner sep=1pt] (m) at ([xshift=20mm,yshift=9mm]c) {\tiny $\beta_1$};
			\draw[->, line width=2pt] (m) -- (n); 
			\node[anchor=east] at ([xshift=13mm,yshift=-3mm]c)  {$\pi_2$:};
			\node[draw, minimum height=5mm,minimum width=5mm] () at ([xshift=15mm,yshift=-3mm]c)  {\footnotesize$s_0$};
			\node[draw, minimum height=5mm,minimum width=5mm] (n) at ([xshift=20mm,yshift=-3mm]c)  {\footnotesize$s_0$};
			\node[inner sep=1pt] (m) at ([xshift=20mm,yshift=-9mm]c) {\tiny $\beta_2$};
			\draw[->, line width=2pt] (m) -- (n); 
			
			\node[] at ([xshift=3mm,yshift=4mm]c) {$v_{12}$};
			\node[] at ([xshift=3mm,yshift=-4mm]c) {$\ustage$};
			\draw[-, dashed,thick] (c) --  ([xshift=6mm]c);
			\draw[-, dashed,thick] ([xshift=6mm,yshift=11mm]c) --  ([xshift=6mm,yshift=-11mm]c);
			%%%%%%%%%%%%%%%%%%%%%%%%%%%%%%%%

			%%%%%%%%%%%%%%%%%%%%%%%%%%%%%%%%
			\coordinate (c) at (5.4, -4.8);
			\node[draw, minimum height=22mm,minimum width=24mm, rounded corners=1mm, thick,anchor=west] at (c) (b13) {};
			\node[anchor=east] at ([xshift=13mm,yshift=3mm]c)  {$\pi_1$:};
			\node[draw, minimum height=5mm,minimum width=5mm] (n) at ([xshift=15mm,yshift=3mm]c)  {\footnotesize $s_2$};
			\node[inner sep=1pt] (m) at ([xshift=15mm,yshift=9mm]c) {\tiny $\beta_1$};
			\draw[->, thick] (m) -- (n); 
			\node[anchor=east] at ([xshift=13mm,yshift=-3mm]c)  {$\pi_2$:};
			\node[draw, minimum height=5mm,minimum width=5mm] (n) at ([xshift=15mm,yshift=-3mm]c)  {\footnotesize$s_0$};
			\node[inner sep=1pt] (m) at ([xshift=15mm,yshift=-9mm]c) {\tiny $\beta_2$};
			\draw[->, thick] (m) -- (n); 
			
			\node[] at ([xshift=3mm,yshift=4mm]c) {$v_{13}$};
			\node[] at ([xshift=3mm,yshift=-4mm]c) {$\fstage$};
			\draw[-, dashed,thick] (c) --  ([xshift=6mm]c);
			\draw[-, dashed,thick] ([xshift=6mm,yshift=11mm]c) --  ([xshift=6mm,yshift=-11mm]c);
			%%%%%%%%%%%%%%%%%%%%%%%%%%%%%%%%

			%%%%%%%%%%%%%%%%%%%%%%%%%%%%%%%%
			\coordinate (c) at (10.8, -4.8);
			\node[draw, minimum height=22mm,minimum width=24mm, rounded corners=1mm, thick,anchor=west] at (c) (b14) {};
			\node[anchor=east] at ([xshift=13mm,yshift=3mm]c)  {$\pi_1$:};
			\node[draw, minimum height=5mm,minimum width=5mm] (n) at ([xshift=15mm,yshift=3mm]c)  {\footnotesize $s_0$};
			\node[draw, minimum height=5mm,minimum width=5mm] () at ([xshift=20mm,yshift=3mm]c)  {\footnotesize $s_1$};
			\node[inner sep=1pt] (m) at ([xshift=15mm,yshift=9mm]c) {\tiny $\beta_1$};
			\draw[->, thick] (m) -- (n); 
			\node[anchor=east] at ([xshift=13mm,yshift=-3mm]c)  {$\pi_2$:};
			\node[draw, minimum height=5mm,minimum width=5mm] () at ([xshift=15mm,yshift=-3mm]c)  {\footnotesize$s_1$};
			\node[draw, minimum height=5mm,minimum width=5mm] (n) at ([xshift=20mm,yshift=-3mm]c)  {\footnotesize$s_2$};
			\node[inner sep=1pt] (m) at ([xshift=20mm,yshift=-9mm]c) {\tiny $\beta_2$};
			\draw[->, line width=2pt] (m) -- (n); 
			
			\node[] at ([xshift=3mm,yshift=4mm]c) {$v_{14}$};
			\node[] at ([xshift=3mm,yshift=-4mm]c) {$\ustage$};
			\draw[-, dashed,thick] (c) --  ([xshift=6mm]c);
			\draw[-, dashed,thick] ([xshift=6mm,yshift=11mm]c) --  ([xshift=6mm,yshift=-11mm]c);
			%%%%%%%%%%%%%%%%%%%%%%%%%%%%%%%%

			%%%%%%%%%%%%%%%%%%%%%%%%%%%%%%%%
			\coordinate (c) at (13.5, -4.8);
			\node[draw, minimum height=22mm,minimum width=24mm, rounded corners=1mm, thick,anchor=west] at (c) (b15) {};
			\node[anchor=east] at ([xshift=13mm,yshift=3mm]c)  {$\pi_1$:};
			\node[draw, minimum height=5mm,minimum width=5mm] (n) at ([xshift=15mm,yshift=3mm]c)  {\footnotesize $s_0$};
			\node[inner sep=1pt] (m) at ([xshift=15mm,yshift=9mm]c) {\tiny $\beta_1$};
			\draw[->, thick] (m) -- (n); 
			\node[anchor=east] at ([xshift=13mm,yshift=-3mm]c)  {$\pi_2$:};
			\node[draw, minimum height=5mm,minimum width=5mm] (n) at ([xshift=15mm,yshift=-3mm]c)  {\footnotesize$s_2$};
			\node[inner sep=1pt] (m) at ([xshift=15mm,yshift=-9mm]c) {\tiny $\beta_2$};
			\draw[->, thick] (m) -- (n); 
			
			\node[] at ([xshift=3mm,yshift=4mm]c) {$v_{15}$};
			\node[] at ([xshift=3mm,yshift=-4mm]c) {$\fstage$};
			\draw[-, dashed,thick] (c) --  ([xshift=6mm]c);
			\draw[-, dashed,thick] ([xshift=6mm,yshift=11mm]c) --  ([xshift=6mm,yshift=-11mm]c);
			%%%%%%%%%%%%%%%%%%%%%%%%%%%%%%%%
			
			\draw[->, thick] ($(b0.west) + (-0.4, 0)$) -- (b0.west);

			\draw[thick, ->] (b0) -- (b1);
			
			\draw[thick, ->] (b1) -- (b2);
			\draw[thick, ->] (b1) -- (b8);

			\draw[thick, ->] (b2) -- (b3);
			\draw[thick, ->] (b3) -- (b4);
			
			\draw[thick, ->] (b4) -- (b5);
			\draw[thick, ->] (b4) -- (b9);
			
			\draw[thick, ->] (b5) -- (b10);
			
			\draw[thick, ->] (b6) -- (b11);
			
			\draw[thick, ->] (b7) -- (b6);
			
			\draw[thick, ->] (b8) -- (b7);
			
			\draw[thick, ->] (b9) -- (b14);
			
			\draw[thick, ->] (b10) -- (b15);
			
			\draw[thick, ->] (b11) -- (b12);
			
			\draw[thick, ->] (b12) -- (b13);
			
			\draw[thick, ->] (b14) -- (b15);

			\coordinate (c1) at (6.2, -6.3);
			\coordinate (c2) at (7, -6.3);
			
			\draw[thick, ->, dashed] (b13) -- (c1);
			\draw[thick, ->, dashed] (b13) -- (c2);

			\coordinate (c1) at (14.3, -6.3);
			\coordinate (c2) at (15.1, -6.3);
			
			\draw[thick, ->, dashed] (b15) -- (c1);
			\draw[thick, ->, dashed] (b15) -- (c2);

		\end{tikzpicture}
	}
	
	\vspace{-2mm}
	
	\caption{We depict (parts of) a winning strategy for the verifier in the verification game for the system in \Cref{fig:system-asyn}.
		Each vertex maps $\pi_1, \pi_2$ to some window of states and tracks the current stage in $\{\ustage, \fstage, \estage\}$.
		We label the vertices $v_0, \cdots, v_{15}$.
		The verifier controls all vertices in the $\estage$-stage, marked by a thick border.
	}\label{fig:strat-asyn}
\end{figure}
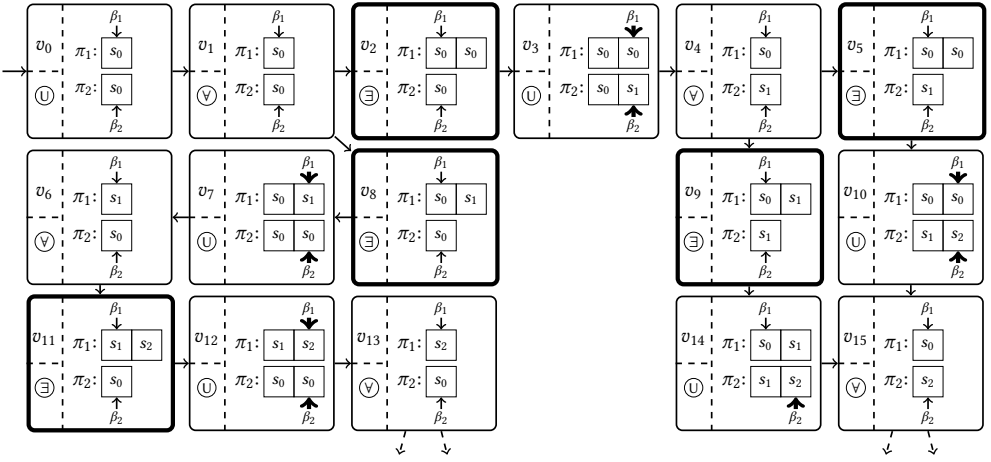

\paragraph{Winning Strategy.}
In our example, the system $\calT$ consists of finitely many states, so the resulting graph-based game contains only finitely many vertices.
\Cref{fig:strat-asyn} depicts a fragment of a winning strategy for the verifier.\footnote{A (positional) strategy for the verifier is a function that, for each vertex controlled by the verifier, fixes a concrete successor vertex (see \Cref{sec:prelim}).
	We visualize a strategy as a restriction of the game graph.
	That is, we include all possible successor vertices for all refuter-controlled vertices, but fix one particular outgoing edge for all verifier-controlled vertices.
} 
Each vertex maps $\pi_1$ and $\pi_2$ to a non-empty window (sequence) of states, the stutterings $\beta_1, \beta_2$ are represented as arrows.
Our game begins in initial vertex $v_0$, where $\pi_1$ and $\pi_2$ map to the singleton window $[s_0]$ containing $\calT$'s initial state, and the $\beta_1$ and $\beta_2$ pointers point to the unique first position.
In $v_1$ (in the $\fstage$-stage), the refuter can now choose a successor state for $\pi_1$ by either appending  state $s_0$ (i.e., move to $v_2$) or state $s_1$ (i.e., move to $v_{8}$) to $\pi_1$'s state window.
We focus on the former case. 
In $v_2$ (which is in the $\estage$-stage), the verifier can append a state to $\pi_2$'s window \emph{and} decide if the $\beta_1$ and $\beta_2$ pointers should be progressed.
The winning strategy depicted here moves to $v_3$, i.e., appends state $s_1$ and advances both $\beta_1$ and $\beta_2$.
In $v_3$ (which is in the $\ustage$-stage), $\beta_1$ points to state $s_0$, and $\beta_2$ points to $s_1$, which differ in the evaluation of $a$, as required by $\ltlN \ltlG (a_{\beta_1} \neq a_{\beta_2})$.
From $v_3$, we (deterministically) move to $v_4$ by removing the first position of the state windows for $\pi_1$ and $\pi_2$:
The $\beta_1$ and $\beta_2$ pointers (which give the current position of the stuttering) point to the second position in each state window, so we can safely drop all positions preceding the current position of the stuttering.
In $v_4$, the refuter can again append state $s_0$ (i.e., move to $v_5$) or $s_1$ (i.e., move to $v_9$) to $\pi_1$'s state window.
We focus on the latter case. 
In $v_9$, the current window for $\pi_2$ ends in state $s_1$, so the verifier has no choice but to append $s_2$ (as it has to respect $\calT$'s transitions). 
A possible move for the verifier \emph{would} be to progress \emph{both} the $\beta_1$ and $\beta_2$ pointers, so they would point to states $s_1$ and $s_2$, respectively. 
However, this would lose the game for the verifier as, in the next $\ustage$-stage, $a_{\beta_1} \neq a_{\beta_2}$ is violated.
Instead, the verifier moves to vertex $v_{14}$, i.e., it only progresses the $\beta_2$ pointer, and leaves the $\beta_1$-pointer  unchanged, effectively stuttering $\beta_1$ (i.e., the $\beta_1$ pointer still points to state $s_0$, repeating the same state as in the previous round).
In $v_{14}$, the stutterings $\beta_1$ and $\beta_2$ thus point to states $s_0$ and $s_2$, respectively, satisfying $a_{\beta_1} \neq a_{\beta_2}$.
In $v_{14}$, we, again, (deterministically) remove the first position of the $\pi_2$ window. 
Additionally, we can shorten $\pi_1$'s window by removing all positions trailing the $\beta_1$-pointer, thereby ensuring that the state windows never grow to more than 2 states (i.e., we deterministically move to vertex $v_{15}$).

The initial sketch in \Cref{fig:strat-asyn} can be extended to a full strategy which \textbf{(1)} ensures that $\ltlN \ltlG (a_{\beta_1} \neq a_{\beta_2})$ holds (where we evaluate one step whenever in the $\ustage$-stage), and \textbf{(2)} advances the $\beta_1$ and $\beta_2$ pointers infinitely often (i.e., stutterings $\beta_1$ and $\beta_2$ are fair).
Notably, our game is designed such that it underapproximates the power of existential quantification: If the verifier can construct appropriate witnesses for $\pi_2, \beta_1, \beta_2$, no matter how the refuter constructs on $\pi_1$, then we can also find witnesses for $\pi_2, \beta_1, \beta_2$ when given the entire (infinite) trace $\pi_1$ (as in the semantics of A-HLTL); using the game we thus proved that $\calT$ satisfies \ref{eq:fair}.

We emphasize that, even though this example is deliberately simple (for demonstration purposes), our method is the first that can verify \ref{eq:fair} automatically:
The system in \Cref{fig:system-asyn} is clearly not terminating \cite{HsuBFS23}, and \ref{eq:fair} is not admissible (cf.~\Cref{sub:admissible}) \cite{BaumeisterCBFS21}.

\section{Preliminaries}\label{sec:prelim}

We write $\nat$ for the set of natural numbers and $\bool := \{\top,\bot\}$ for the set of Booleans. 
Given a set $X$, we write $X^\omega$ for the set of infinite sequences over $X$ and $X^{\leq n}$ for the set of finite sequences of length at most $n \in \nat$. 
For $u \in X^{\leq n}$, $i\in \nat$, and $x \in X$, we define: $|u|$ as the length of $u$, $u(i) \in X$ as the $i$th element in $u$ (starting with the $0$th), and $u \cdot x$ as the sequence where we append $x$ to $u$.
For $i, j \in \nat$, we define $u[i, j] := u(i) u(i+1) \cdots u(j)$ if $|u| > j$, and otherwise (if $|u| \leq j$) $u[i, j] := u(i) u(i+1) \cdots u(|u|-1)$.
Finally, we define $u[i, j) := u[i, j-1]$.
We denote unnamed functions using $\lambda$ notation: For example, $\lambda n\ldot n + 1$ denotes the function $\nat \to \nat$ that maps each $n \in \nat$ to $n + 1$.

\paragraph{First Order Theories.}

HyperLTL \cite{ClarksonFKMRS14} and A-HLTL \cite{BaumeisterCBFS21,HsuBFS23} are evaluated over transition systems and access the current state using finitely many \emph{atomic propositions} (APs), which we can think of as Boolean variables.
However, in many cases, the variables of a system are not Boolean but come from richer domains (e.g., the set of integers). 
We extend A-HLTL by including first-order formulas (modulo some fixed first-order theory) as atomic expressions.
That is, instead of accessing traces at the level of atomic propositions, we can reason about relational formulas over complex data types and use interpreted predicates like $=, \neq, \leq, \ldots$ \cite{BeutnerF25c}.
Formally, we assume some first-order signature and some fixed background theory $\theory$ with universe $\values$ (which we can think of as the set of all values). 
For a set of variables $\calX$, we write $\calF_\calX$ for the set of all first-order formulas over variables in $\calX$.
Given a formula $\theta \in \calF_\calX$ and a variable assignment $A : \calX \to \values$, we write $A \models^\theory \theta$ if $\theta$ is satisfied by $A$ (modulo $\theory$) \cite{barwise1977introduction}.

\paragraph{Transition Systems}

As the basic computation model, we use state-based transition systems. 
For this, let $\calX$ be a fixed finite set of system variables.
A \emph{transition system} (TS) is a tuple $\calT = (S, s_0, \kappa, \ell)$, where $S$ is a (possibly infinite) set of states, $s_0 \in S$ is an initial state, $\kappa : S \to (2^S \setminus \{\emptyset\})$ is a transition function, and $\ell : S \to (\calX \to \values)$ is a function that maps each state to a variable assignment over $\calX$.
A path in $\calT$ is an infinite sequence $p = s_0s_1s_2 \cdots \in S^\omega$ (starting in $s_0$) such that $s_{i+1} \in \kappa(s_i)$ for all $i \in \nat$.
We write $\paths{\calT} \subseteq S^\omega$ for the set of all paths in $\calT$. 
Given a path $p = s_0s_1s_2 \cdots$, the associated trace is given by applying $\ell$ pointwise, i.e., $\ell(p) := \ell(s_0)\ell(s_1)\ell(s_2) \cdots \in (\calX \to \values)^\omega$. 
We write $\traces{\calT} \subseteq (\calX \to \values)^\omega$ for the set of all traces generated by $\calT$.
Given two trace $t, t' \in (\calX \to \values)^\omega$, we say $t'$ is a \emph{fair stuttering of $t$} (written $t' \stutter t$) if $t'$ is obtained from $t$ by stuttering any position for an arbitrary (but \emph{finite}) number of steps.
Formally, $t' \stutter t$ iff there exists a function $f : \nat \to \nat$ such that \textbf{(1)} $\forall i \in \nat\ldot t'(i) = t\big(f(i)\big)$, \textbf{(2)} $f$ is monotonically increasing (i.e., $i \geq j$ implies $f(i) \geq f(j)$), and \textbf{(3)} $f$ is surjective (i.e., $\forall i \in \nat\ldot \exists j \in \nat\ldot f (j) = i$). 

\paragraph{Büchi Automata.}
During the game, we need to track if the traces and stutterings constructed by the verifier satisfy the LTL body of the A-HLTL formula.
To accomplish this, we translate the LTL formula to a \emph{deterministic Büchi automaton} (DBA).
A DBA is similar to a DFA over finite words but accepts \emph{infinite} words. 
We can thus run this automaton in parallel with the (infinite) behavior of refuter and verifier, and determine if their infinite play satisfies the LTL formula.\footnote{Note that not every LTL formula can be translated to a deterministic Büchi automaton. 
Instead, we can use slightly more complex deterministic \emph{parity} automata, which capture every $\omega$-regular property \cite{Piterman07}. 
Throughout this paper, we nevertheless use simpler Büchi automata to simplify the presentation.}
A (DBA) over alphabet $\Sigma$ is a tuple $\calA = (Q, q_0, \delta, F)$, where $Q$ is a finite set of states, $q_0 \in Q$ is an initial state, $\delta : Q \times \Sigma \to Q$ is a deterministic transition function, and $F \subseteq Q$ is a set of accepting states. 
Given a word $u \in \Sigma^\omega$, there exists a unique run $q_0q_1\cdots \in Q^\omega$ of $\calA$ (starting in $q_0$) defined by $q_{i+1} = \delta(q_i, u(i))$ for all $i \in \nat$. 
The run is accepting if it visits states in $F$ \emph{infinitely often}. 
We write $\calL(\calA) \subseteq \Sigma^\omega$ for the set of words where the unique run of $\calA$ is accepting. 

\paragraph{Büchi Games.}

As a concrete game formalism, we use graph-based games, where two players move a token from vertex to vertex along the edges of the graph.
Our later game will be designed such that the outgoing edges in the graph precisely capture the ability of either player in the current game situation (e.g., append a state to a state window, progress a stuttering, etc.). 
The objective of the verifier is to ensure that the infinite game play satisfies the LTL body of the A-HLTL formula.
As we track this LTL body using a Büchi automaton, we naturally obtain a Büchi winning objective for the verifier.
Formally, a \emph{Büchi game} is a tuple $\calG = (V_\verifier, V_\refuter, \alpha, F)$, where $V := V_\verifier \uplus V_\refuter$ is a (possibly infinite) set of vertices (vertices in $V_\verifier$ are controlled by the verifier, those in $V_\refuter$ by the refuter), $\alpha : V \to (2^V \setminus \{\emptyset\})$ is a transition function, and $F \subseteq V$ is a set of accepting vertices. 
A play in $\calG$ is an infinite sequence $\rho \in V^\omega$ such that $\rho(i+1) \in \alpha(\rho(i))$ for all $i \in \nat$. The play $\rho$ is won by $\verifier$ if it visits states in $F$ infinitely often.
A positional strategy (for the verifier) (cf.~\cite{martin1975borel}) is a function $\sigma : V_\verifier \to V$ such that $\sigma(v) \in \alpha(v)$ for all $v \in V_\verifier$.
A play $\rho \in V^\omega$ is compatible with $\sigma$ if for every $i \in \nat$ with $\rho(i) \in V_\verifier$, we have that $\rho(i+1) = \sigma(\rho(i))$. 
As we have done in \Cref{sec:overview}, we can think of a strategy as a restriction of the game's state space: we include all outgoing edges for all vertices controlled by the refuter; for each verifier-controlled vertex $v \in V_\verifier$ we only include the edge to vertex $\sigma(v)$. 
The verifier wins vertex $v \in V$ if there exists a strategy $\sigma$ such that all $\sigma$-compatible plays $\rho$ (i.e., all infinite paths in the restricted vertex space) that start in $v$ (i.e., $\rho(0) = v$) are won by the verifier.  
Given a set of vertices $X \subseteq V$, we write $\mathit{wins}_\calG(\verifier, X)$ if the verifier wins $\calG$ from any vertex $v \in X$.

\section{Asynchronous HyperLTL}\label{sec:ahltl}

HyperLTL extends LTL \cite{Pnueli77} with explicit quantification over traces \cite{ClarksonFKMRS14}.
A-HLTL then extends the primary quantification over traces with secondary quantification over stutterings of these traces.
\begin{remark}
	In the original variant of A-HLTL \cite{BaumeisterCBFS21,HsuBFS23}, quantification over stutterings was achieved by quantifying over so-called \emph{trajectories}. For example, $\forall \pi_1. \forall \pi_2. \trajE \tau. \ltlG(o_{\pi_1, \tau} = o_{\pi_2, \tau})$ states that for all pairs of traces $\pi_1, \pi_2$, there exists some trajectory $\tau$ that aligns both traces on the output $o$. Here, the pairs $(\pi_1, \tau)$ and $(\pi_2, \tau)$ denote \emph{independent} stutterings of $\pi_1$ and $\pi_2$, respectively. 
	This detour via trajectories leads to a convoluted semantics, as each trajectory \emph{implicitly} quantifies over stutterings for \emph{all} traces, see \cite{BaumeisterCBFS21,HsuBFS23} for details. 
	In this paper, we propose a novel variant of A-HLTL by \emph{explicitly} quantifying over stutterings of traces.
	Our variant is easier to understand, allows for a simpler semantics, and is \emph{equivalent} to the original variant (see \ifFull{Appendix \ref{app:semantics}}{the full version}).
	The above example can be expressed in our new A-HLTL variant as $\forall \pi_1. \forall \pi_2. \exists \beta_1 \stutter \pi_1. \exists \beta_2 \stutter \pi_2. \ltlG(o_{\beta_1} = o_{\beta_2})$, i.e., for each trace-trajectory pair $(\pi_1, \tau)$, $(\pi_2, \tau)$, we \emph{explicitly} quantify over a stuttering of that trace. 
\end{remark}

\paragraph{Syntax}
Let $\traceVars = \{\pi_1, \pi_2, \ldots, \pi_n\}$ be a set of \emph{trace variables} and $\trajVars = \{\beta_1, \beta_2, \ldots, \beta_k\}$ be a set of \emph{stuttering variables}. 
We define $\calX_{\trajVars} := \{x_{\beta} \mid x \in \calX, \beta \in \trajVars\}$ as the set of system variables indexed by stuttering variables.
An A-HLTL formula $\varphi$ is generated by the following grammar
\begin{align*}
	\varphi &:= \exists \pi \ldot \varphi \mid \forall \pi \ldot \varphi \mid  \phi \\
	\phi &:= \exists \beta \stutter \pi \ldot \phi \mid \forall \beta \stutter \pi \ldot \phi \mid  \psi \\
	\psi &:= \theta \mid \neg \psi \mid \psi \land \psi \mid \ltlN \psi \mid \psi \ltlU \psi
\end{align*}%
where $\pi \in \traceVars$ is a trace variable, $\beta \in \trajVars$ is a stuttering variable, and $\theta \in \calF_{\calX_{\trajVars}}$ is a first-order formula over stuttering-variable-indexed system variables (i.e., $\calX_\trajVars$). 
We write $\ltlF \psi$ and $\ltlG \psi$ for LTL's derived \emph{eventually} and \emph{globally} operator, respectively.

\paragraph{Semantics}

An A-HLTL formula from the above grammar has the form 
\begin{align*}
	\varphi = \quant \pi_1 \ldots \quant \pi_n\ldot \; \quant \beta_1 \stutter \pi_{l_1} \ldots \quant \beta_k \stutter \pi_{l_l} \ldot \; \psi
\end{align*}
where $\pi_1, \ldots, \pi_n \in \traceVars$ are trace variables, $\beta_1,\ldots, \beta_k \in \trajVars$ are stuttering variables quantifying over stutterungs of traces $\pi_{l_1}, \ldots,\pi_{l_l}$,  respectively (where $l_1, \ldots, l_k \in \{1, \ldots, n\}$), $\quant \in \{\forall, \exists\}$ are quantifiers, and $\psi$ is an LTL formula over atoms from $\calF_{\calX_{\trajVars}}$ (i.e., first-order formulas over variables from $\calX_{\trajVars}$).
Similar to HyperLTL, we first quantify over traces $\pi_1, \ldots, \pi_n$ in our transition system $\calT$ in typical first-order semantics.
After the $n$ traces are fixed, we proceed to the secondary quantification over stutterings of these traces, i.e., $\quant \beta \stutter \pi$ quantifies over a fair stuttering of (the previously quantified) trace $\pi$.
After resolving the quantifier prefix, we are thus left with $k$ stutterings $\beta_1, \ldots, \beta_k$, and can evaluate the LTL formula $\psi$. 
Here, each atom in the LTL formula $\psi$ is a first-order formula over variables from $\calX_{\trajVars}$ (modulo $\theory$), i.e., over the system variables indexed by stuttering variables. 
In such a formula, each variable $x_{\beta} \in \calX_{\trajVars}$ refers to the current value of $x$ on stuttering $\beta$.
For example, in \ref{eq:asyn-od}, we used the atomic formula $o_{\beta_1} = o_{\beta_2}$, where ``$=$'' is an interpreted predicate symbol from theory $\theory$ and $o_{\beta_1}, o_{\beta_2} \in \calX_{\{\beta_1, \beta_2\}}$ are indexed system variables.

To define the semantics formally, we maintain a \emph{trace assignment} $\Pi : \traceVars \rightharpoonup (\calX \to \values)^\omega$ mapping trace variables to traces (used to evaluate trace quantification), and a \emph{stuttering assignment} $\Delta : \trajVars \rightharpoonup (\calX \to \values)^\omega$ mapping stuttering variables to traces (used to evaluate stuttering quantification).
Given a stuttering assignment $\Delta$ and position $i \in \nat$, we write $\Delta_{(i)} : \calX_{\trajVars} \to \values$ for the variable assignment to $\calX_\trajVars$ in the $i$th step, defined as $\Delta_{(i)}(x_{\beta}) := \Delta(\beta)(i)(x)$, i.e., the value of $x_\beta$ is defined as the value of $x$ in the $i$th step on stuttering $\Delta(\beta)$.
Given a set of traces $\traceSet \subseteq (\calX \to \values)^\omega$, and position $i \in \nat$, we define the semantics of A-HLTL as follows
\allowdisplaybreaks
\begin{align*}
	\Pi, \Delta, i &\models_\traceSet \exists \pi \ldot \varphi &\text{iff} \quad &\exists t \in \traceSet \ldot \Pi[\pi \mapsto t], \Delta, i \models_\traceSet  \varphi\\
	\Pi, \Delta, i &\models_\traceSet  \forall \pi \ldot \varphi &\text{iff} \quad &\forall t \in \traceSet \ldot \Pi[\pi \mapsto t], \Delta, i \models_\traceSet  \varphi\\[2mm]
	\Pi, \Delta, i &\models_\traceSet  \exists \beta \stutter \pi \ldot \phi &\text{iff} \quad  & \exists t' \in (\calX \to \values)^\omega \ldot t' \stutter \Pi(\pi) \land \Pi, \Delta[\beta \mapsto t'], i \models_\traceSet  \phi\\
	\Pi, \Delta, i &\models_\traceSet  \forall \beta \stutter \pi \ldot \phi &\text{iff} \quad  & \forall  t' \in (\calX \to \values)^\omega \ldot t' \stutter \Pi(\pi) \Rightarrow \Pi, \Delta[\beta \mapsto t'], i \models_\traceSet  \phi\\[2mm]
	\Pi, \Delta,i &\models_\traceSet  \theta &\text{iff} \quad   &\Delta_{(i)} \models^\theory \theta \\
	\Pi, \Delta,i &\models_\traceSet   \neg \psi &\text{iff} \quad & \Pi, \Delta,i \not\models_\traceSet   \psi \\
	\Pi, \Delta,i &\models_\traceSet   \psi_1 \land \psi_2 &\text{iff} \quad  &\Pi, \Delta,i \models_\traceSet  \psi_1 \text{ and }  \Pi, \Delta,i \models_\traceSet   \psi_2\\
	\Pi, \Delta,i &\models_\traceSet   \ltlN  \psi &\text{iff} \quad & \Pi, \Delta, i+1 \models_\traceSet  \psi \\
	\Pi, \Delta,i &\models_\traceSet   \psi_1 \ltlU \psi_2 &\text{iff} \quad & \exists j \geq i \ldot \Pi, \Delta, j\models_\traceSet   \psi_2 \text{ and } \forall i \leq k < j \ldot  \Pi, \Delta, k \models_\traceSet  \psi_1.
\end{align*}%
We first populate a trace assignment $\Pi$ by following the trace-quantifier prefix over traces in $\traceSet$ and adding traces to $\Pi$.
For each quantified stuttering $\quant \beta \stutter \pi$, we then quantify over a stuttering $t'$ of trace $\Pi(\pi)$ and add it to $\Delta$.
Note that $t'$ is an arbitrary trace obtained by stuttering $\Pi(\pi)$; it may not be a trace in $\traceSet$.
Finally, we can evaluate the LTL body on the traces in $\Delta$ following the usual evaluation of boolean and temporal operators. 
Here, an atomic formula $\theta$ holds in step $i$ if the variable assignment $\Delta_{(i)}$ (which assigns values to the indexed variables in $\calX_{\trajVars}$) satisfies $\theta$ (modulo theory $\theory$).
A transition system $\calT$ satisfies an A-HLTL property $\varphi$, written $\calT \models \varphi$, if $\emptyset, \emptyset, 0 \models_{\traces{\calT}} \varphi$, where $\emptyset$ denotes a trace/stuttering  assignment with an empty domain.

\begin{example}\label{ex:multiple-traj}
	In most A-HLTL formulas, we quantify (usually existentially) over \emph{one} stuttering for each trace variable.
	However, A-HLTL can also quantify over multiple stutterings of the same trace \cite{HsuBFS23}. 
	For example, in \ref{eq:asyn-od} from \Cref{sec:intro}, the low-security input is fixed \emph{initially} and never changes during evaluation (as in \citet{ZdancewicM03}'s original definition). 
	If the low-security input can change over time, the specification needs to take the infinite \emph{sequence} of low-security inputs into account.
	We can express such an extension as follows
	\begin{align}\label{eq:NIAE}
		\forall \pi_1. \forall \pi_2\ldot\; \forall \beta_1 \stutter \pi_1. \forall \beta_2 \stutter \pi_2\ldot \; \exists \beta_3 \stutter \pi_1. \exists \beta_4 \stutter \pi_2 \ldot  \; \ltlG (l_{\beta_1} = l_{\beta_2}) \rightarrow \ltlG (o_{\beta_3} = o_{\beta_4}).\tag{$\varphi_\mathit{NI}$}
	\end{align}
	This formula requires that there \emph{exist} stutterings $\beta_3, \beta_4$ of $\pi_1, \pi_2$ that align the output $o$ (similar to \ref{eq:asyn-od}), assuming the \emph{universally quantified} stutterings $\beta_1, \beta_2$ can align the low-security input $l$.
	Phrased differently, any two traces with stutter-equivalent input should produce stutter-equivalent output. 
	\demo
\end{example}

\begin{example}\label{ex:od-versions}
	Multiple stutterings are also needed to analyze traces w.r.t.~different speeds. 
	We illustrate this by examining the various variants of observational determinism found in the literature. 
	Given a low-security input $l \in \calX$ and two low-security outputs $a, b \in \calX$, the A-HLTL formula $$\forall \pi_1\ldot \forall \pi_2\ldot \exists \beta_1 \stutter \pi_1. \exists \beta_2 \stutter \pi_2 \ldot  (l_{\beta_1} = l_{\beta_2}) \rightarrow \ltlG \big(a_{\beta_1} = a_{\beta_2} \land b_{\beta_1} = b_{\beta_2}\big)$$
	requires that $a$ and $b$ are, \emph{together}, stutter-equivalent on all pairs of traces (similar to the definitions of \citet{Terauchi08,ZdancewicM03}).
	However, especially when considering distributed systems, the individual variables on a trace might stem from different distributed components with independent timing. 
	In this case, we can express $$\forall \pi_1\ldot \forall \pi_2\ldot \exists \beta_1 \stutter \pi_1, \beta_2 \stutter \pi_2,  \beta_3 \stutter \pi_1, \beta_4 \stutter \pi_2 \ldot  (l_{\beta_1} = l_{\beta_2}) \rightarrow \ltlG \big(a_{\beta_1} = a_{\beta_2} \land b_{\beta_3} = b_{\beta_4}\big),$$
	requiring that variables $a$ and $b$ are \emph{individually} stutter equivalent, i.e., we can choose separate stutterings for both outputs (similar to the definition of \citet{HuismanWS06,BartocciHNC23}).
	\demo
\end{example}

\begin{example}\label{ex:ni}
	\citet{McLean94}'s notion of non-inference states that for every trace $\pi_1$, there should exist some trace $\pi_2$ that has the same \emph{sequence} of low-security events but a fixed dummy high-security input. 
	Crucially, \citet{McLean94}'s definition only reasons about the (ordered) sequence of low-security events, i.e., it does \emph{not} reason about the absolute (synchronous) timesteps.
	In A-HLTL, we can directly express non-inference as 
	\begin{align*}
		\forall \pi_1\ldot \exists \pi_2\ldot \; \exists \beta_1 \stutter \pi_1, \beta_2 \stutter \pi_2\ldot \; \ltlG (\bigwedge_{x \in L} x_{\beta_1} = x_{\beta_2}) \land \ltlG (\bigwedge_{x \in H} x_{\beta_2} = \dagger).
	\end{align*}
	That is, for every trace $\pi_1$, there exists some trace $\pi_2$ such that some stutterings $\beta_1, \beta_2$ can align $\pi_1, \pi_2$ so that the low-security variables ($L \subseteq \calX$) are equal up to stuttering (so the \emph{sequence of low-security events is the same}). At the same time, the high-security inputs on $\pi_2$ ($H \subseteq \calX$) should globally be set to some dummy value (which we denote with $\dagger$). \demo
\end{example}

We emphasize that while stuttering in itself is not particularly interesting, it is the key technical gadget that \emph{unifies} many of the widely used information-flow properties in the \emph{same} logic (see also \Cref{sec:related-work}).
That is, by extending HyperLTL with the ability to quantify over stutterings of traces, we can suddenly express a much wider range of information-flow policies.
Any verification approach for A-HLTL can, therefore, be applied to all those properties (and variants thereof by, e.g., declaring some additional variable as output), instead of ``re-inventing the wheel'' for every new information-flow property (or variant thereof).

\section{Game-based Verification of A-HLTL}\label{sec:verification}

We say an A-HLTL formula is a \emph{$\forall^*\exists^*$ formula} if no universal trace and stuttering quantifier occurs in the scope of an existentially quantified trace or stuttering (see, e.g., \ref{eq:asyn-od} and \ref{eq:NIAE}).
The idea of our game-based verification approach is to approximate the quantifier alternation in an A-HLTL formula as a game between a verifier ($\verifier$) and a refuter ($\refuter$). 
The refuter controls universally quantified traces and stutterings, and the verifier responds to each move to the refuter by updating existentially quantified  traces and stutterings (cf.~\Cref{sec:overview}).
As all universal quantifiers precede any existential quantification, the verifier updates existentially quantified traces and stutterings based on a subset of the information available in the A-HLTL semantics: In the semantics of a $\forall^*\exists^*$ formula, we can pick existentially quantified traces and stutterings after \emph{all} universally quantified traces and stutterings are fixed.
In our game, the verifier picks them step-wise, thus knowing only a finite \emph{prefix}.
All formulas studied by \citet{BaumeisterCBFS21} and \citet{HsuBFS23} are $\forall^*\exists^*$ formulas.

In the following, we let $\calT = (S_\calT, s_{0, \calT}, \kappa_\calT, \ell_\calT)$ be a fixed system and let 
\begin{align*}
	\varphi = \quant \pi_1 \ldots \quant \pi_n\ldot \; \quant \beta_1 \stutter \pi_{l_1} \ldots \quant \beta_k \stutter \pi_{l_k} \ldot \; \psi
\end{align*}
be a fixed $\forall^*\exists^*$ A-HLTL formula (where $l_1, \ldots, l_k \in \{1, \ldots, n\}$).
We define $\traceVars_\forall \subseteq \{\pi_1, \ldots, \pi_n\}$ (resp.~$\traceVars_\exists$) as all universally (resp.~existentially) quantified trace variables, and $\trajVars_\forall \subseteq \{\beta_1, \ldots, \beta_k\}$ (resp.~$\trajVars_\exists$) as all universally (resp.~existentially) quantified stuttering variables. 
For each stuttering variable $\beta \in \trajVars$, we define $\mathit{base}(\beta) \in \traceVars$ as the unique trace that $\beta$ is based on (i.e., the unique $\pi$, where the quantifier prefix contains $\quant \beta \stutter \pi$). 
For example, in \ref{eq:NIAE}, $\mathit{base}(\beta_1) = \mathit{base}(\beta_3) = \pi_1$ and $\mathit{base}(\beta_2) = \mathit{base}(\beta_4) = \pi_2$.

In this section, we will construct a Büchi game -- called $\calG_{\calT, \varphi, Z}$ -- which, if won by the verifier, implies that $\calT \models \varphi$.
As we already discussed, the high-level idea of this game is to let the verifier (resp.~refuter) control all existentially (resp.~universally) quantified traces and stutterings.
To formalize this, we will cast this game as a graph-based game, where each vertex in the game records the current state of the game (e.g., the state windows, the positions of the stutterings, etc.).
The transitions of the game are then designed such that the set of possible successor vertices precisely corresponds to all possible moves that the players can make (e.g., append a state to a state window, progress a stuttering, etc.).

\subsection{Tracking Acceptance and Fairness}\label{sub:formula-transformation}

While the game progresses, we need to track if the paths and stutterings constructed by the players satisfy $\psi$ (the LTL body of $\varphi$).
We accomplish this by translating $\psi$ to a DBA.
Additionally, we need to ensure that all stutterings are fair, i.e., progress \emph{infinitely often}.
To ensure this, we consider the following modified LTL formula
\begin{align*}
	\psi_\mathit{mod} := \Big(\bigwedge_{\beta \in \trajVars_\forall} \ltlG \ltlF \moved_{\beta}\Big) \to \Big( \big(\bigwedge_{\beta \in \trajVars_\exists} \ltlG \ltlF \moved_{\beta}\big) \land \psi  \Big),
\end{align*}
which includes fresh Boolean variables $\moved_{\beta}$ for every $\beta \in \trajVars$. 
The intuition is that $\moved_{\beta}$ holds whenever the stuttering $\beta$ has performed a non-stuttering step in the last round.
The LTL formula $\psi_\mathit{mod}$ then requires that existentially quantified stutterings should progress infinitely often (expressed using the LTL operator combination $\ltlG\ltlF$), and $\psi$ must hold, provided that all universally quantified stutterings (which will be in control of the refuter) are fair. 

Let $\Theta \subseteq \calF_{\calX_{\trajVars}}$ be the \emph{finite} set of first-order formulas used in $\psi$.
For example, in \ref{eq:asyn-od}, $\Theta = \big\{ l_{\beta_1} = l_{\beta_2}, o_{\beta_1} = o_{\beta_2}  \big\}$.
In the following, we assume that $\calA_\psi = (Q_\psi, q_{0, \psi}, \delta_\psi, F_\psi)$ is a DBA over alphabet $2^{\Theta \cup \{\moved_{\beta} \mid \beta \in \trajVars\}}$ accepting exactly those infinite words that satisfy $\psi_\mathit{mod}$.
That is, in each step, we need to tell $\calA_\psi$'s transition function which of the formulas in $\Theta$ hold and which of the stutterings have been progressed (via the $\moved_{\beta}$ variables), and $\calA_\psi$ tracks if this behavior satisfies $\psi_\mathit{mod}$.

\subsection{Game Vertices}

In our asynchronous setting, different stutterings can point to (different positions on) the same trace.
As outlined in \Cref{sec:overview}, we accommodate this by maintaining a \emph{window} of states, i.e., for each trace variable, we track a finite sequence of successive states in $\calT$, akin to a sliding window.
Each stuttering variable $\beta$ is then a (relative) pointer to the state window assigned to trace $\mathit{base}(\beta)$, i.e., the trace that $\beta$ is a stuttering of.
We parameterize our game with a bound $Z \in \nat$, determining the maximal length of each state window.
If we only deal with a single stuttering per trace (as is the case in most properties), a window of size $Z=1$ generally suffices.
If we consider properties with multiple stutterings on the same trace, the bound naturally defines a trade-off between the size of the game (a larger $Z$ generates more vertices) and the flexibility of stutterings (a larger $Z$ allows multiple stutterings of the same trace to diverge further).
In many cases, we can statically infer a bound $Z$, while maintaining completeness (cf.~\Cref{sec:complete}).

We can now define our Büchi game $\calG_{\calT, \varphi, Z}$ used to verify that $\calT \models \varphi$.
Each regular game vertex in $\calG_{\calT, \varphi, Z}$ has the form $\langle \stage, \Xi, \mu, \flat, q \rangle$, where \textbf{(1)} $\stage \in \{\fstage, \estage, \ustage\}$ tracks the current stage of the game; \textbf{(2)} $\Xi : \traceVars \to (S_\calT)^{\leq Z + 1}$ maps each trace variable to a window of states of length at most $Z+1$; \textbf{(3)} $\mu : \trajVars \to \{0, \ldots, Z\}$ maps each stuttering variable $\beta$ to a \emph{relative} position within the window $\Xi(\mathit{base}(\beta))$; \textbf{(4)} $\flat \subseteq \trajVars$ records which of the stutterings has made progress in the last step (used to evaluate $\moved_{\beta}$); and \textbf{(5)} $q \in Q_\psi$ tracks the current state of $\calA_\psi$.
In addition to these regular game vertices, we add a dedicated error vertex $v_\mathit{error}$, which we will reach whenever the window bound $Z$ is insufficient to accommodate all stutterings. 

\begin{example}\label{ex:running}
	Throughout this section, we will use our example from \Cref{fig:strat-asyn} to illustrate our transition rules. 
	In \Cref{fig:strat-asyn}, we represent a vertex $\langle \stage, \Xi, \mu, \flat, q\rangle$ as follows:
	We depict $\Xi : \{\pi_1, \pi_2\} \to (S_\calT)^{\leq 2}$ as an array of nodes; depict the stuttering pointers $\mu  : \{\beta_1, \beta_2\} \to \{0, 1\}$ as arrows; and color the stuttering pointer for $\beta$ in thick iff $\beta \in \flat$.
	For example, the game vertex $\big\langle\ustage, \big[\pi_1 \mapsto [s_0, s_1], \pi_2 \mapsto [s_1, s_2]\big], [\beta_1 \mapsto 0, \beta_2 \mapsto 1], \{\beta_2\}, \_ \big\rangle$ is depicted as vertex $v_{14}$. 
	\demo
\end{example}

We define our final Büchi game as $\game{\calT}{\varphi}{Z} := (V_\verifier, V_\refuter, \alpha, F)$ where $V_\verifier$ contains all vertices of the form $\langle \estage, \Xi, \mu, \flat, q\rangle$, and $V_\refuter$ contains all vertices of the form $\langle \fstage, \Xi, \mu, \flat, q\rangle$, $\langle \ustage, \Xi, \mu, \flat, q\rangle$, and the error vertex $v_\mathit{error}$. 
For the accepting states, we define $F := \{\langle \stage, \Xi, \mu, \flat, q\rangle \mid q \in F_\psi\}$, i.e., the verifier wins a play by infinitely often visiting vertices in which the automaton state is accepting. 
An infinite play in $\game{\calT}{\varphi}{Z}$ is thus won by the verifier iff the simulated run of $\calA_\psi$ is accepting iff $\psi_\mathit{mod}$ is satisfied on the infinite game-play.
The transition function $\alpha$ is defined in \Cref{sec:sub:transition-rules}.

\subsection{Transition Rules}\label{sec:sub:transition-rules}

For the error vertex $v_\mathit{error}$, we define $\alpha(v_\mathit{error}) := \{v_\mathit{error}\}$.
As $v_\mathit{error}$ is a sink vertex that is not contained in $F$, a visit to $v_\mathit{error}$ thus loses the game for the verifier.
For each regular game vertex $v$, we define $\alpha(v) := \{v' \mid v \rightsquigarrow v' \}$, where $\rightsquigarrow$ is defined via the inference rules in \Cref{fig:rules-forall,fig:rules-exist,fig:rules-update}.
Here, each inference rule adds a $\rightsquigarrow$-transition, provided the premise(s) (above the inference rule) are met.

\begin{figure}[!t]
	
	\centering
		\begin{tikzpicture}
			\draw[thick, rounded corners=0mm,gray] (0,1.6) rectangle (5.8, 2.9);
			\node[align=left,anchor=west] at (4.85, 2.2) {\small\textbf{(1.1)}};
			\node[align=left,anchor=west] at (0, 2.2) {\small$\begin{aligned}
					&\Big\{s_\pi \in \kappa_\calT\big( \Xi(\pi)(|\Xi(\pi)| - 1) \big) \Big\}_{\pi \in \traceVars_\forall}\\
					&\mathit{sched} \subseteq \trajVars_\forall
				\end{aligned}$};
			
			\draw[thick, rounded corners=0mm,gray] (0,0.2) rectangle (5.8,1.4);
			\node[align=left,anchor=west] at (4.85, 0.8) {\small \textbf{(1.2)}};
			\node[align=left,anchor=west] at (0, 0.8) {\small$\Xi' = \lambda \pi\ldot \begin{cases}
					\begin{aligned}
						&\Xi(\pi) \cdot s_\pi  &&\text{if } \pi \in \traceVars_\forall\\
						&\Xi(\pi) \quad &&\text{otherwise}
					\end{aligned}
				\end{cases}$};

			\draw[thick, rounded corners=0mm,gray] (6.0,1.2) rectangle (11.9,2.2);
			\node[align=left,anchor=west] at (11, 1.7) {\small \textbf{(1.3)}};
			\node[align=left,anchor=west] at (6.0, 1.7) {\small$\mu' = \lambda \beta\ldot \begin{cases}
					\begin{aligned}
						&\mu(\beta) + 1 &&\text{if } \beta \in \mathit{sched}\\
						&\mu(\beta) &&\text{otherwise} 
					\end{aligned}
				\end{cases}$};
			
			\draw[thick, rounded corners=0mm,gray] (6.0,0.2) rectangle (11.9,1);
			\node[align=left,anchor=west] at (11, 0.6) {\small \textbf{(1.4)}};
			\node[align=left,anchor=west] at (6.0, 0.6) {\small$\flat' = \flat \cup \mathit{sched}$};
			
			\draw[very thick] (-0.1, 0) -- (12, 0);

			\node[align=center] at (6, -0.4) {$\big\langle \fstage, \Xi, \mu, \flat, q\big\rangle\rightsquigarrow \big\langle\estage, \Xi', \mu', \flat', q\big\rangle$};
		\end{tikzpicture}
	\vspace{-2mm}
	\caption{Transitions for vertices in the $\fstage$-stage.}\label{fig:rules-forall}
\end{figure}

\paragraph*{$\fstage$-Stage.}

For each vertex $\langle \fstage, \Xi, \mu, \flat, q\rangle$ in the $\fstage$-stage, the refuter can progress all universally quantified traces and stutterings.  
The transitions from such vertices are defined in \Cref{fig:rules-forall}. 
First, the refuter picks states $\big\{s_\pi \in \kappa_\calT\big( \Xi(\pi)(|\Xi(\pi)| - 1) \big) \big\}_{\pi \in \traceVars_\forall}$ for all universally quantified traces (premise \textbf{(1.1)}).
Note how each $s_\pi$ is a successor state in $\calT$ of the \emph{last} state in the state-window $\Xi(\pi)$ (i.e., $\Xi(\pi)(|\Xi(\pi)| - 1)$), so each state window always models a continuous window along some path in $\calT$.
Additionally, the refuter can decide which of the universally quantified stutterings should be progressed by picking a subset $\mathit{sched} \subseteq \trajVars_\forall$ (premise \textbf{(1.1)}).
After $\{s_\pi\}_{\pi \in \traceVars_\forall}$ and $\mathit{sched}$ are fixed, we can update the game vertex:
We append $s_\pi$ to the state window $\Xi(\pi)$ for each $\pi \in \traceVars_\forall$ (premise \textbf{(1.2)}).
At the same time, $\mathit{sched}$ determines which (universally quantified) stuttering should be progressed (i.e., make a non-stuttering step by moving to the next position in the state window).
For all $\beta \in \mathit{sched}$, we increment the $\mu(\beta)$ pointer (thus pointing to the next state in $\mathit{base}(\beta)$'s state window), and leave $\mu(\beta)$ unchanged for all other stutterings (premise \textbf{(1.3)}).
We record which stutterings have progressed by setting $\flat' = \flat \cup \mathit{sched}$ (premise \textbf{(1.4)}).
Note how each possible successor vertex of $\big\langle \fstage, \Xi, \mu, \flat, q\big\rangle$ \emph{precisely} corresponds to the intended actions that the refuter can take in each vertex (i.e., extend the trace windows of universally quantified traces and decide on universally quantified stutterings). 

\begin{example}
	In the example from \Cref{fig:strat-asyn}, the refuter only controls the trace variable $\pi_1$ ($\pi_2, \beta_1, \beta_2$ are existentially quantified). 
	For example, vertex $v_1$ in \Cref{fig:strat-asyn} represents  game vertex $\big\langle\fstage, \big[\pi_1 \mapsto [s_0], \pi_2 \mapsto [s_0]\big], [\beta_1 \mapsto 0, \beta_2 \mapsto 0], \emptyset, \_ \big\rangle$. 
	According to the rules in \Cref{fig:rules-forall}, this vertex has two successors. 
	In premise \textbf{(1.1)}, we can either pick $s_{\pi_1} = s_0$ or $s_{\pi_1} = s_1$, the two possible successor states of $s_0$ (the last state in $\pi_1$'s window), cf.~the transition system $\calT$ in \Cref{fig:system-asyn}.
	This state is then appended to $\pi_1$'state window, leading to vertices $v_2$ and $v_8$, respectively.
	Note that the stuttering pointers are left unchanged as $\trajVars_\forall = \emptyset$.
	 \demo
\end{example}

\paragraph*{$\estage$-Stage.}

Analogously, for vertices in the $\estage$-stage, the verifier can progress all existentially quantified traces and stutterings, defined in \Cref{fig:rules-exist}.

\begin{example}
	Consider vertex $v_9$ in \Cref{fig:strat-asyn} which represents game vertex $\big\langle\estage, \big[\pi_1 \mapsto [s_0, s_1], \pi_2 \mapsto [s_1]\big], [\beta_1 \mapsto 0, \beta_2 \mapsto 0], \emptyset, \_ \big\rangle$. 
	According to the rules in \Cref{fig:rules-forall}, this vertex has four successors:
	The only choice for $s_{\pi_2}$ (i.e., the state that is appended to $\pi_2$'s state window) is $s_{\pi_2} = s_2$ (as this is the only successor of state $s_1$ in $\calT$, cf.~\Cref{fig:system-asyn}).
	For the scheduling, the verifier can pick any $\mathit{sched} \subseteq \trajVars_\exists = \{\beta_1, \beta_2\}$. 
	Let us pick $\mathit{sched} = \{\beta_2\}$. 
	According to the transition rules in \Cref{fig:rules-exist}, we obtain vertex $\big\langle\ustage, \big[\pi_1 \mapsto [s_0, s_1], \pi_2 \mapsto [s_1, s_2]\big], [\beta_1 \mapsto 0, \beta_2 \mapsto 1], \{\beta_2\}, \_ \big\rangle$, i.e., we append $s_2$ to $\pi_2$'s state window, increment $\mu(\beta_2)$ by $1$ (pointing to the next state in $\pi_2$ state window), leave $\mu(\beta_1)$ unchanged (so $\beta_1$ still points to the same state in $\pi_1$'s state window as in the previous round, effectively stuttering $\beta_1$), and record that we have progressed $\beta_2$ by setting $\flat = \{\beta_2\}$. 
	In \Cref{fig:strat-asyn}, this vertex is depicted as $v_{14}$. 
	\demo
\end{example}

\begin{figure}[!t]
	
	\centering
		\begin{tikzpicture}
			\draw[thick, rounded corners=0mm,gray] (0,1.6) rectangle (5.8, 2.9);
			\node[align=left,anchor=west] at (0, 2.2) {\small$\begin{aligned}
					&\Big\{s_\pi \in \kappa_\calT\big( \Xi(\pi)(|\Xi(\pi)| - 1) \big) \Big\}_{\pi \in \traceVars_\exists}\\
					&\mathit{sched} \subseteq \trajVars_\exists
				\end{aligned}$};
			
			\draw[thick, rounded corners=0mm,gray] (0,0.2) rectangle (5.8,1.4);
			\node[align=left,anchor=west] at (0, 0.8) {\small$\Xi' = \lambda \pi\ldot \begin{cases}
					\begin{aligned}
						&\Xi(\pi) \cdot s_\pi  &&\text{if } \pi \in \traceVars_\exists\\
						&\Xi(\pi) \quad &&\text{otherwise}
					\end{aligned}
				\end{cases}$};

			\draw[thick, rounded corners=0mm,gray] (6.0,1.2) rectangle (11.9,2.2);
			\node[align=left,anchor=west] at (6.0, 1.7) {\small$\mu' = \lambda \beta\ldot \begin{cases}
					\begin{aligned}
						&\mu(\beta) + 1 &&\text{if } \beta \in \mathit{sched}\\
						&\mu(\beta) &&\text{otherwise} 
					\end{aligned}
				\end{cases}$};
			
			\draw[thick, rounded corners=0mm,gray] (6.0,0.2) rectangle (11.9,1);
			\node[align=left,anchor=west] at (6.0, 0.6) {\small$\flat' = \flat \cup \mathit{sched}$};
			
			\draw[very thick] (-0.1, 0) -- (12, 0);

			\node[align=center] at (6, -0.4) {$\big\langle \estage, \Xi, \mu, \flat, q\big\rangle\rightsquigarrow \big\langle\ustage, \Xi', \mu', \flat', q\big\rangle$};
		\end{tikzpicture}
	\vspace{-2mm}
	\caption{Transitions for vertices in the $\estage$-stage.}\label{fig:rules-exist}
\end{figure}

\begin{figure}[!t]
	
	\centering
		\begin{tikzpicture}
			
			\draw[thick, rounded corners=0mm,gray] (0,0.2) rectangle (8.6, 1.2);
			\node[align=left,anchor=west] at (7.6, 0.7) {\small \textbf{(2.1)}};
			\node[align=left,anchor=west] at (0, 0.7) {\small$\begin{aligned}
					\exists \beta_1, \beta_2  \in \trajVars \ldot \mathit{base}(\beta_1) = \mathit{base}(\beta_2) \land |\mu(\beta_1) - \mu(\beta_2)| \geq Z
				\end{aligned}$};
			
			\draw[very thick] (-0.1, 0) -- (8.7, 0);

			\node[align=center] at (4.3, -0.4) {$\big\langle \ustage, \Xi, \mu, \flat, q\big\rangle\rightsquigarrow v_\mathit{error}$};
	\end{tikzpicture}
	
	\vspace{2mm}
	
		\begin{tikzpicture}
			\draw[thick, rounded corners=0mm,gray] (0,2.5) rectangle (6.4,3.5);
			\node[align=left,anchor=west] at (5.5, 3) {\small \textbf{(3.1)}};
			\node[align=left,anchor=west] at (0, 3) {\small$\mathit{mo} = \big\{\pi \mid \forall \beta\ldot (\mathit{base}(\beta) = \pi) \Rightarrow \mu(\beta) \neq 0\big\}$};
			
			\draw[thick, rounded corners=0mm,gray] (0,1.2) rectangle (6.4,2.3);
			\node[align=left,anchor=west] at (5.5, 1.75) {\small \textbf{(3.2)}};
			\node[align=left,anchor=west] at (0, 1.75) {\small$\Xi' = \lambda \pi \ldot \begin{cases}
					{\begin{aligned}
							&\Xi(\pi)[0,Z) & \text{if } \pi \not\in \mathit{mo} \\
							&\Xi(\pi)[1, Z+1)& \text{if } \pi \in \mathit{mo}
					\end{aligned}}
				\end{cases}$};
			
			\draw[thick, rounded corners=0mm,gray] (6.6,1.2) rectangle (12.7,2.3);
			\node[align=left,anchor=west] at (11.8, 1.75) {\small \textbf{(3.3)}};
			\node[align=left,anchor=west] at (6.6, 1.75) {\small$\mu' = \lambda \beta\ldot \begin{cases}
					\begin{aligned}
						&\mu(\beta) - 1 && \text{if } \mathit{base}(\beta) \in \mathit{mo} \\
						&\mu(\beta) && \text{otherwise}
				\end{aligned}\end{cases}$};

			\draw[thick, rounded corners=0mm,gray] (0,0.2) rectangle (12.7, 1.0);
			\node[align=left,anchor=west] at (11.8, 0.6) {\small\textbf{(3.4)}};
			\node[align=left,anchor=west] at (0, 0.6) {\small$\begin{aligned}
					q' = \delta_\psi\Big(q, \Big\{ \theta \in \Theta \mid \big[x_{\beta} \mapsto \ell_\calT\big(\Xi(\mathit{base}(\beta))(\mu(\beta))\big)(x)  \big]_{x \in \calX, \beta \in \trajVars}  \models^\theory \theta \Big\} \cup \Big\{ \moved_{\beta} \mid \beta \in \flat   \Big\} \Big)
				\end{aligned}$};
			
			\draw[very thick] (-0.1, 0) -- (12.8, 0);
			
			\node[align=center] at (6.35, -0.4) {$\big\langle \ustage, \Xi, \mu, \flat, q\big\rangle\rightsquigarrow \big\langle\fstage, \Xi', \mu', \emptyset, q'\big\rangle$};
	\end{tikzpicture}
	
	\vspace{-2mm}

	\caption{Transitions for vertices in the $\ustage$-stage.}\label{fig:rules-update}
\end{figure}

\paragraph*{$\ustage$-Stage.}

The transitions of the $\ustage$-stage are defined in \Cref{fig:rules-update}.
In \textbf{(2.1)}, we first check if the window size $Z$ is insufficient to capture all stutterings for some trace variable, i.e., if there exist two stutterings $\beta_1, \beta_2$ that point to the same state window ($\mathit{base}(\beta_1) = \mathit{base}(\beta_2)$) and differ by at least $Z$ steps. 
If this is the case, we move to $v_\mathit{error}$ and thus let the verifier lose.
Otherwise, we re-enter the $\fstage$-stage. 
Simultaneously, we update the automaton state of $\calA_\psi$ (premise \textbf{(3.4)}):
For this, we need to evaluate all first-order formulas $\theta \in \Theta$ (recall that $\Theta \subseteq \calF_{\calX_{\trajVars}}$ is the finite set of first-order formulas used as atoms in $\varphi$'s LTL body $\psi$). 
Here, $\theta \in \Theta$ holds in the current game vertex iff $\big[x_{\beta} \mapsto \ell_\calT\big(\Xi(\mathit{base}(\beta))(\mu(\beta))\big)(x)  \big] \models^\theory \theta.$
That is, for each indexed variable $x_{\beta}$, we take the current state for the $\beta$ stuttering, which is the $\mu(\beta)$th state within $\mathit{base}(\beta)$'s state window, i.e., $\Xi(\mathit{base}(\pi))(\mu(\beta))$, and look up the value of $x$ via $\calT$'s labeling function $\ell_\calT$.
To evaluate the auxiliary propositions of the form $\moved_{\beta}$ (as used by $\psi_\mathit{mod}$ to ensure fair stutterings, cf.~\Cref{sub:formula-transformation}), we use $\flat$ and set $\moved_{\beta}$ iff $\beta \in \flat$, i.e., iff $\beta$ has been progressed in the previous $\fstage$-stage or $\estage$-stage.
Afterward, we reset $\flat$ to the empty set.
In addition to updating the automaton state, we also trim the state windows in $\Xi$ and update the relative pointers accordingly. 
For this, we first compute a set $\mathit{mo} \subseteq \traceVars$, which contains all trace variables $\pi$ where no stuttering indexes the first position, i.e., all stutterings $\beta$ with $\mathit{base}(\beta) = \pi$ index a position greater than $0$ (premise \textbf{(3.1)}). 
Premise \textbf{(3.2)} then shifts the windows of all traces in $\mathit{mo}$ by removing their first position. 
By construction of $\mathit{mo}$, all stutterings have progressed past the first position, so the state at this position will never be used for the evaluation of $\psi_\mathit{mod}$.
In addition to shifting the window, we also trim the end of the window to ensure that the window has length at most $Z$.
For all stutterings $\beta \in \mathit{mo}$ where we shifted the state window (i.e., $\mathit{base}(\beta) \in \mathit{mo}$), we correct the relative positions in $\mu$ by decrementing the pointer (premise \textbf{(3.3)}). 
For every stuttering pointer $\beta$, the current state in $\mathit{base}(\beta)$'s state window thus does not change (we counteract all window shifts by decrementing the relative position). 

\begin{example}
	Consider vertex $v_{14}$ in \Cref{fig:strat-asyn}, i.e., vertex $\big\langle\ustage, \big[\pi_1 \mapsto [s_0, s_1], \pi_2 \mapsto [s_1, s_2]\big], [\beta_1 \mapsto 0, \beta_2 \mapsto 1], \{\beta_2\}, \_ \big\rangle$.
	In premise \textbf{(3.1)}, we first compute $\mathit{mo} = \{\pi_2\}$. Note that $\pi_1$ is \emph{not} in $\mathit{mo}$ as $\beta_1$ points to $\pi_1$'s first state, so we cannot drop the first position in $\pi_1$'s window yet. 
	We then, update the state windows in premise \textbf{(3.2)}: We remap $\pi_1$ to $[s_0, s_1][0, 1) = [s_0]$ (as $\pi_1 \not\in \mathit{mo}$), and $\pi_2$ to $[s_1, s_2][1, 2) = [s_2]$ (recall that we set $Z = 1$).
	We correct for this, by decrementing $\mu(\beta_2)$ since $\mathit{base}(\beta_2) = \pi_2 \in \mathit{mo}$ (premise \textbf{(3.3)}).
	The stuttering $\beta_2$ thus still maps to the $s_2$ state in $\pi_2$'s window (which now is at position $0$ in the updated state window). 
	This update results in vertex $\big\langle\fstage, \big[\pi_1 \mapsto [s_0], \pi_2 \mapsto [s_2]\big], [\beta_1 \mapsto 0, \beta_2 \mapsto 0], \emptyset, \_ \big\rangle$, depicted as vertex $v_{15}$ in \Cref{fig:strat-asyn}.
	To progress the DBA (which we omitted in \Cref{fig:strat-asyn}), we need to  evaluate the first-order formula in $\Theta = \big\{ a_{\beta_1} \neq a_{\beta_2} \big\}$ (used as atomic formulas in \ref{eq:fair}) and the propositions $\moved_{\beta_1}$ and $\moved_{\beta_2}$ (premise \textbf{(3.4)}).
	In vertex $v_{14}$, we obtain the current value of $a_{\beta_1}$ (resp.~$a_{\beta_2}$) by looking at the $\mu(\beta_1) = 0$th (resp.~$\mu(\beta_2) = 1$th) position in state window $\Xi(\mathit{base}(\beta_1)) = \Xi(\pi_1) = [s_0, s_1]$ (resp.~$\Xi(\mathit{base}(\beta_2)) = \Xi(\pi_2) = [s_1, s_2]$), which is state $s_0$ (resp.~$s_2$), so $a_{\beta_1} = \ell_\calT(s_0)(a) = 0$ (resp.~$a_{\beta_2} = \ell_\calT(s_2)(a) = 1$). 
	Since $\flat = \{\beta_2\}$, we get $\big\{ \moved_{\beta} \mid \beta \in \flat   \big\} = \{\moved_{\beta_2}\}$, so we update the DBA state to $q' = \delta(q, \{ a_{\beta_1} \neq a_{\beta_2}, \moved_{\beta_2} \})$.
	Even though we omitted the DBA in \Cref{fig:strat-asyn}, it is easy to see that all paths in the strategy satisfy $\psi_\mathit{mod} = \ltlN \ltlG (a_{\beta_1} \neq a_{\beta_2}) \land \ltlG\ltlF \moved_{\beta_1} \land \ltlG\ltlF \moved_{\beta_2}$. 
	 \demo
\end{example}

\subsection{Initial Vertices}

Initially, we want to provide the verifier with as much information as possible, so we consider all possible state windows of length $Z$ for universally quantified traces (cf.~\Cref{sub:lookahead}). 
Formally, we define 
$V_\mathit{init}$ as all vertices $\langle\ustage, \Xi, \mu, \emptyset, q_{0, \psi}\rangle$ where $\mu(\beta) = 0$ for all $\beta \in \trajVars$, and the state windows in $\Xi$ satisfy:
$\Xi(\pi) = [s_{0, \calT}]$ for all $\pi \in \traceVars_\exists$ (i.e., all existentially quantified traces start in a state window of length $1$), and for all $\pi \in \traceVars_\forall$, $\Xi(\pi) = [s_0, s_1, \ldots, s_{Z-1}]$ where $s_0 = s_{0, \calT}$ and $s_{i+1} \in \kappa_\calT(s_i)$ for $0 \leq i < Z - 1$ (i.e., all universally quantified traces start in a state window of length $Z$ that contains consecutive states in $\calT$ starting from $s_{0, \calT}$).
We are interested in checking if the verifier can win from all vertices in $V_\mathit{init}$, i.e., whether $\mathit{wins}_{\game{\calT}{\varphi}{Z}}(\verifier, V_{\mathit{init}})$.

\subsection{Soundness}

\begin{restatable}[Soundness]{theorem}{sound}\label{theo:soudness}
	Consider any transition system $\calT$, a $\forall^*\exists^*$ A-HLTL formula $\varphi$, and a bound $Z \in \nat_{\geq 1}$.
	If $\mathit{wins}_{\game{\calT}{\varphi}{Z}}(\verifier, V_{\mathit{init}})$, then $\calT \models \varphi$. 
\end{restatable}
\begin{proof}[Proof Sketch.]
	We use a winning strategy for the verifier to construct witnesses for existentially quantified traces and stutterings.
	For this, we step-wise simulate the universal quantified traces and stutterings (taking the refuter's perspective) and use the strategy's response to construct witnesses. 
	A full proof is given in \ifFull{\Cref{app:sound}}{the full version}.
\end{proof}

Moreover, it is easy to see that our game-based approximation is monotone in the window size:

\begin{restatable}{lemma}{monotone}\label{lem:monotone}
	For any $\calT, \varphi$, if $\mathit{wins}_{\game{\calT}{\varphi}{Z}}(\verifier, V_{\mathit{init}})$  and $Z' \geq Z$, then $\mathit{wins}_{\game{\calT}{\varphi}{Z'}}(\verifier, V_{\mathit{init}})$.
	For any $Z \in \nat_{\geq 1}$, there exists $\calT, \varphi$ such that $\neg \mathit{wins}_{\game{\calT}{\varphi}{Z}}(\verifier, V_{\mathit{init}})$ but $\mathit{wins}_{\game{\calT}{\varphi}{Z+1}}(\verifier, V_{\mathit{init}})$.
\end{restatable}

In the example in \Cref{sec:overview}, we dealt with a unique stuttering for each trace, so a window of size $Z = 1$ suffices to accommodate this single stuttering.	
If we consider \emph{multiple} stutterings on the same trace, we sometimes need larger state windows and thus bounds $Z > 1$.
We give a concrete example in \ifFull{\Cref{app:windowSize}}{the full version}.

\begin{remark}\label{rem:finite}
	In general, game $\game{\calT}{\varphi}{Z}$ contains infinitely many states. 
	If the system $\calT$ is represented symbolically (e.g., using first-order logic), we can derive a symbolic description of  $\game{\calT}{\varphi}{Z}$ and either leverage existing approaches for solving infinite-state games \cite{LaveauxWW22,BeyeneCPR14,HeimD24,FarzanK18,AlfaroHM01,SamuelDK21,SchmuckHDN24} or let the user solve the game interactively \cite{CorrensonF25} (we discuss this further in \Cref{sec:related-work}).
	If $\calT$ is a finite-state transition system (i.e., the set of states $S_\calT$ is finite), $\game{\calT}{\varphi}{Z}$ is a finite-state Büchi game, which we can compute directly.
	In this case, the number of vertices of $\game{\calT}{\varphi}{Z}$ is exponential in the bound $Z$ and the number of traces in $\varphi$ (as is usual for self-compositions \cite{BartheDR04}), but only \emph{linear} in $|S_\calT|$.
	As Büchi games can be solved in polynomial time \cite{McNaughton93}, checking if $\mathit{wins}_{\game{\calT}{\varphi}{Z}}(\verifier, V_{\mathit{init}})$ is decidable in  polynomial time in the size of the system.
\end{remark}

\section{Completeness}\label{sec:complete}

In general, our game-based verification approach is incomplete, i.e., the verifier might lose $\game{\calT}{\varphi}{Z}$, even though $\calT \models \varphi$ holds.
In general, we thus cannot use the game to conclude the violation of a property; only prove satisfaction.
This incompleteness is inevitable: For a finite-state system $\calT$, checking if $\mathit{wins}_{\game{\calT}{\varphi}{Z}}(\verifier, V_{\mathit{init}})$ is decidable (cf.~\Cref{rem:finite}), but $\calT \models \varphi$ is, in general, undecidable \cite{BaumeisterCBFS21}.
However, our method is complete for many fragments and thus constitutes a (finite-state) decision procedure.
In particular, the fragments for which our approach is complete subsume many previously known decidable classes \cite{BaumeisterCBFS21,HsuBFS23}.

\subsection{Alternation-Free Formulas}

We first consider the case in which there is at most one stuttering per trace variable, and all traces and stutterings are quantified either existentially or universally.
In this case, a window size of $Z = 1$ already ensures completeness:

\begin{restatable}{theorem}{altfree}
	Let $\calT$ be any transition system, and let $\varphi$ be an $\exists^*$ or $\forall^*$ A-HLTL formula with at most one stuttering per trace.
	Then $\mathit{wins}_{\game{\calT}{\varphi}{1}}(\verifier, V_{\mathit{init}})$ if and only if $\calT \models \varphi$. 
\end{restatable}

\subsection{Terminating Systems and Lookahead}\label{sub:lookahead}

In our game construction, we ensure that the state window of \emph{universally} quantified traces is -- already in the initial state -- always of length $Z$. 
Recall that our game approximates the A-HLTL semantics: In the semantics of a $\forall^*\exists^*$ formula, all existentially quantified traces and stutterings are fixed knowing the \emph{entire} universally quantified traces and stutterings. In our game, we let the verifier construct existentially quantified traces and stutterings step-wise by responding to the moves of the refuter.
By always keeping the state window of universally quantifier traces of length $Z$ (potentially including trailing states that are not pointed to by any stuttering), we thus offer limited clairvoyance to the verifier, i.e., the verifier can peek at the next $Z$ states on universally quantified traces and make more informed decisions \cite{AbadiL91}.
This is particularly interesting in systems where only a fixed lookahead is needed to know the entire system execution.
We say the transition system $\calT = (S_\calT, s_{0, \calT}, \kappa_\calT, \ell_\calT)$ is \emph{terminating} (also called \emph{tree-shaped} \cite{HsuBFS23}) if there exists a bound $D \in \nat$ (called the \emph{depth}), such that all paths of length at least $D$ reach some sink state, i.e., a state $s \in S_\calT$ with $\kappa_\calT(s) = \{s\}$.
\citet{HsuBFS23} propose a QBF-based model-checking approach for A-HLTL in terminating systems.
We can show that our approach is complete for terminating systems when choosing $Z = D$:

\begin{restatable}{theorem}{term}\label{prop:term}
	Let $\calT$ be a terminating transition system with depth $D$ and let $\varphi$ be a $\forall^*\exists^*$ A-HLTL formula.
	Then $\mathit{wins}_{\game{\calT}{\varphi}{D}}(\verifier, V_{\mathit{init}})$ if and only if $\calT \models \varphi$. 
\end{restatable}

\subsection{Admissible Formulas}\label{sub:admissible}

Next, we study the so-called \emph{admissible} formulas proposed by \citet{BaumeisterCBFS21}.
In the following, we assume that, in theory $\theory$, ``$=$'' is interpreted as equality on $\values$.

\begin{definition}
	An admissible formula has the form $\forall \pi_1\ldots \forall \pi_n \ldot \exists \beta_1 \stutter \pi_1 \ldots \exists \beta_n \stutter \pi_n \ldot \psi$, where $\psi$ is a boolean combination of: \textbf{(1)} Any number of \emph{state formulas}, i.e., formulas that use no temporal operators; and \textbf{(2)} A \emph{single} (positively occurring) \emph{phase formula} of the form $\bigwedge_{i < j} \ltlG \big(\bigwedge_{a \in P_{i, j} } (a_{\beta_i} = a_{\beta_j})\big)$, where $P_{i, j} \subseteq \calX$.\demo
\end{definition}

That is, we existentially quantify over a \emph{unique} stuttering for each trace and limit the relational formulas allowed in $\psi$.
We can think of the set $P_{i, j}$ as denoting a set of \emph{colors}, defined by the evaluation of variables in $P_{i, j}$ (so there are $|\values|^{|P_{i, j}|}$ many colors). 
A phase formula then asserts $\ltlG (\bigwedge_{a \in P_{i, j}} a_{\beta_i} = a_{\beta_j})$ for each pair $i < j$ of traces, and thus requires that $\pi_i$ and $\pi_j$ traverse the same sequence of $P_{i, j}$-colors, albeit at different speed (via stutterings $\beta_i$ and $\beta_j$, respectively).
For example, \ref{eq:asyn-od} is admissible (see \cite{BaumeisterCBFS21} for further examples).  
We can show that our approach is complete for admissible formulas (already for $Z = 1$):

\begin{restatable}{theorem}{admissible}\label{theo:admissible}
	Let $\calT$ be any transition system, and let $\varphi$ be an admissible A-HLTL formula.
	Then $\mathit{wins}_{\game{\calT}{\varphi}{1}}(\verifier, V_{\mathit{init}})$ if and only if $\calT \models \varphi$. 
\end{restatable}

In the following, we provide a proof outline of \Cref{theo:admissible}.
The first direction follows from \Cref{theo:soudness}. 
For the other, assume $\varphi$ is admissible, and $\calT \models \varphi$.
Let $\psi_\mathit{phase} = \bigwedge_{i < j} \ltlG \big(\bigwedge_{a \in P_{i, j} } (a_{\beta_i} = a_{\beta_j})\big)$ be the unique phase formula in $\psi$.
The objective of the verifier is thus to find stutterings that satisfy $\psi_\mathit{phase}$ (whenever such a stuttering exists).

\paragraph{Safe Progress Set.}

Each vertex $v$ in $\game{\calT}{\varphi}{1}$ controlled by the verifier $\verifier$ has the form 
\begin{align*}
	v = \big\langle \estage, &\big[\pi_i \mapsto [s_i, s_i']\big]_{i=1}^n, \big[\beta_i \mapsto 0\big]_{i=1}^n, \emptyset, q \big\rangle,
\end{align*}
i.e., the traces $\pi_1, \ldots, \pi_n$ are mapped to length-$2$ state windows $[s_1, s'_1], \ldots, [s_n, s'_n]$, respectively, and all stutterings $\beta_1, \ldots, \beta_n$ point to the $0$th position in the windows.
We assume that $v$ locally satisfies all coloring constraints in $\psi_\mathit{phase}$, i.e., for all $i < j$ and $a \in P_{i, j}$,  $\ell_\calT(s_i)(a) = \ell_\calT(s_j)(a)$; if this is not the case, $\psi_\mathit{phase}$ is already violated.
Call this assumption \textbf{(A)}.
In vertex $v$, the verifier can decide on which of the stutterings $\beta_1, \ldots, \beta_n$ should be progressed.
We can thus identify each possible successor vertex of $v$ by a so-called \emph{progress set} $M \subseteq \{1, \ldots, n\}$.
Intuitively, $M$ contains exactly those indices $i$ on which $\beta_i$ performs a non-stuttering step by progressing from state $s_i$ to $s_i'$.
For $M \subseteq \{1, \ldots, n\}$, we define states $\mathit{next}^M_v(1), \ldots, \mathit{next}^M_v(n) \in S_\calT$  by defining $\mathit{next}^M_v(i) := s_i'$ if $i \in M$ and otherwise as $\mathit{next}^M_v(i) := s_i$.
That is, every stuttering $\beta_i$ progressed in $M$ is mapped to the second position in the state window in $v$ (i.e., $s_i'$), and all non-progressed stutterings remain in the first state (i.e., $s_i$).
We say progress set $M$ is \emph{safe} if for every $i < j$, and every $a \in P_{i, j}$, $\ell_\calT(\mathit{next}^M_v(i))(a) = \ell_\calT(\mathit{next}^M_v(j))(a)$, i.e., $M$ ensures that all coloring constraints are satisfied \emph{locally} in the next step.

\paragraph{Maximal Safe Progress Set.}

Crucially, the coloring constraints in  $\psi_\mathit{phase}$ are equalities and thus symmetric: 
For example, assume that \textbf{(1)} progress set $\{i\}$ is safe, and \textbf{(2)} progress set $\{j\}$ is also safe.
Then, by \textbf{(1)}, $s_i'$ has the same $P_{i, j}$-color as $s_j$, by \textbf{(2)}, $s_i$ has the same $P_{i, j}$-color as $s_j'$, and, by \textbf{(A)}, $s_i$ has the $P_{i, j}$-same color as $s_j$.
This already implies that $s_i'$ has the same $P_{i, j}$-color as $s_j'$; progress set $\{i, j\}$ is also safe.
More generally, the set of safe progress sets forms a join-complete semilattice:

\begin{restatable}{lemma}{union}\label{lem:union}
	If $M_1$ and $M_2$ are safe progress sets, then $M_1  \cup M_2$ is a safe progress set.
\end{restatable}

We can thus define $M^\mathit{max}_v$ as the unique maximal safe progress set in vertex $v$ (defined as the union of all safe progress sets).
We can now, informally, define a winning strategy $\sigma^\mathit{max}$ for $\verifier$:
In each vertex $v$ that satisfies \textbf{(A)} and still needs to satisfy $\psi_\mathit{phase}$, we define $\sigma^\mathit{max}(v) := v'$, where
\begin{align*}
	v' := \big\langle \ustage, &\big[\pi_i \mapsto [s_i, s_i'] \big]_{i=1}^n, \big[\beta_i \mapsto \mathit{ite}(i \in M^\mathit{max}_v, 1, 0) \big]_{i=1}^n, \big\{ \beta_i \mid i \in M^\mathit{max}_v \big\}, q \big\rangle.
\end{align*}
Here, we write $\mathit{ite}(b, x, y)$ to be $x$ if $b$ holds and otherwise $y$ (short for if-then-else). 
In $v'$, the verifier advances the $\beta_i$-stuttering on $\pi_i$ (by incrementing the $\mu$-pointer from position $0$ to $1$) iff $i \in M^\mathit{max}$. 
If $v$ does not need to satisfy $\psi_\mathit{phase}$ (e.g., because the state formulas already suffice to satisfy $\psi$), the verifier progresses all stutterings $\beta_1, \ldots, \beta_n$.

Intuitively, $\sigma^\mathit{max}$ thus always (locally) progresses as many stutterings as possible while ensuring that the coloring constraints in $\psi_\mathit{phase}$ hold in the next step.
For any combination of traces $\pi_1, \ldots, \pi_n$ that satisfies $\exists \beta_1 \stutter \pi_1 \ldots \exists \beta_n \stutter \pi_n\ldot \psi_\mathit{phase}$, strategy $\sigma^\mathit{max}$ automatically finds a stuttering such that $\psi_\mathit{phase}$ holds.
It is thus easy to see that the verifier wins $\game{\calT}{\varphi}{1}$ by following strategy $\sigma^\mathit{max}$, proving \Cref{theo:admissible}.

\subsection{A New Decidable Fragment}

Using a similar idea as in \Cref{sub:admissible}, we can show that our game-based method is complete for an even broader class of properties.
The main idea is that we want to preserve the ability to find a maximal safe progress set, i.e., whenever we can safely progress some stuttering $\beta_i$ and can safely progress some stuttering $\beta_j$, we can also safely progress both $\beta_i$ and $\beta_j$.

Given an assignment $A  : \calX \to \values$, and a stuttering variable $\beta \in \trajVars$, we write $A^\beta : \calX_{\{\beta\}} \to \values$ for the indexed assignment defined as $A^\beta(x_\beta) := A(x)$. 
For two stuttering variables $\beta_i, \beta_j \in \trajVars$, we say a formula $\theta \in \calF_{\calX_{\{\beta_i, \beta_j\}}}$ is \emph{rectangle closed} if the following holds: for any assignments $A, B, C, D : \calX \to \values$ we have: 
\begin{align*}
	\Big(A^{\beta_i} \cup B^{\beta_j} \models^\theory \theta \Big)\land \Big(C^{\beta_i} \cup B^{\beta_j} \models^\theory \theta\Big) \land \Big(A^{\beta_i} \cup D^{\beta_j} \models^\theory \theta \Big) \Rightarrow \Big(C^{\beta_i} \cup D^{\beta_j} \models^\theory \theta \Big),
\end{align*}
i.e., if $(A, B)$, $(C, B)$, and $(A, D)$ all satisfy $\theta$, then so does $(C, D)$. 

In admissible formulas, we always deal with equalities of the form $a_{\beta_i} = a_{\beta_j}$ (for $a \in P_{i, j}$), which are clearly rectangle closed. 
More generally, all equalities of the form $\theta = (e_1 = e_2)$ -- where $e_1$ and $e_2$ are first-order terms (cf.~\cite{barwise1977introduction}) over a \emph{unique} (not necessarily the same) stuttering variable -- are rectangle closed. 
For example, $a_{\beta_1} = (b_{\beta_2} + c_{\beta_2})$ is rectangle closed, but $a_{\beta_1} = (b_{\beta_2} + c_{\beta_3})$ is not.

\begin{definition}
	A \emph{rectangle closed invariant} is a formula $\varphi = \forall \pi_1\ldots \pi_n. \exists \beta_1 \stutter \pi_1 \ldots \exists \beta_n \stutter \pi_n \ldot \psi$, where $\psi$ is a boolean combination of: \textbf{(1)} Any number of state formulas; and \textbf{(2)} A \emph{single} (positively occurring) formula $\ltlG \bigwedge_{\theta \in \Theta}\theta$ where every $\theta \in \Theta$ is rectangle closed.\demo
\end{definition}

It is easy to see that any admissible formula is a rectangle closed invariant (any equality $a_{\beta_i} = a_{\beta_j}$ is rectangle closed).
Note that an admissible formula only allows constraints of the form $a_{\beta_i} = a_{\beta_j}$, i.e., assert equalities between the \emph{same} variable on both traces. 
Rectangle closed invariants allow arbitrary expressions and can thus relate \emph{arbitrary} variables between two traces. 
Using the same proof idea as for \Cref{theo:admissible}, we can show the more general result:

\begin{theorem}
	Let $\calT$ be any transition system, and let $\varphi$ be a rectangle closed invariant.
	Then $\mathit{wins}_{\game{\calT}{\varphi}{1}}(\verifier, V_{\mathit{init}})$ if and only if $\calT \models \varphi$. 
\end{theorem}

\begin{example}
	For example, assume that a system consists of bits (Boolean variables) $w1, \ldots, wn$, but an attacker cannot observe the bits individually.
	Instead, using a probing attack \cite{WangFTS17}, it can detect if \emph{some} bit is currently enabled. 
	We can generalize \citet{ZdancewicM03}'s observational determinism to this more restrictive attacker setting by defining 
	$$\forall \pi_1. \forall \pi_2. \exists \beta_1 \stutter \pi_1. \exists \beta_2 \stutter \pi_2\ldot (l_{\beta_1} = l_{\beta_2}) \to  \ltlG \big( ({w1}_{\beta_1} \lor \cdots \lor {wn}_{\beta_1} )= ({w1}_{\beta_2} \lor \cdots \lor {wn}_{\beta_2} ) \big),$$
	requiring that only the value of $w1 \lor \cdots\lor wn$ agrees on both traces.  \demo
\end{example}

\section{Implementation and Evaluation}\label{sec:eval}

Beyond our theoretical contributions, our game can be used directly for the (semi-) automated verification of A-HLTL.
Notably, our entire framework applies to both finite and infinite-state systems. 
For infinite-state systems, we envision our approach to be used for interactive verification (similar to successful approaches for HyperLTL \cite{CorrensonF25}); we discuss this further in \Cref{sec:related-work}.
For finite-state systems, we can directly construct and solve the game, resulting in a fully automated verification method.
Note that such finite-state methods are still interesting in an infinite-state setting. Abstractions (e.g., generated by a set of predicates \cite{GrafS97} or an abstract domain \cite{CousotC77}) typically abstract infinite variable domains, while maintaining the temporal behavior of the system, i.e., the abstraction of an infinite-state system results in a finite-state system where each step of the infinite-state system corresponds to one computation step of the finite-state abstraction.
Verifying properties like \citet{ZdancewicM03}'s OD on the abstraction, therefore, still requires finding an appropriate stuttering.
The core computational challenge of A-HLTL verification is still present in the finite-state abstraction, making finite-state verification results and tools relevant.

As a proof-of-concept, we have implemented our game-based method for finite-state systems in a tool we call \tool{}. 
\tool{} reads a system $\calT$ (in the form of a symbolic NuSMV system \cite{CimattiCGGPRST02}), an A-HLTL formula $\varphi$, and a bound $Z \in \nat$, and automatically constructs and solves $\calG_{\calT, \varphi, Z}$ (encoded as an explicit-state game).
If desired, \tool{} can also resolve traces on \emph{different} systems \cite{GoudsmidGS21}.
Internally, we use \texttt{spot} \cite{Duret-LutzRCRAS22} to convert LTL formulas to deterministic automata and use \texttt{oink} \cite{Dijk18} to solve $\calG_{\calT, \varphi, Z}$.

\begin{table}[!t]
	\caption{We compare \tool{} with \texttt{HyperQB} \cite{HsuBFS23} on terminating systems. We depict the size of the system(s) ($|S|$), the verification result, and the total time taken by each tool ($t$).
		For \tool{}, we additionally give the time needed to construct ($t_{\mathit{const}}$) and solve ($t_\mathit{solve}$) $\calG_{\calT, \varphi, 1}$, as well as the number of vertices ($|\calG|$).
		Execution time is measured in seconds,  and the timeout (denoted ``-'') is set to 5 minutes. 
		Note that for \tool{}, the total time $t$ can be slightly larger than the sum of $t_{\mathit{const}}$ and $t_\mathit{solve}$, as it, e.g., includes parsing and preprocessing. }\label{tab:exp-bm}
	
	\vspace{-2mm}
	
	\centering
	
	\def\arraystretch{1.2}
	\setlength\tabcolsep{1.2mm}
	\small
	
		\begin{tabular}{ll@{\hspace{6mm}}lc@{\hspace{3mm}}c@{\hspace{5mm}}llll}
			\toprule[1pt]
			&&&& \cite{HsuBFS23}  & \multicolumn{4}{c}{\tool{}} \\
			\cmidrule[0.6pt](l{0mm}r{4mm}){5-5}
			\cmidrule[0.6pt](l{-1mm}){6-9}
			\textbf{Instances} & \textbf{Property} & $\boldsymbol{|S|}$ & \textbf{Result} & $\boldsymbol{t}$ & $\boldsymbol{t}_{const}$ & $|\boldsymbol{\calG|}$ & $\boldsymbol{t}_{solve}$ & $\boldsymbol{t}$ \\
			\midrule
			\textsc{ACDB} & $\varphi_\mathit{OD}$ & 110 & \xmark{} & 2.96 & 1.00 & 17,119 & 0.54 & \textbf{1.86}\\
			\textsc{ACDB}$_\mathit{ndet}$ & $\varphi_\mathit{OD}$ & 679 & \xmark{} & 12.06 & 3.58 & 71,269 & 2.29 & \textbf{7.18} \\
			\midrule[0.6pt]
			\textsc{ConcLeak} & $\varphi_\mathit{OD}$ & 577 & \xmark{} & 27.16 & 3.44 & 59,495 & 2.29 & \textbf{6.51}\\
			\textsc{ConcLeak}$_\mathit{ndet}$ & $\varphi_\mathit{OD}$ & 2821 & \xmark{} & - & 62.1 & 899,014 & 45.9 & \textbf{114.4} \\
			\midrule
			\textsc{SpecExec}$_\mathit{V1}$ & $\varphi_\mathit{SNI}$ & 29/57 & \xmark{} & 3.71 & 0.199 & 1,219 & 0.09 & \textbf{0.69} \\
			\textsc{SpecExec}$_\mathit{V2}$ & $\varphi_\mathit{SNI}$ & 109/169 & \cmark{} & 8.15 & 0.15 & 355 & 0.05 & \textbf{0.82} \\
			\textsc{SpecExec}$_\mathit{V3}$ & $\varphi_\mathit{SNI}$ & 88/169 & \xmark{} & 11.65 & 0.38 & 4,041 & 0.19 & \textbf{1.69} \\
			\textsc{SpecExec}$_\mathit{V4}$ & $\varphi_\mathit{SNI}$ & 94/169 & \xmark{} & 10.45 & 0.71 & 11,797 & 0.34 & \textbf{2.01} \\
			\textsc{SpecExec}$_\mathit{V5}$ & $\varphi_\mathit{SNI}$ & 94/169 & \xmark{} & 7.91 & 0.15 & 925 & 0.05 & \textbf{0.96} \\
			\textsc{SpecExec}$_\mathit{V6}$ & $\varphi_\mathit{SNI}$ & 85/169 & \xmark{} & 10.80 & 0.62 & 9,573 & 0.29 & \textbf{1.93} \\
			\textsc{SpecExec}$_\mathit{V7}$ & $\varphi_\mathit{SNI}$ & 109/169 & \cmark{} & 7.65 & 0.13 & 355 & 0.04 & \textbf{0.82} \\
			\midrule
			\textsc{DBE} & $\varphi_\mathit{SC}$ & 9/7 & \cmark{} & 1.22 & 0.13 & 469 & 0.06 & \textbf{0.43} \\
			\textsc{LP} & $\varphi_\mathit{SC}$ & 81/77 & \cmark{} & 3.22 & 0.78 & 18,705 & 0.63 & \textbf{1.63} \\
			\textsc{EFLP} & $\varphi_\mathit{SC}$ & 81/249 & \cmark{} & 15.34 & 2.6 & 68,015 & 2.4 & \textbf{5.82}\\
			\midrule
			\textsc{CacheTA} & $\varphi_\mathit{OD}$ & 49 & \xmark{} & 2.51 & 0.44 & 7,683 & 0.34 & \textbf{1.03} \\
			\textsc{CacheTA}$_\mathit{ndet+loops}$ & $\varphi_\mathit{OD}$ & 59 & \xmark{} & \textbf{3.79} & 2.26 & 72,812 & 2.54 & 5.07 \\
			\textsc{CacheTA}$_\mathit{ndet+loops}$ & $\varphi_\mathit{OD}$ & 87 & \xmark{} & 23.31 & 6.26 & 170,509 & 6.9 & \textbf{13.57} \\
			\bottomrule[1pt]
		\end{tabular}
\end{table}

\paragraph{Terminating Systems.}

First, we evaluate \tool{} on terminating systems using the benchmarks from \citet{HsuBFS23}.
While the QBF-based bounded model-checking approach of \citet{HsuBFS23} is, in theory, applicable to arbitrary formulas, their implementation (\texttt{HyperQB}) only implements a few fixed formula templates.
We use the following subset of admissible templates: observational determinism (\ref{eq:asyn-od}), speculative non-interference \cite{GuarnieriKMRS20} ($\varphi_{\mathit{SNI}}$), and correct compilation ($\varphi_\mathit{SC}$), and check them on various NuSMV models.
See \cite{HsuBFS23} for details.
All three properties are admissible and thus fall in the fragment for which \tool{} is complete (cf.~\Cref{theo:admissible}).
We depict the verification results and times in \Cref{tab:exp-bm}.
We used the fixed window size of $Z = 1$, which suffices for completeness. We write \cmark (resp.~\xmark) if the property is satisfied (resp.~violated), which -- by completeness -- we can directly infer by solving the game. 
Generally, we observe that the \hyperqb{} performs well if the system contains many states but is very shallow, leading to a small QBF encoding. 
Still, \tool{} performs faster than \hyperqb{} in a majority of the instances. 

\paragraph{Reactive Systems.}

A particular strength of our approach is the ability to, for the first time, automatically verify A-HLTL on reactive, i.e., non-terminating, systems. 
We depict a few test cases in \Cref{tab:non-terminating}.
First, we check \ref{eq:asyn-od} on the programs from \Cref{fig:insec} (and a synchronous version thereof). 
To increase the size of the system, we scale the number of bits stored in each variable.
We also check more complex (non-admissible) A-HLTL properties like \ref{eq:NIAE} (cf.~\Cref{ex:multiple-traj}) on non-terminating systems. 
For example, the \textsc{Buffer} instance models a system that propagates the low-security input (which can change in each step) to the output (potentially with a delay). 
The system, therefore, violates \ref{eq:asyn-od} but satisfies \ref{eq:NIAE}.
We stress that these examples are \emph{not} meant as real-world examples but rather serve as simple abstractions of real-world components.
The A-HLTL properties reason about the infinite executions of these programs, making verification very challenging (none of the previous methods could automatically verify these examples).

\begin{table}[!t]
	\caption{We verify A-HLTL properties on non-terminating (reactive) systems. We give the verification result and execution time of \tool{}.}\label{tab:non-terminating}
	
		\vspace{-2mm}
	
	\centering
	
	\def\arraystretch{1.2}
	\setlength\tabcolsep{1.2mm}
	\small
	
	\begin{minipage}{0.5\linewidth}
		\centering
			\begin{tabular}{lccc}
				\toprule
				\textbf{System} & \textbf{Property} & \textbf{Result} & $\boldsymbol{t}_\tool{}$ \\
				\midrule
				\Cref{fig:insec}$_\mathit{syn}$ (2 bit) & \ref{eq:asyn-od} & \cmark & 0.31 \\
				\Cref{fig:insec}$_\mathit{syn}$ (4 bit) & \ref{eq:asyn-od} & \cmark & 0.45 \\
				\Cref{fig:insec}$_\mathit{syn}$ (8 bit) & \ref{eq:asyn-od} & \cmark & 1.21 \\
				\Cref{fig:insec} (2 bit) & \ref{eq:asyn-od} & \cmark & 0.33 \\
				\Cref{fig:insec} (4 bit) & \ref{eq:asyn-od} & \cmark & 1.15 \\
				\Cref{fig:insec} (8 bit) & \ref{eq:asyn-od} & \cmark & 109.43 \\
				\bottomrule
			\end{tabular}
	\end{minipage}%
	\begin{minipage}{0.5\linewidth}
		\centering
			\begin{tabular}{lccc}
				\toprule
				\textbf{System} & \textbf{Property} & \textbf{Result} & $\boldsymbol{t}_\tool{}$ \\
				\midrule
				\Cref{fig:system-asyn} & \ref{eq:asyn-od} & \cmark & 0.25 \\
				\textsc{Buffer} & \ref{eq:asyn-od} & \xmark{} & 0.46 \\
				\textsc{Buffer} & \ref{eq:NIAE} & \cmark{} & 0.54 \\
				\textsc{Buffer}$_{\mathit{delay}, Z = 3}$ & \ref{eq:NIAE} & \cmark{} & 31.63 \\
				\textsc{Buffer}$_\mathit{flipped}$ & \ref{eq:asyn-od} & \xmark{} & 0.76 \\
				\textsc{Buffer} & \ref{eq:NIAE} & \cmark{} & 1.32 \\
				\bottomrule
			\end{tabular}
	\end{minipage}
\end{table}

\paragraph{Beyond Finite-State Systems}

We emphasize that \tool{} implements our novel game-based verification approach in its simplest possible form by computing an \emph{explicit-state} parity game.
Our preliminary experiments show that even such a direct implementation performs well compared to previous QBF-based methods.
We stress that the main complexity in A-HLTL verification stems from the expressiveness of the logic itself; the complexity (size) of the system plays a secondary role. 
All properties checked in \Cref{tab:exp-bm,tab:non-terminating} use the high-level asynchronous reasoning of A-HLTL, and delegate the search for an appropriate asynchronous stuttering to the verification tool (cf.~\Cref{sec:related-work}).
Our prototype is not meant as a full-fledged verification tool. 
Instead, it demonstrates that the game-based approach works well in a \emph{finite-state setting}. 
A key advantage of our game lies in its applicability to infinite-state systems.
For example, in the domain of \emph{synchronous} HyperLTL, \citet{CoenenFST19} showed that the synchronous game-based approach works well in a finite-state setting, leading to successful adoption to infinite-state systems; either by solving infinite-state games automatically \cite{BeutnerF22,ItzhakySV24} or interactively \cite{CorrensonF25}.
Our approach handles more general asynchronous properties (including well-known examples like OD), and provides a valuable abstraction: instead of reasoning over infinite traces and stutterings, we only need to reason \emph{locally} over a finite window of states and pointers within that window.
As we argue in \Cref{sec:related-work}, our game thus forms the foundation to extend verification to infinite-state systems.

\section{Related Work}\label{sec:related-work}

\paragraph{Relational Program Logics.}

\emph{Relational Hoare Logic} \cite{Benton04} -- and related Hoare-style hyperproperty specifications \cite{DOsualdoFD22,DardinierM24,MaillardHRM20,GladshteinZAAS24,AssafNSTT17,SousaD16,Beutner24c,DickersonYZD22} -- relate the initial and final states of \emph{multiple} program runs and are thus inherently ``\emph{asynchronous}'' (program executions can take a different number of steps to termination). 
However, RHLs struggle to express \emph{temporal} hyperproperties, i.e., properties that reason about the \emph{temporal behavior} along possibly infinite executions (as, e.g., found in infinite protocol interactions or reactive systems), which are needed for properties like \citet{ZdancewicM03}'s observational determinism.

\begin{wrapfigure}{R}{0.25\linewidth}
	\vspace{-5mm}
	\begin{lrbox}{\myboxi}
\begin{code}
$o$ = $0$
$c$ = $0$
repeat
@@$o$ = $\lfloor$ $c$ / $h$ $\rfloor$
@@$c$ = $c$ + $1$
\end{code}
	\end{lrbox}
	\centering
	\scalebox{1.0}{\usebox{\myboxi}}

	\caption{Example program}\label{fig:od-ex}
	\vspace{-4mm}
\end{wrapfigure}

\paragraph{Hypersafety and Commutativity.}
Various works have studied the verification of hypersafety properties by exploiting commutativity \cite{FarzanV20,FarzanKP22,EilersD023,EilersMH20,AntonopoulosKLNNN23,ShemerGSV19}.
These approaches attempt to find an alignment of different executions that aids verification by exploiting the fact that we can interleave different program executions. 
Our work differs from these approaches in that \emph{quantification over alignments (aka stutterings) is part of the specification} and not just a technique that helps during verification. 
Let us take \citet{ZdancewicM03}'s observational determinism (OD) as an example, which states that all pairs of traces with (initially) identical low-security input are stutter-equivalent w.r.t.~the output.
To verify this property, intermediate program steps, i.e., steps where the output does not change, can be interleaved arbitrarily. 
The problem is that one of the verification challenges lies in the \emph{identification} of points where the output changes.

\begin{example}\label{ex:statetrace}
	Consider the simple integer-based program in \Cref{fig:od-ex}.
	Here, $\lfloor \cdot \rfloor$ rounds down to the nearest integer. 
	It is easy to see that the output sequence of this program (per loop iteration) is $0^h 1^h2^h \cdots$, i.e., for the first $h$ loop iterations, the value of $o$ is $0$, and so forth. 
	The output of the program is stutter-equivalent independent of $h$; the program satisfies OD. 
	Yet, verifying this using commutativity-based frameworks is difficult. 
	Most existing frameworks \cite{FarzanV20,EilersMH20,BeutnerF22,ItzhakySV24,UnnoTK21,ShemerGSV19} search for \emph{some} alignment that facilitates an easy proof of the property.
	To express OD, we do, however, want to compare the output of two executions \emph{whenever the output changes}, i.e., identifying points at which the executions should synchronize is (implicitly) part of the OD specification.
	In the above frameworks, we would need to enforce the alignment manually (by, e.g., defining explicit observation points or adding synchronization assertions).
	Only \emph{afterward} can the $k$-safety verifier find some alignment (between two points where the output changes) that aids verification.
	In our example, we cannot infer the points where the output changes syntactically.
	In contrast, A-HLTL offers a richer, higher-level specification language, allowing us to easily express \citet{ZdancewicM03}'s OD by using first-class \emph{quantification} over stutterings (see \ref{eq:asyn-od}), i.e., the identification of points that need to be aligned is \emph{part of the specification itself}.
	\demo
\end{example}

Our approach thus differs from existing verification approaches in the expressiveness of the specification language:
Most hypersafety frameworks exclusively focus on the verification challenges stemming from infinite state space, whereas our work focuses on the complexity stemming from the specification logic itself. 

\paragraph*{Temporal Logics for Asynchronous Hyperproperties.}

In recent years, many logics have been proposed for the specification of asynchronous hyperproperties
Examples include A-HLTL (the main object of study in this paper) \cite{BaumeisterCBFS21}, $H_\mu$ and mumbling $H_\mu$ (extensions of the polyadic $\mu$-calculus), \cite{GutsfeldMO21,GutsfeldOO24}, 
OHyperLTL \cite{BeutnerF22} (an extension of HyperLTL with explicit observation points), 
HyperLTL$_S$ \cite{BozzelliPS21} (an extension of HyperLTL where operators can skip evaluation steps), HyperLTL$_C$ \cite{BozzelliPS21} (an extension of HyperLTL where only some traces progress), variants of team semantics \cite{KrebsMV018,GutsfeldMOV22}, and extensions of first-order logic with explicit quantification over time \cite{BartocciFHNC22a}.
\citet{BozzelliPS22} studied the expressiveness of many of these logics and found that most of them are \emph{incomparable}.

\paragraph*{Verification of A-HLTL}

\citet{BaumeisterCBFS21} study a restricted class of A-HLTL properties called \emph{admissible} (see \Cref{sub:admissible}).
They show that finite-state model-checking of admissible formulas is decidable, and provide a manual reduction to (synchronous) HyperLTL.  
\citet{HsuBFS23} study \emph{terminating} systems (see \Cref{sub:lookahead}) and present a  QBF-based bounded model-checking approach.
Unlike previous techniques for A-HLTL, our method is soundly applicable to arbitrary $\forall^*\exists^*$ A-HLTL formulas, i.e., it might succeed in verifying a property that does not fall into any known decidable fragment. 

\paragraph{$\forall^*\exists^*$ and Game-Based Verification.}

In recent years, many techniques for verifying hyperproperties \emph{beyond} $k$-safety have been proposed \cite{UnnoTK21,BeutnerF22,AntonopoulosKLNNN23,ItzhakySV24}.
These techniques combine the search for a suitable alignment with the search for witness traces for existential quantifiers.
For example, \citet{UnnoTK21} encode verification of $\forall^*\exists^*$ properties as a specialized form of Horn clauses and let the solver figure out an appropriate alignment. 
Currently, their encoding only searches for \emph{some} alignment that aids verification.
It is an interesting future work to study if the encoding can be extended with \emph{explicit quantification} over alignments (as in A-HLTL). 
\citet{CorrensonNFW24,BeutnerHBF24,HsuBFS23} employ bounded unrolling to verify or refute $\forall\exists$ properties. 
Our verification approach is rooted in a game-based interpretation of traces and stutterings.
For \emph{synchronous} HyperLTL properties, such game-based interpretations have been employed successfully to approximate expensive quantifier alternations \cite{CoenenFST19,BeutnerF22CSF}.
The dynamics of our asynchronous game differ substantially from synchronous games. 
Our key idea of using \emph{windows} and \emph{relative pointers} allows the players to have direct control over the stuttering of traces, which is key to approximating A-HLTL's semantics.
In particular, our asynchronous game allows the players to control the speed (i.e., stuttering) of the traces \emph{individually}.
\citet{ItzhakySV24,BeutnerF22} employ a game-based semantics for asynchronous $\forall^*\exists^*$ properties expressed in OHyperLTL.
In this logic, the \emph{user} indicates at which points the executions should synchronize, and  the verifier attempts to find an alignment between these so-called \emph{observation points}.
As argued before, A-HLTL allows explicit \emph{quantification} over stutterings, which is necessary to succinctly express properties like OD; to express OD in OHyperLTL, we would need to manually place observation points whenever the output changes, which is often not possible syntactically (cf.~\Cref*{ex:statetrace}).
It is interesting to see if the CHC encoding by \citet{ItzhakySV24} can be extended to allow explicit reasoning over stutterings, perhaps using our novel idea of using state windows and relative pointers.

\paragraph{Refinement, Stuttering, and Simulation.}
A particularly interesting class of $\forall\exists$ hyperproperties are \emph{refinement} properties.
Given two systems $\calT_1$ and $\calT_2$, and a set of observable variables $\mathit{Obs}$, we can easily express a lock-step refinement property using (synchronous) HyperLTL as $\varphi_{\mathit{refine}}^\mathit{syn} := \forall \pi_1\ldot \exists \pi_2 \ldot  \ltlG \big(\bigwedge_{a \in \mathit{Obs}} a_{\pi_1} = a_{\pi_2} \big)$.
A popular method for \emph{proving} a refinement property (such as $\varphi_{\mathit{refine}}^\mathit{syn}$) is to establish a \emph{simulation relation} between $\calT_1$ and $\calT_2$ \cite{Milner71}. 
A popular technique for automatically \emph{checking} the existence of a (synchronous) simulation relation is to exploit the connection between simulation relations and winning strategies in \emph{simulation games} \cite{Stirling95}.
While $\varphi_{\mathit{refine}}^\mathit{syn}$ requires a lock-step refinement,  it often suffices to establish a \emph{stuttering} refinement property (see, e.g., \cite{Leroy06}), which we can easily express in A-HLTL as 
\begin{align}\label{eq:refine}
	\forall \pi_1\ldot \exists \pi_2\ldot \exists \beta_1 \stutter \pi_1. \exists \beta_2 \stutter \pi_2\ldot  \ltlG \big(\bigwedge_{a \in \mathit{Obs}} a_{\beta_1} = a_{\beta_2} \big).\tag{$\varphi_{\mathit{refine}}$}
\end{align}
That is, for any execution $\pi_1$ in $\calT_1$, there exists some execution $\pi_2$ in $\calT_2$ that, \emph{up to stuttering}, agrees on the values of all observable variables.
We can prove $\varphi_{\mathit{refine}}$ by finding a \emph{stutter simulation} \cite{Namjoshi97,NamjoshiT20}: 
As long as either of the systems makes a ``tau-step'', i.e., a step that does not change any variable in $\mathit{Obs}$ and thus does not emit a ``visible event'', the other system is allowed to stutter.
Our game-based technique can thus be used as a principled method for proving stuttering refinement.
In this light, we can see our game-based verification method as an asynchronous extension of synchronous simulation games; at least in the very special case of \ref{eq:refine}.
Notably, our game-based method -- when applied to \ref{eq:refine} -- naturally allows both copies to be stuttered independently (as long as the stuttering is progressed \emph{infinitely often} on both traces), i.e., our technique does not require \emph{synchronous progress} where both systems eventually progress \emph{together} (cf.~\cite{ChoSLGD23}).

\paragraph{Interactive Verification of Infinite-State Systems.}

In our implementation, we focused on finite-state systems in order to evaluate our method against existing approaches.
All our results (including soundness and completeness) also apply to infinite-state systems.
While we could attempt to construct and solve the resulting infinite-state game automatically (see, e.g., \cite{LaveauxWW22,BeyeneCPR14,HeimD24,FarzanK18,AlfaroHM01,SamuelDK21,SchmuckHDN24}), the more promising direction is to aim for an interactive proof method. 
\citet{CorrensonF25} recently used the game-based approach for \emph{synchronous} HyperLTL \cite{CoenenFST19,BeutnerF22CSF} to develop an \emph{interactive} proof system to verify $\forall\exists$ properties in infinite-state software.
Intuitively, their approach uses co-induction to incrementally unfold the underlying game and let the user implicitly define a strategy.
Our game-based approach can handle asynchronous hyperproperties that quantify over stutterings, but still yields a plain two-player game.
This game enables a clear verification objective for the user and provides the much-needed abstraction from infinite traces and stutterings; the user only needs to reason about stuttering in the current state. 
Our game then allows us to employ frameworks similar to those used by \citet{CorrensonF25} to let the user interactively prove A-HLTL properties by incrementally unfolding the game arena. 
Our work thus lays the foundation for developing a unified proof system that can handle many prominent information-flow policies in reactive systems.

\section{Conclusion and Future Work}

In this paper, we have presented a principled approach for the (automated or interactive) verification of asynchronous hyperproperties expressed in A-HLTL.
Our method approximates quantifier alternation as a game and seamlessly supports quantification over stutterings.
This allows for sound verification of arbitrary $\forall^*\exists^*$ A-HLTL properties, well beyond the fragments supported by previous methods. 
Moreover, we have shown that our method is complete for many fragments and thus constitutes a finite-state decision procedure. 

Our work leaves numerous exciting directions for future extension.
On the theoretical side, one can, e.g., extend our game-based view to support A-HLTL formulas beyond the $\forall^*\exists^*$ fragment by incorporating incomplete information \cite{Reif84}; intuitively, the verifier should not base its decision on universal traces and stutterings quantified \emph{after} the (existentially quantified) objects controlled by the verifier \cite{BeutnerF25a,BeutnerF25b}.
On the practical side, we consider the verification of infinite-state (software) systems to be of particular interest. 
As we argued above, our game-based methods (and proofs) also apply to (symbolically represented) \emph{infinite-state} systems, yielding infinite-state games. 
Exploring effective techniques to solve the resulting (now infinite-state) game and developing abstraction techniques for a (possibly unbounded) window of states are challenging open problems. 
More concretely, it is interesting to apply the interactive proof system of \citet{CorrensonF25} to our game, allowing the user to verify asynchronous hyperproperties in infinite-state reactive systems interactively.

\begin{acks}
This work was supported by the European Research Council (ERC) Grant HYPER (101055412), and by the German Research Foundation (DFG) as part of TRR 248 (389792660).
\end{acks}

\bibliographystyle{ACM-Reference-Format}
\bibliography{references}

\iffullversion

\newpage

\appendix

\section{A-HLTL Syntax from \citet{BaumeisterCBFS21}}\label{app:semantics}

The original variant of A-HLTL proposed by \citet{BaumeisterCBFS21}, differs from the variant studied in this paper. 
In particular, our variant directly quantifies over stutterings of traces and avoids the notation of so-called trajectories.
In this section, we recap the variant of \citet{BaumeisterCBFS21} and show that our variant is equally expressive.

\subsection{A-HLTL by \citet{BaumeisterCBFS21}}

Let $\traceVars$ be a set of trace variables and $\calW = \{\tau, \tau_1, \tau_2, \ldots\}$ be a set of trajectory variables.
We define $\calX_{\traceVars, \calW} := \{x_{\pi, \tau} \mid x \in \calX, \pi \in \traceVars, \tau \in \calW\}$, i.e., we index the system variables with a trace variable and a trajectory variable. 
A-HLTL formulas are defined by the following grammar
\begin{align*}
	\varphi &:= \exists \pi \ldot \varphi \mid \forall \pi \ldot \varphi \mid  \phi \\
	\phi &:= \trajE \tau \ldot \phi \mid \trajA \tau \ldot \phi \mid  \psi \\
	\psi &:= \theta \mid \neg \psi \mid \psi \land \psi \mid \ltlN \psi \mid \psi \ltlU \psi
\end{align*}%
where $\pi \in \traceVars$, $\tau \in \calW$ is a trajectory, and $\theta \in \calF_{\calX_{\traceVars, \calW}}$. 
We first quantify over traces (in category $\varphi$), followed by secondary quantification of trajectories (in category $\phi$), followed by an LTL body $\psi$.
Intuitively, each trace-trajectory pair $(\pi, \tau)$ describes a (fair) stuttering of trace $\pi$.
The atomic formulas within the LTL body are first-order formulas $\theta \in \calF_{\calX_{\traceVars, \calW}}$, i.e., formulas over indexed system variables.
Different from our definition, the variables are thus not indexed by a stuttering variable, but by a pair of trace and trajectory. 
Here, each variable $x_{\pi, \tau}$ refers to the value of variable $x$ on trace $\pi$ under trajectory (stuttering) $\tau$.
The main difference to our variant is thus that by quantifying over a single trajectory $\tau$, we implicitly quantify stutterings for all traces at the same time, i.e., $(\pi_1, \tau), \ldots, (\pi_n, \tau)$ all denote (\emph{independent}) stutterings of traces $\pi_1, \ldots, \pi_n$, respectively. 
Our variant of A-HLTL makes this explicit by quantifying over a stuttering for each trace individually.

We assume that an A-HLTL formula is closed, i.e., for all atomic formulas $\theta$ used in the LTL body $\psi$ and any variable $x_{\pi, \tau}$ used in $\theta$, the trace variable $\pi$ and trajectory variable $\tau$ are bound by some outer quantifier.

\paragraph{Semantics.}

\citet{BaumeisterCBFS21} define the A-HLTL semantics in terms of explicit \emph{trajectories}, guiding stuttering of traces.
We use a simpler (but equivalent) semantics that avoids trajectories and instead directly reasons about stuttering expansions of traces, following \citet[Prop.~1]{BozzelliPS22}.

Let $\traceSet \subseteq (\calX \to \values)^\omega$ be a set of traces. 
In the semantics of A-HLTL, we maintain a \emph{trace assignment} $\Pi : \traceVars \rightharpoonup (\calX \to \values)^\omega$ mapping trace variables to traces (used to evaluate trace quantification). 
We say a trace assignment $\Pi'$ is a stuttering of  $\Pi$ iff $\Pi'(\pi) \stutter \Pi(\pi)$ for every $\pi \in \traceVars$. 
We write $\mathit{FSE}(\Pi)$ for the set of all stutterings of $\Pi$.
To handle trajectories, we use a higher-order \emph{trajectory assignment} $\Delta : \calW \rightharpoonup (\traceVars \rightharpoonup (\calX \to \values)^\omega)$, mapping trajectory variables to trace assignments.
Intuitively, $\Delta(\tau)$ is the stuttered trace assignment under trajectory $\tau$. 
Given $\Delta$, $x \in \calX$, $\pi \in \traceVars$, $\tau \in \calW$, and $i \in \nat$, the value of variable $x_{\pi, \tau}$ in the $i$th step is thus $\Delta(\tau)(\pi)(i)(x)$, i.e., the value of indexed variable $x_{\pi, \tau}$ is defined as the value of variable $x$ in the $i$th on the trace bound to $\pi$ in the trace assignment bound to $\tau$.
Given a trajectory assignment $\Delta$ and position $i \in \nat$, we define the extended variable evaluation $\Delta_{(i)} : \calX_{\traceVars, \calW} \to \values$ by $\Delta_{(i)}(x_{\pi, \tau}) := \Delta(\tau)(\pi)(i)(x)$.
We can now define the semantics as

\allowdisplaybreaks
\begin{align*}
	\Pi, \Delta, i &\models_\traceSet \exists \pi \ldot \varphi &\text{iff} \quad &\exists t \in \traceSet \ldot \Pi[\pi \mapsto t], \Delta, i \models_\traceSet  \varphi\\
	\Pi, \Delta, i &\models_\traceSet  \forall \pi \ldot \varphi &\text{iff} \quad &\forall t \in \traceSet \ldot \Pi[\pi \mapsto t], \Delta, i \models_\traceSet  \varphi\\[2mm]
	\Pi, \Delta, i &\models_\traceSet  \trajE \tau\ldot \phi &\text{iff} \quad  & \exists \Pi' \in \mathit{FSE}(\Pi)\ldot \Pi, \Delta[\tau \mapsto \Pi'], i \models_\traceSet  \phi\\
	\Pi, \Delta, i &\models_\traceSet  \trajA \tau\ldot \phi &\text{iff} \quad  & \forall \Pi' \in \mathit{FSE}(\Pi)\ldot \Pi, \Delta[\tau \mapsto \Pi'], i \models_\traceSet  \phi\\[2mm]
	\Pi, \Delta,i &\models_\traceSet  \theta &\text{iff} \quad   &\Delta_{(i)} \models^\theory \theta \\
	\Pi, \Delta,i &\models_\traceSet   \neg \psi &\text{iff} \quad & \Pi, \Delta,i \not\models_\traceSet   \psi \\
	\Pi, \Delta,i &\models_\traceSet   \psi_1 \land \psi_2 &\text{iff} \quad  &\Pi, \Delta,i \models_\traceSet  \psi_1 \text{ and }  \Pi, \Delta,i \models_\traceSet   \psi_2\\
	\Pi, \Delta,i &\models_\traceSet   \ltlN  \psi &\text{iff} \quad & \Pi, \Delta, i+1 \models_\traceSet  \psi \\
	\Pi, \Delta,i &\models_\traceSet   \psi_1 \ltlU \psi_2 &\text{iff} \quad & \exists j \geq i \ldot \Pi, \Delta, j\models_\traceSet   \psi_2 \text{ and } \forall i \leq k < j \ldot  \Pi, \Delta, k \models_\traceSet  \psi_1.
\end{align*}%
We first populate a trace assignment $\Pi$ by following the trace-quantifier prefix and adding traces to $\Pi$.
For each trajectory, we then consider FSEs of $\Pi$ (quantified universally or existentially) and add them to $\Delta$. 
Finally, we can evaluate the LTL body following the usual evaluation of boolean and temporal operators. 
Here, an atomic formula $\theta$ holds in step $i$ if the variable assignment $\Delta_{(i)}$ (which assigns values to the indexed variables in $\calX_{\traceVars, \calW}$) satisfies $\theta$ (modulo theory $\theory$).

\subsection{From Trajectories to Stutterings}

It is easy to see that our variant of A-HLTL is equivalent to trajectory-based variant of \cite{BaumeisterCBFS21,HsuBFS23}. 
Intuitively, our definition is more verbose by explicitly quantifying over stutterings for each trace, whereas the trajectory-based variant uses a single trajectory to quantify over stutterings for all traces. 
We show both directions separately.

Given a trajectory-based formula
\begin{align*}
	\quant_1 \pi_1 \ldots \quant_n \pi_n. \trajQ_1 \tau_1 \ldots \trajQ_k \tau_k\ldot \psi
\end{align*}
where $\quant_1, \ldots, \quant_n \in \{\forall, \exists\}$ and $\trajQ_1, \ldots, \trajQ_k \in \{\trajA, \trajE\}$ we can obtain by replacing each trace trajectory pair $(\pi_i, \tau_j)$ with a fresh stuttering variable $\beta_{\pi_i, \tau_j}$.
Formally, we define the set of stuttering variables $\trajVars := \{ \beta_{\pi, \tau} \mid \pi \in \traceVars, \tau \in \calW \}$ and consider the A-HLTL formula 
\begin{align*}
	\quant_1 \pi_1 \ldots \quant_n \pi_n\ldot \overline{\trajA_1} \beta_{\pi_1, \tau_1} \stutter \pi_1. \cdots \overline{\trajA_1} \beta_{\pi_n, \tau_1} \stutter \pi_n. \cdots \overline{\trajA_k} \beta_{\pi_1, \tau_k} \stutter \pi_1 \cdots \overline{\trajA_k} \beta_{\pi_n, \tau_k} \stutter \pi_n \ldot \overline{\psi}
\end{align*}
where we define $\overline{\trajA} := \forall$ and $\overline{\trajE} := \exists$ to replace trajectory quantification with stuttering quantification. 
We define $\overline{\psi}$ as the LTL formula obtained form $\psi$ by replacing each variable $x_{\pi, \tau}$ (indexed by a trace-trajectory pair) with variable $x_{\beta_{\pi, \tau}}$ (indexed by a stuttering variable).

\begin{example}
	Consider the formula 
	\begin{align*}
		\forall \pi_1. \forall \pi_2. \trajE \tau. \ltlG(o_{\pi_1, \tau} = o_{\pi_2, \tau}).
	\end{align*}
	Using our encoding, we translate this to formula 
	\begin{align*}
		\forall \pi_1. \forall \pi_2. \exists \beta_1 \stutter \pi_1. \exists \beta_2 \stutter \pi_2. \ltlG(o_{\beta_1} = o_{\beta_2}),
	\end{align*}
	i.e., for each trace-trajectory pair ($(\pi_1, \tau)$ and $(\pi_2, \tau)$) we simply quantify over a stuttering of that trace. 
	Note that we write $\beta_1$ for $\beta_{\pi_1, \tau}$ and $\beta_2$ for $\beta_{\pi_2, \tau}$. 
	\demo
\end{example}

For each $n \times k$ combination of trace-trajectory pair, we simply add an explicit stuttering.
This is more verbose but also highlights much better what is actually happening.
In fact, we believe one of the core reasons why A-HLTL was hard to understand is that a lot of quantification was happening implicitly: Every time we quantify over a trajectory, we quantify over stutterings of \emph{all} traces, even though these stutterings are completely independent.

\subsection{From Stutterings to Trajectories}

For the other direction, we simply replace each stuttering quantifier with quantification over a unique trajectory. 
Formally, given a stuttering-based A-HLTL formula 
\begin{align*}
	\quant_1 \pi_1 \ldots \quant_n \pi_n. \quant'_1 \beta_1 \stutter \pi_{k_1} \ldots \quant'_m \beta_m \stutter \pi_{k_m}\ldot \psi
\end{align*}
We define the equivalent trajectory-based formula  
\begin{align*}
	\quant_1 \pi_1 \ldots \quant_n \pi_n. \overline{\quant'_1} \tau_{\beta_1} \ldots \overline{\quant'_m} \tau_{\beta_m}\ldot \overline{\psi}
\end{align*}
where we define $\overline{\forall} = \trajA$ and $\overline{\exists} = \trajE$.
We define the new LTL body $\overline{\psi}$ as $\psi$, where we replace each variable $x_\beta$ with $x_{\mathit{base}(\beta), \tau_{\beta}}$. 
Each trajectory $\tau_{\beta}$ is thus used on a unique trace (namely $\mathit{base}(\beta)$), i.e., instead of quantifying over a stuttering $\beta$ of trace $\pi$, we quantify over a trajectory $\tau_\beta$ (and thus a stuttering of \emph{all} traces) but only use trajectory $\tau_\beta$ on trace $\pi = \mathit{base}(\beta)$.

\begin{example}
	Consider the formula 
	\begin{align*}
		\forall \pi_1. \forall \pi_2. \exists \beta_1 \stutter \pi_1. \exists \beta_2 \stutter \pi_2. \ltlG(o_{\beta_1} = o_{\beta_2}),
	\end{align*}
	Using our encoding, we translate this to formula 
	\begin{align*}
		\forall \pi_1. \forall \pi_2. \trajE \tau_{\beta_1}. \trajE \tau_{\beta_2}. \ltlG(o_{\pi_1, \tau_{\beta_1}} = o_{\pi_2, \tau_{\beta_2}}).
	\end{align*}
	Note that this encoding is less-optimal (in this example, we could use the \emph{same} trajectory on $\pi_1$ and $\pi_2$), but nevertheless yields an equivalent formula. 
	\demo
\end{example}

\section{Soundness Proof}\label{app:sound}

In this section, we prove that our game-based method is sound for the verification of Forall-Exist A-HLTL properties by proving \Cref{theo:soudness}.

\sound*

To prove \Cref{theo:soudness}, we distinguish whether $\varphi$ contains universally quantified stutterings or not. 
We provide a proof for the first part in \Cref{prop:part1} and the latter in \Cref{prop:part2}.

\begin{proposition}\label{prop:part1}
	Assume $\varphi$ is a $\forall^*\exists^*$ A-HLTL formula without universally quantified stuttering.
	If $\mathit{wins}_{\game{\calT}{\varphi}{Z}}(\verifier, V_{\mathit{init}})$, then $\calT \models \varphi$. 
\end{proposition}
\begin{proof}
	As $\varphi$ contains no universally quantified stuttering, it has the form 
	\begin{align*}
		\varphi = \forall \pi_1\ldots \pi_n\ldot \exists \pi_{n+1}\ldots \pi_{n+m}\ldot \exists \beta_1 \stutter \pi_{l_1} \ldots \beta_k \stutter \pi_{l_k}\ldot \psi
	\end{align*}
	where $l_1, \ldots, l_k \in \{1, \ldots, n+m\}$.
	We assume $\mathit{wins}_{\game{\calT}{\varphi}{Z}}(\verifier, V_{\mathit{init}})$, so let $\sigma$ be a winning strategy for the verifier in $\calG_{\calT, \varphi, Z}$.
	To show $\calT \models \varphi$, let $t_1, \ldots, t_n \in \traces{\calT}$ be arbitrary traces for the universally quantified trace variables $\pi_1, \ldots, \pi_n$, and let $p_1, \ldots, p_n \in \paths{\calT}$ be paths that generate $t_1, \ldots, t_n$ (i.e., $\ell_\calT(p_i) = t_i$ for all $1 \leq i \leq n$).
	In a first step, we will use $\sigma$ to construct traces $t_{n+1}, \ldots, t_{n+m} \in \traces{\calT}$ for the existentially quantified $\pi_{n+1}, \ldots, \pi_{n+m}$.
	For this, consider the initial vertex in $\calG_{\calT, \varphi, Z}$
	\begin{align*}
		v_0 := \Big\langle \ustage, &\big[ \pi_1 \mapsto p_1[0,Z-1], \ldots, \pi_n \mapsto p_n[0,Z-1], \pi_{n+1}\mapsto [s_{0, \calT}], \ldots, \pi_{n+m} \mapsto [s_{0, \calT}] \big], \\
		&\big[ \beta_i \mapsto 0 \big]_{i = 0}^k, \emptyset, q_{0, \psi}  \Big\rangle
	\end{align*}
	where each universally quantified trace $\pi_1, \ldots, \pi_n$ starts in the initial state window of length $Z$ on paths $p_1, \ldots, p_n$, respectively, and $\pi_{n+1}, \ldots, \pi_{n+m}$ start in the length $1$ initial state window. 
	As $v_0 \in V_\mathit{init}$, strategy $\sigma$ wins $\calG_{\calT, \varphi, Z}$ for $\verifier$ from $v_0$.
	Let $V = V_\verifier \cup V_\refuter$ be the set of vertices in $\calG_{\calT, \varphi, Z}$.
	We now \emph{iteratively} construct a play $\vec{v}\in V^\omega$ by simulating $\sigma$ using the fixed paths $p_1, \ldots, p_n$, i.e., we let the refuter challenge the verifier by always extending the state windows in accordance with $p_1, \ldots, p_n$.
	We will then extract traces $t_{n+1}, \ldots, t_{n+m}$ from $\sigma$'s response. 
	To perform this simulation formally, we keep track of indices $c_1, \ldots, c_n \in \nat$.
	Intuitively, each $c_i$ tracks the position on $p_i$ that should be added to $\pi_i$'s state window at the next opportunity.

	\begin{center}
		\vspace{-2mm}
		\rule[0.5ex]{0.3\linewidth}{1pt} \textbf{Constructing $\vec{v}$} \rule[0.5ex]{0.3\linewidth}{1pt}
		\vspace{-2mm}
	\end{center}
	
	\noindent
	Initially, we set $\vec{v} = [v_0]$ and $c_1= \cdots = c_n = Z$, i.e., the next position that should be added to $\pi_i$'s state window is the $Z$th position on $p_i$.
	We then repeat the following (\emph{ad infinitum}):
	
	Let $v$ be the last vertex in $\vec{v}$.
	\begin{enumerate}[leftmargin=*]
		\item If $v = \langle \ustage, \Xi, \mu, \flat, q \rangle$, we define $v'$ as the unique successor state of $v$ (cf.~\Cref{fig:rules-update}).
		Additionally, we define $$\mathit{mo} = \big\{\pi \mid \forall \beta\ldot (\mathit{base}(\beta) = \pi) \Rightarrow \mu(\beta) \neq 0\big\}$$ as all universally quantified trace variables where none of the stuttering in $\varphi$ point to the first position (similar to \Cref{fig:rules-update} \textbf{(3.1)}).
		For each $\pi_i \in \mathit{mo}$, we define $c_i := c_i + 1$.
		Intuitively, this accounts for the fact that -- in the next $\fstage$-stage -- we can add a fresh state to $\pi_i$'s state window as the first state will be dropped. 
		\item If $v = \langle \estage, \Xi, \mu, \flat, q \rangle$, we define $v' := \sigma(v)$, i.e., query the winning strategy for the verifier. 
		\item If $v = \langle \fstage, \Xi, \mu, \flat, q \rangle$, we define $v' := \langle \estage, \Xi', \mu, \flat, q \rangle$, where
		\begin{align*}
			\Xi' := \big[\pi_1 \mapsto \Xi(\pi_1) \cdot p_1(c_1), \ldots, \pi_n \mapsto \Xi(\pi_n) \cdot p_n(c_n), \pi_{n+1} \mapsto \Xi(\pi_{n+1}), \ldots,  \pi_{n+m} \mapsto \Xi(\pi_{n+m})\big].
		\end{align*}
		That is, for each universally quantified trace $\pi_i$, we append state $p_i(c_i)$ to $\pi_i$'s state window, i.e., we extend the state window in accordance with the fixed path $p_i$ (using the index $c_i$). 
	\end{enumerate}
	Afterward, we append $v'$ to $\vec{v}$, i.e., define $\vec{v} := \vec{v} \cdot v'$, and repeat. 
	
	\begin{center}
		\vspace{-2mm}
		\rule[0.5ex]{0.3\linewidth}{1pt} \textbf{Constructing $\vec{v}$ -- END} \rule[0.5ex]{0.3\linewidth}{1pt}
		\vspace{-2mm}
	\end{center}
	
	\noindent
	Let $\vec{v} \in V^\omega$ be the infinite play computed by the above construction. 
	It is easy to see that $\vec{v}$ is a play in $\calG_{\calT, \varphi, Z}$ that is compatible with $\sigma$.
	As $\sigma$ is winning for $\verifier$, the play $\vec{v}$ is won by $\verifier$. 
	By construction of $\psi_\mathit{mod}$ (cf.~\Cref{sub:formula-transformation}), this implies that the verifier progresses all stutterings (all of which are existentially quantified) infinitely many times. 
	We will now use $\vec{v}$ to \emph{iteratively} construct paths $p_{n+1}, \ldots, p_{n+m} \in \paths{\calT}$, defining witness traces for the existentially quantified trace variables $\pi_{n+1}, \ldots, \pi_{n+m}$.
	Intuitively, we add a state to $p_i \in \{p_{n+1}, \ldots, p_{n+m}\}$, whenever this state was added to the state window of $\pi_i$ by the verifier.
	Formally, For $n < i \leq n+m$, we construct each $p_i$ as follows:
	
	\begin{center}
		\vspace{-2mm}
		\rule[0.5ex]{0.3\linewidth}{1pt} \textbf{Constructing $p_{n+1}, \ldots, p_{n+m}$} \rule[0.5ex]{0.3\linewidth}{1pt}
		\vspace{-2mm}
	\end{center}
	
	\noindent
	We initially set $p_i := \epsilon$ and $j = 0$.
	We then repeat the following (\emph{ad infinitum}):
	\begin{enumerate}[leftmargin=*]
		\item If $\vec{v}(j) = \langle \ustage, \Xi, \mu, \flat, q\rangle$ and $\mu(\beta)  \neq 0$ for all $\beta$ with $\mathit{base}(\beta) = \pi_i$, i.e., all stuttering on $\pi_i$ have progressed past the first state in $\Xi(\pi_i)$.
		We define $p_i := p_i \cdot \Xi(\pi_i)(0)$, i.e., we extend $p_i$ with the first position in $\pi_i$'s state window.
		\item Otherwise, we leave $p_i$ unchanged. 
	\end{enumerate}
	Afterward, we increment $j := j + 1$, and repeat.
	
	\begin{center}
		\vspace{-2mm}
		\rule[0.5ex]{0.3\linewidth}{1pt} \textbf{Constructing $p_{n+1}, \ldots, p_{n+m}$ -- END} \rule[0.5ex]{0.3\linewidth}{1pt}
		\vspace{-2mm}
	\end{center}
	
	\noindent
	As $\sigma$ advances all stutterings infinitely often, and we add a state to $p_i$ iff all stutterings have moved past the first position, we add states to $p_i$ infinitely many times. 
	Let $p_{n+1}, \ldots, p_{n+m} \in \paths{\calT}$ be the resulting \emph{infinite} paths.
	Now define trace $t_i := \ell_\calT(p_i)$ for $n < i \leq n+m$, and let $\Pi := [\pi_i \mapsto t_i]_{i=1}^n$ be the resulting trace assignment. 
	We claim that $\Pi,  \emptyset, 0 \models_\calT \exists \beta_1 \stutter \pi_{l_1} \ldots \beta_k \stutter \pi_{l_k}\ldot \psi$, i.e., traces  $t_{n+1},\ldots, t_{n+m}$ are valid choices for the existentially quantified trace variables in $\varphi$. 
	To show this, we need to -- for each stuttering $\beta \in \{ \beta_1, \ldots, \beta_k \}$ -- construct a (fair) stuttering of trace $\mathit{base}(\beta)$, that satisfy $\psi$.
	To construct these stutterings we, again, use the vertex sequence $\vec{v}$ constructed previously.
	For each $\beta_i \in \{\beta_1, \ldots, \beta_k \}$, we define an infinite state sequence $r_{\beta_i} \in S_\calT^\omega$, defined, for each $i \in \nat$ as $r_{\beta_i}(j) := \Xi(\mathit{base}(\beta_i))(\mu(\beta_i))$, where $\vec{v}(3 \cdot j) = \langle \ustage, \Xi, \mu, \flat, q\rangle$.
	That is, we always look at every third position (i.e., all vertices in the $\ustage$-stage) and use the state that is pointed to by $\mu(\beta_i)$ in state window $\Xi(\mathit{base}(\beta_i))$. 
	First, it is easy to see that for each $\beta_i \in \{\beta_1, \ldots, \beta_k \}$,  $r_{\beta_i}$ is a fair stuttering of the path $\Pi(\mathit{base}(\beta_i))$ (i.e., $r_{\beta_i} \stutter \Pi(\mathit{base}(\beta_i))$).
	Second, we claim that the stuttered state sequences $\{r_{\beta_i}\}$ satisfy $\psi$, i.e., 
	\begin{align*}
		\Pi, \big[ \beta_1 \mapsto \ell_\calT(r_{\beta_1}), \cdots, \beta_k \mapsto \ell_\calT(r_{\beta_k}) \big], 0 \models_\calT \psi.
	\end{align*}
	This follows from the assumption that $\vec{v}$ is won by $\verifier$, i.e., the states in each $\ustage$-stage, together, satisfy $\psi_\mathit{mod}$ and, therefore, $\psi$ (cf.~\Cref{sub:formula-transformation}).
	Note that in $\calG_{\calT, \varphi, Z}$, in each $\ustage$-stage, we update $\calA_\psi$ by, for each stuttering $\beta_i$, taking the $\mu(\beta_i)$'s state in $\Xi(\mathit{base}(\beta_i))$; just as in the definition of $r_{\beta_i}$.
	Putting everything together, $\{r_{\beta_i}\}$ are (fair) stutterings of $p_1, \ldots, p_{n+m}$ that satisfy $\psi$. 
	By the A-HLTL semantics, this already implies that $\Pi,  \emptyset, 0 \models_\calT \exists \beta_1 \stutter \pi_{l_1} \ldots \beta_k \stutter \pi_{l_k}\ldot \psi$.
	This construction works for arbitrary choices of the universally quantified trace variables $\pi_1, \ldots, \pi_n$ in $\varphi$, proving $\calT \models \varphi$ as required.
\end{proof}

\begin{proposition}\label{prop:part2}
	Assume $\varphi$ is a $\forall^*\exists^*$ formula that contains a universally quantified stuttering.
	If $\mathit{wins}_{\game{\calT}{\varphi}{Z}}(\verifier, V_{\mathit{init}})$, then $\calT \models \varphi$. 
\end{proposition}
\begin{proof}
	We assume that $\mathit{wins}_{\game{\calT}{\varphi}{Z}}(\verifier, V_{\mathit{init}})$. 
	First, we observe that this implies that there exists \emph{at most one} universally quantified stuttering for each trace variable in $\varphi$: 
	If there are at least two stutterings to the same trace, the verifier can never win, as the refuter can easily enforce a visit to $v_\mathit{error}$ by letting the two universally quantified stutterings diverge for more than $Z$ steps.
	So, in the following, we assume that there exists exactly one universally quantified stuttering for each trace variable, i.e., $\varphi$ has the form 
	\begin{align*}
		\varphi = \forall \pi_1\ldots \pi_n\ldot \forall \beta'_1 \stutter \pi_1, \ldots, \beta'_n \stutter \pi_n  \ldot \exists \beta_1 \stutter \pi_{l_1} \ldots \beta_k \stutter \pi_{l_k}\ldot \psi.
	\end{align*}
	where $l_1, \ldots, l_k \in \{1, \ldots, n\}$. 
	That is, $\beta'_1, \ldots, \beta'_n$ are the fixed universally quantified stutterings of $\pi_1, \ldots, \pi_n$ respectively, and $\beta_1, \dots, \beta_k$ the existentially quantified stutterings. 
	Let $\calT = (S_\calT, s_{0, \calT}, \kappa_\calT, \ell_\calT)$ and $\calA_\psi = (Q_\psi, q_{0, \psi}, \delta_\psi, F_\psi)$.
	We assume that $\sigma$ is a winning strategy for the verifier in $\calG_{\calT, \varphi, Z}$.
	
	To show $\calT \models \varphi$, let $t_1, \ldots, t_n \in \traces{\calT}$ be arbitrary traces for the universally quantified $\pi_1, \ldots, \pi_n$, and let $p_1, \ldots, p_n \in S_\calT^\omega$ be paths in $\calT$ that generates $t_1, \ldots, t_n$, i.e., $\ell_\calT(p_i) = t_i$.
	Let $\Pi := [\pi_1 \mapsto t_1, \ldots, \pi_n \mapsto t_n]$. 
	Moreover, let $p_1' \ldots, p_n'\in S_\calT^\omega$ be (fair) stutterings of $p_1, \ldots, p_n$ (i.e, $p_i' \stutter p_i$), which correspond to the universally quantified stutterings $\beta'_1, \ldots, \beta_n'$.
	We can represent the stutterings as binary sequences $b_1', \ldots, b_n' \in \bool^\omega$ that dictate when the path is progressed. I.e., for each $j \in \nat$, $b_i'(j) = \top$ iff in the $j$th step $p_i'$ did a non-stuttering step. 
	
	To show $\calT \models \varphi$, we need to find witnesses for the (existentially quantified) stutterings $\beta_1, \ldots, \beta_k$..
	Concretely, for each $\beta_i \in \{\beta_1, \ldots, \beta_k\}$, we will construct a sequence $r_{\beta_i} \in S_\calT^\omega$ as stutterings of path $\Pi(\mathit{base}(\beta_i)$ such that 
	\begin{align*}
		\Pi, \Big[ &\beta_1' \mapsto \ell_\calT(p_1'), \ldots, \beta_n' \mapsto \ell_\calT(p_n'), \beta_1 \mapsto \ell_\calT(r_{\beta_1}), \ldots, \beta_k \mapsto \ell_\calT(r_{\beta_k}) \Big], 0 \models \psi.
	\end{align*}
	As in the proof of \Cref{prop:part1}, we will simulate $\sigma$ to construct these stuttered paths. 
	Consider the initial vertex in $\calG_{\calT, \varphi, Z}$
	\begin{align*}
		v_0 := \Big\langle \ustage, &\big[ \pi_1 \mapsto p_1[0,Z-1], \ldots, \pi_n \mapsto p_n[0, Z-1] \big], \big[ \beta \mapsto 0 \big]_{\beta \in \{\beta_1', \ldots, \beta_n', \beta_1, \ldots, \beta_k\}}, \emptyset,q_{0, \psi}  \Big\rangle
	\end{align*}
	where each universally quantified trace $\pi_i$ starts in the initial state window of length $Z$ on path $p_i$. 	
	
	We now iteratively construct a sequence of vertices $\vec{v}\in V^\omega$.
	In this sequence we will let the refuter simulate $p_1, \ldots, p_n$ and the stutterings $p_1' \ldots, p_n'$ (corresponding to stutterings $\beta_1', \ldots, \beta_n'$), and then use $\sigma$'s response to construct witnesses for $\beta_1, \ldots, \beta_k$. 
	We keep track of variables $c_1, \ldots, c_n \in \nat$.
	Again, each $c_i$ tracks the position on $p_i$ that currently is in the \emph{last} position of the state window for $\pi_i$, which we use to extend the state window in accordance with $p_i$.
	We give the construction of $\vec{v}$ as an algorithm as follows:
	
	\begin{center}
		\vspace{-2mm}
		\rule[0.5ex]{0.3\linewidth}{1pt} \textbf{Constructing $\vec{v}$} \rule[0.5ex]{0.3\linewidth}{1pt}
		\vspace{-2mm}
	\end{center}
	
	\noindent
	Initially, we set $\vec{v} = [v_0]$ and $c_1= \cdots = c_n = Z$.
	We set $j := 0$. 
	We then repeat the following (ad infinitum):\\
	Let $v= \vec{v}(j)$ be the last vertex in $\vec{v}$.
	\begin{enumerate}
		\item If $v = \langle \ustage, \Xi, \mu, \flat, q \rangle$, we define $v'$ as the unique successor state of $v$ (cf.~\Cref{fig:rules-update}).
		For all states where the update shifts the window (i.e., the first position is removed, so a position at the back of the window opens), we increment the counters $c_1, \ldots, c_n$. 
		Formally, we define $\mathit{mo} = \big\{\pi \mid \forall \beta\ldot (\mathit{base}(\beta) = \pi) \Rightarrow \mu(\beta) \neq 0\big\}$ as the set of all (universally quantified) trace variables where none of the stutterings in $\varphi$ points to the first position (similar to \Cref{fig:rules-update}).
		For each $\pi_i \in \mathit{mo}$, we define $c_i := c_i + 1$. 
		\item If $v = \langle \estage, \Xi, \mu, \flat, q \rangle$, we define $v' := \sigma(v)$.
		\item If $v = \langle \fstage, \Xi, \mu, \flat, q \rangle$, we first define 
		\begin{align*}
			\Xi' := \big[\pi_1 \mapsto \Xi(\pi_1) \cdot p_1(c_1), \ldots, \pi_n \mapsto \Xi(\pi_n) \cdot p_n(c_n)\big]
		\end{align*}
		i.e., for each universally quantified trace $\pi_i$, we append state $p_i(c_i)$, i.e., append the state window in accordance with the fixed path $p_i$. 
		By construction, the counters $c_1, \ldots, c_n$ always point to positions on $p_1,\ldots, p_n$  that are a successor state of the windows $\Xi(\pi_1), \ldots, \Xi(\pi_n)$. 
		The state we append is thus a valid transition in $\calT$.
		We define 
		\begin{align*}
			\mu' := \lambda \beta\ldot \begin{cases}
				\begin{aligned}
					&\mu(\beta) + 1 \quad &&\text{if } \beta = \beta_i' \land b_i'(j) = \top \text{ for some } i\\
					&\mu(\beta) \quad &&\text{otherwise}
				\end{aligned}
			\end{cases}
		\end{align*}
		That is, we progress a universally quantified stuttering $\beta_i'$ (with base path $\pi_i$) iff $p_i'$ (the fixed stuttering of $\pi_i$) did a non-stuttering step in the $j$th step (as given by $b_i' \in \bool^\omega$).
		This ensures that we update the stuttering pointer $\beta_1', \ldots, \beta_n'$ such that we reproduce the fixed stuttering $p_1', \ldots, p_n'$ (of paths $p_1, \ldots, p_n$).
		Likewise, we define 
		\begin{align*}
			\flat' := \flat \cup \{ \beta_i' \mid b_i'(j) = \top \}
		\end{align*}
		Finally, we define $v' := \langle \estage, \Xi', \mu', \flat', q \rangle$.
		It is easy to see that $\calG_{\calT, \varphi, Z}$ allows a transition from $v$ to $v'$;
		we appended a state to all state windows and (possibly) progressed the universally quantified stuttering $\beta_1', \ldots, \beta_n'$. (cf.~\Cref{fig:rules-forall}). 
	\end{enumerate}
	Afterward, we append $v'$ to $\vec{v}$, i.e., define $\vec{v} := \vec{v} \cdot v'$. 
	We then increment $j := j + 1$, and repeat. 
	
	\begin{center}
		\vspace{-2mm}
		\rule[0.5ex]{0.3\linewidth}{1pt} \textbf{Constructing $\vec{v}$ -- END} \rule[0.5ex]{0.3\linewidth}{1pt}
		\vspace{-2mm}
	\end{center}
	
	Let $\vec{v} \in V^\omega$ be the infinite play computed by the above construction. 
	It is easy to see that $\vec{v}$ is a play in $\calG_{\calT, \varphi, Z}$ that is compatible with $\sigma$.
	As, by assumption, $\sigma$ is winning, the play $\vec{v}$ is won by $\verifier$. 
	Now by assumption $p_1', \ldots, p_n'$ are \emph{fair} stutterings of $p_1, \ldots, p_n$, so the boolean sequences $b_1', \ldots, b_n'$ contain $\top$ infinitely many times. 
	By construction of $\vec{v}$, the stuttering pointers for  $\beta_1', \ldots, \beta_n'$ are thus also progressed infinitely many times. 
	This implies that the premise $\ltlG\ltlF \moved_{\beta_1'} \land \cdots \land \ltlG\ltlF \moved_{\beta_n'}$ of $\psi_\mathit{mod}$ (cf.~\Cref{sub:formula-transformation}) holds. 
	As $\vec{v}$ is won by $\sigma$, this implies that conclusion of $\psi_\mathit{mod}$ also holds. 
	In particular, $\sigma$ progresses all stutterings $\beta_1, \ldots, \beta_k$  infinitely many times.
	
	We will now use $\vec{v}$ to construct stutterings for $\beta \in \{\beta_1', \ldots, \beta_n',  \beta_1, \ldots, \beta_k\}$ (similar to the proof of \Cref{prop:part1}).
	For each path $\beta \in \{\beta_1', \ldots, \beta_n', \beta_1, \ldots, \beta_k\}$ we define a sequence $r_{\beta} \in S_\calT^\omega$ as follows:
	For each $j \in \nat$, we let $\vec{v}(3 \cdot j) = \langle \ustage, \Xi, \mu, \flat, q\rangle$ and define $r_{\beta} (j) := \Xi(\mathit{base}(\beta))(\mu(\beta))$. 
	That is, we always look at every third position (i.e., all vertices in the $\ustage$-stage) and use the state that is pointed to by $\mu(\beta)$ on state window $\Xi(\mathit{base}(\beta))$.
	Note that for each $1 \leq i \leq n$, $r_{\beta_i'} = p_i'$, i.e., the state sequence for the universally quantified stutterings $\beta_1', \ldots, \beta_n'$ are exactly the fixed stuttering of $p_1, \ldots, p_n$ we used in the construction of $\vec{v}$ (as we progressed the $\beta_i'$ stutterings in exactly those steps where $p_i'$ did a non-stuttering step on $p_i$, i.e., those steps where $b_i'$ was set to true).
	
	It is easy to see that for each $\beta \in \{\beta_1', \ldots, \beta_n', \beta_1, \ldots, \beta_k\}$,  $r_{\beta}$ is a fair stuttering of the path $\Pi(\mathit{base})(\beta)$ (similar to the proof of \Cref{prop:part1}).
	Moreover, we claim that the paths $\{r_{\beta}\}$ satisfy $\psi$, i.e., 
	\begin{align*}
		\Pi, \Big[ &\beta_1' \mapsto \ell_\calT(r_{\beta_1'}), \ldots, \beta_n' \mapsto \ell_\calT(r_{\beta_n'}), \beta_1 \mapsto \ell_\calT(r_{\beta_1}), \ldots, \beta_k \mapsto \ell_\calT(r_{\beta_k}) \Big], 0 \models \psi.
	\end{align*}
	This directly follows from the assumption that $\vec{v}$ is won by $\verifier$:
	As $p_1', \ldots, p_n'$ are \emph{fair} stutterings, the premise of $\psi_\mathit{mod}$ holds, which -- as $\calA_\psi$ tracks $\psi_\mathit{mod}$ --  implies that the conclusion of $\psi_\mathit{mod}$ (and in particular $\psi$) also holds (cf.~\Cref{sub:formula-transformation}).
	
	So $\{r_{\beta}\}$ are (fair) stutterings of $p_1, \ldots, p_n$ that satisfy $\psi$, and for each $1 \leq i \leq n$, $r_{\beta_i'}$ equals the (universally quantified) $p_i'$. 
	By the A-HLTL semantics, this implies that $$\Pi, \emptyset, 0 \models \forall \beta'_1 \stutter \pi_1, \ldots \beta_n' \stutter \pi_n \ldot \exists \beta_1 \stutter \pi_{l_1} \ldots \beta_k \stutter \pi_{l_k}\ldot \psi$$ as desired. 
	This construction, works for arbitrary choices of the universally quantified traces  $\pi_1, \ldots, \pi_n$ and all choices of stutterings $\beta_1', \ldots, \beta_n'$, i.e., for all such choices, we can use $\sigma$ to construct stuttering $\beta_1, \ldots, \beta_k$ that satisfy $\psi$. 
	This proves $\calT \models \varphi$ as required. 
\end{proof}

\section{On the Window Size}\label{app:windowSize}

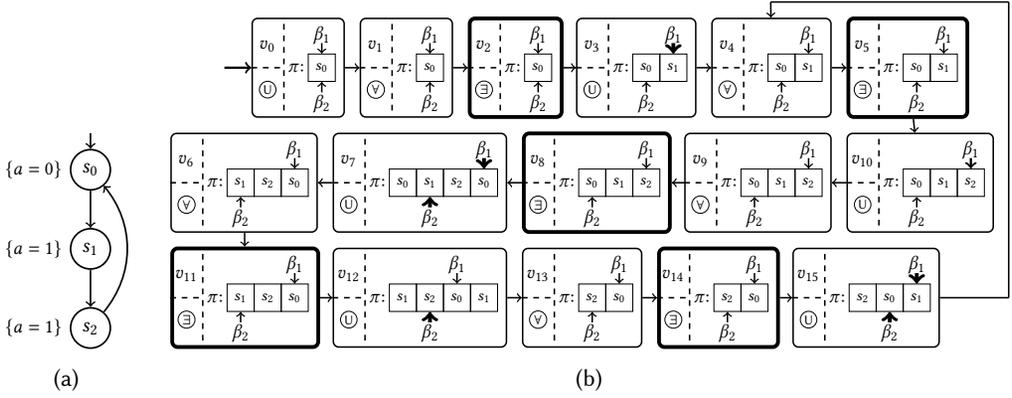
\begin{figure}[!t]
	
	\begin{subfigure}{0.15\linewidth}
		\centering
		\scalebox{0.8}{
			\begin{tikzpicture}
				\node[circle,draw, thick,label=left:{\small$\{a=0\}$}] at (0, 0) (n0) {$s_0$};
				
				\node[circle,draw, thick,label=left:{\small$\{a=1\}$}] at (0, -1.3) (n1) { $s_1$};
				
				\node[circle,draw, thick,label=left:{\small $\{a=1\}$}] at (0, -2.6) (n2) { $s_2$};

				\draw[->, thick] (n0) -- (n1); 
				
				\draw[->, thick] (0, 0.6) -- (n0); 
				
				\draw[->, thick] (n1) to (n2);
				\draw[->, thick,bend right=40] (n2) to (n0);
				
			\end{tikzpicture}
		}
		\subcaption{}\label{fig:stutter-system}
	\end{subfigure}%
	\begin{subfigure}{0.85\linewidth}
		
		\centering
		\scalebox{0.72}{
			\begin{tikzpicture}
				%%%%%%%%%%%%%%%%%%%%%%%%%%%%%%%%
				\coordinate (c) at (0, 0);
				\node[draw, minimum height=18mm,minimum width=17mm, rounded corners=1mm,thick,anchor=west] at (c) (b0) {};
				\node[anchor=east] at ([xshift=11mm]c)  {$\pi$:};
				\node[draw, minimum height=5mm,minimum width=5mm] (n) at ([xshift=13mm]c)  {\footnotesize $s_0$};
				\draw[->, thick] ($(n) + (0, 0.45)$) -- (n); 
				\draw[->, thick] ($(n) + (0, -0.45)$) -- (n); 
				
				\node[anchor=center] at ($(n) + (0, 0.6)$) () {$\beta_1$};
				\node[anchor=center] at ($(n) + (0, -0.65)$) () {$\beta_2$};
				
				\node[] at ([xshift=3mm,yshift=4mm]c) {$v_0$};
				\node[] at ([xshift=3mm,yshift=-4mm]c) {$\ustage$};
				\draw[-, dashed,thick] (c) --  ([xshift=6mm]c);
				\draw[-, dashed,thick] ([xshift=6mm,yshift=8mm]c) --  ([xshift=6mm,yshift=-8mm]c);
				%%%%%%%%%%%%%%%%%%%%%%%%%%%%%%%%
				
				%%%%%%%%%%%%%%%%%%%%%%%%%%%%%%%%
				\coordinate (c) at (2.0, 0);
				\node[draw, minimum height=18mm,minimum width=17mm, rounded corners=1mm,thick,anchor=west] at (c) (b1) {};
				\node[anchor=east] at ([xshift=11mm]c)  {$\pi$:};
				\node[draw, minimum height=5mm,minimum width=5mm] (n) at ([xshift=13mm]c)  {\footnotesize $s_0$};
				\draw[->, thick] ($(n) + (0, 0.45)$) -- (n); 
				\draw[->, thick] ($(n) + (0, -0.45)$) -- (n); 
				
				\node[anchor=center] at ($(n) + (0, 0.6)$) () {$\beta_1$};
				\node[anchor=center] at ($(n) + (0, -0.65)$) () {$\beta_2$};
				
				\node[] at ([xshift=3mm,yshift=4mm]c) {$v_1$};
				\node[] at ([xshift=3mm,yshift=-4mm]c) {$\fstage$};
				\draw[-, dashed,thick] (c) --  ([xshift=6mm]c);
				\draw[-, dashed,thick] ([xshift=6mm,yshift=8mm]c) --  ([xshift=6mm,yshift=-8mm]c);
				%%%%%%%%%%%%%%%%%%%%%%%%%%%%%%%%
				
				%%%%%%%%%%%%%%%%%%%%%%%%%%%%%%%%
				\coordinate (c) at (4.0, 0);
				\node[draw, minimum height=18mm,minimum width=17mm, rounded corners=1mm,thick,anchor=west, line width=2pt] at (c) (b2) {};
				\node[anchor=east] at ([xshift=11mm]c)  {$\pi$:};
				\node[draw, minimum height=5mm,minimum width=5mm] (n) at ([xshift=13mm]c)  {\footnotesize $s_0$};
				\draw[->, thick] ($(n) + (0, 0.45)$) -- (n); 
				\draw[->, thick] ($(n) + (0, -0.45)$) -- (n); 
				
				\node[anchor=center] at ($(n) + (0, 0.6)$) () {$\beta_1$};
				\node[anchor=center] at ($(n) + (0, -0.65)$) () {$\beta_2$};
				
				\node[] at ([xshift=3mm,yshift=4mm]c) {$v_2$};
				\node[] at ([xshift=3mm,yshift=-4mm]c) {$\estage$};
				\draw[-, dashed,thick] (c) --  ([xshift=6mm]c);
				\draw[-, dashed,thick] ([xshift=6mm,yshift=8mm]c) --  ([xshift=6mm,yshift=-8mm]c);
				%%%%%%%%%%%%%%%%%%%%%%%%%%%%%%%%
				
				%%%%%%%%%%%%%%%%%%%%%%%%%%%%%%%%
				\coordinate (c) at (6.0, 0);
				\node[draw, minimum height=18mm,minimum width=22mm, rounded corners=1mm,thick,anchor=west] at (c) (b3) {};
				\node[anchor=east] at ([xshift=11mm]c)  {$\pi$:};
				\node[draw, minimum height=5mm,minimum width=5mm] (n1) at ([xshift=13mm]c)  {\footnotesize $s_0$};
				\node[draw, minimum height=5mm,minimum width=5mm] (n2) at ([xshift=18mm]c)  {\footnotesize $s_1$};
				\draw[->, line width=2pt] ($(n2) + (0, 0.45)$) -- (n2); 
				\draw[->, thick] ($(n1) + (0, -0.45)$) -- (n1); 
				
				\node[anchor=center] at ($(n2) + (0, 0.6)$) () {$\beta_1$};
				\node[anchor=center] at ($(n1) + (0, -0.65)$) () {$\beta_2$};
				
				\node[] at ([xshift=3mm,yshift=4mm]c) {$v_3$};
				\node[] at ([xshift=3mm,yshift=-4mm]c) {$\ustage$};
				\draw[-, dashed,thick] (c) --  ([xshift=6mm]c);
				\draw[-, dashed,thick] ([xshift=6mm,yshift=8mm]c) --  ([xshift=6mm,yshift=-8mm]c);
				%%%%%%%%%%%%%%%%%%%%%%%%%%%%%%%%
				
				%%%%%%%%%%%%%%%%%%%%%%%%%%%%%%%%
				\coordinate (c) at (8.5, 0);
				\node[draw, minimum height=18mm,minimum width=22mm, rounded corners=1mm,thick,anchor=west] at (c) (b4) {};
				\node[anchor=east] at ([xshift=11mm]c)  {$\pi$:};
				\node[draw, minimum height=5mm,minimum width=5mm] (n1) at ([xshift=13mm]c)  {\footnotesize $s_0$};
				\node[draw, minimum height=5mm,minimum width=5mm] (n2) at ([xshift=18mm]c)  {\footnotesize $s_1$};
				\draw[->,  thick] ($(n2) + (0, 0.45)$) -- (n2); 
				\draw[->, thick] ($(n1) + (0, -0.45)$) -- (n1); 
				
				\node[anchor=center] at ($(n2) + (0, 0.6)$) () {$\beta_1$};
				\node[anchor=center] at ($(n1) + (0, -0.65)$) () {$\beta_2$};
				
				\node[] at ([xshift=3mm,yshift=4mm]c) {$v_4$};
				\node[] at ([xshift=3mm,yshift=-4mm]c) {$\fstage$};
				\draw[-, dashed,thick] (c) --  ([xshift=6mm]c);
				\draw[-, dashed,thick] ([xshift=6mm,yshift=8mm]c) --  ([xshift=6mm,yshift=-8mm]c);
				%%%%%%%%%%%%%%%%%%%%%%%%%%%%%%%%
				%		
				%%%%%%%%%%%%%%%%%%%%%%%%%%%%%%%%
				\coordinate (c) at (11, 0);
				\node[draw, minimum height=18mm,minimum width=22mm, rounded corners=1mm,thick,anchor=west,line width=2pt] at (c) (b5) {};
				\node[anchor=east] at ([xshift=11mm]c)  {$\pi$:};
				\node[draw, minimum height=5mm,minimum width=5mm] (n1) at ([xshift=13mm]c)  {\footnotesize $s_0$};
				\node[draw, minimum height=5mm,minimum width=5mm] (n2) at ([xshift=18mm]c)  {\footnotesize $s_1$};
				\draw[->,  thick] ($(n2) + (0, 0.45)$) -- (n2); 
				\draw[->, thick] ($(n1) + (0, -0.45)$) -- (n1); 
				
				\node[anchor=center] at ($(n2) + (0, 0.6)$) () {$\beta_1$};
				\node[anchor=center] at ($(n1) + (0, -0.65)$) () {$\beta_2$};
				
				\node[] at ([xshift=3mm,yshift=4mm]c) {$v_5$};
				\node[] at ([xshift=3mm,yshift=-4mm]c) {$\estage$};
				\draw[-, dashed,thick] (c) --  ([xshift=6mm]c);
				\draw[-, dashed,thick] ([xshift=6mm,yshift=8mm]c) --  ([xshift=6mm,yshift=-8mm]c);
				%%%%%%%%%%%%%%%%%%%%%%%%%%%%%%%%

				%%%%%%%%%%%%%%%%%%%%%%%%%%%%%%%%
				\coordinate (c) at (11, -2.1);
				\node[draw, minimum height=18mm,minimum width=27mm, rounded corners=1mm,thick,anchor=west] at (c) (b10) {};
				\node[anchor=east] at ([xshift=11mm]c)  {$\pi$:};
				\node[draw, minimum height=5mm,minimum width=5mm] (n1) at ([xshift=13mm]c)  {\footnotesize $s_0$};
				\node[draw, minimum height=5mm,minimum width=5mm] (n2) at ([xshift=18mm]c)  {\footnotesize $s_1$};
				\node[draw, minimum height=5mm,minimum width=5mm] (n3) at ([xshift=23mm]c)  {\footnotesize $s_2$};
				\draw[->,  very thick] ($(n3) + (0, 0.45)$) -- (n3); 
				\draw[->, thick] ($(n1) + (0, -0.45)$) -- (n1); 
				
				\node[anchor=center] at ($(n3) + (0, 0.6)$) () {$\beta_1$};
				\node[anchor=center] at ($(n1) + (0, -0.65)$) () {$\beta_2$};
				
				\node[] at ([xshift=3mm,yshift=4mm]c) {$v_{10}$};
				\node[] at ([xshift=3mm,yshift=-4mm]c) {$\ustage$};
				\draw[-, dashed,thick] (c) --  ([xshift=6mm]c);
				\draw[-, dashed,thick] ([xshift=6mm,yshift=8mm]c) --  ([xshift=6mm,yshift=-8mm]c);
				%%%%%%%%%%%%%%%%%%%%%%%%%%%%%%%%

				%%%%%%%%%%%%%%%%%%%%%%%%%%%%%%%%
				\coordinate (c) at (8, -2.1);
				\node[draw, minimum height=18mm,minimum width=27mm, rounded corners=1mm,thick,anchor=west] at (c) (b9) {};
				\node[anchor=east] at ([xshift=11mm]c)  {$\pi$:};
				\node[draw, minimum height=5mm,minimum width=5mm] (n1) at ([xshift=13mm]c)  {\footnotesize $s_0$};
				\node[draw, minimum height=5mm,minimum width=5mm] (n2) at ([xshift=18mm]c)  {\footnotesize $s_1$};
				\node[draw, minimum height=5mm,minimum width=5mm] (n3) at ([xshift=23mm]c)  {\footnotesize $s_2$};
				\draw[->,  thick] ($(n3) + (0, 0.45)$) -- (n3); 
				\draw[->, thick] ($(n1) + (0, -0.45)$) -- (n1); 
				
				\node[anchor=center] at ($(n3) + (0, 0.6)$) () {$\beta_1$};
				\node[anchor=center] at ($(n1) + (0, -0.65)$) () {$\beta_2$};
				
				\node[] at ([xshift=3mm,yshift=4mm]c) {$v_9$};
				\node[] at ([xshift=3mm,yshift=-4mm]c) {$\fstage$};
				\draw[-, dashed,thick] (c) --  ([xshift=6mm]c);
				\draw[-, dashed,thick] ([xshift=6mm,yshift=8mm]c) --  ([xshift=6mm,yshift=-8mm]c);
				%%%%%%%%%%%%%%%%%%%%%%%%%%%%%%%%

				%%%%%%%%%%%%%%%%%%%%%%%%%%%%%%%%
				\coordinate (c) at (5.0, -2.1);
				\node[draw, minimum height=18mm,minimum width=27mm, rounded corners=1mm,thick,anchor=west, line width=2pt] at (c) (b8) {};
				\node[anchor=east] at ([xshift=11mm]c)  {$\pi$:};
				\node[draw, minimum height=5mm,minimum width=5mm] (n1) at ([xshift=13mm]c)  {\footnotesize $s_0$};
				\node[draw, minimum height=5mm,minimum width=5mm] (n2) at ([xshift=18mm]c)  {\footnotesize $s_1$};
				\node[draw, minimum height=5mm,minimum width=5mm] (n3) at ([xshift=23mm]c)  {\footnotesize $s_2$};
				\draw[->,  thick] ($(n3) + (0, 0.45)$) -- (n3); 
				\draw[->, thick] ($(n1) + (0, -0.45)$) -- (n1); 
				
				\node[anchor=center] at ($(n3) + (0, 0.6)$) () {$\beta_1$};
				\node[anchor=center] at ($(n1) + (0, -0.65)$) () {$\beta_2$};
				
				\node[] at ([xshift=3mm,yshift=4mm]c) {$v_8$};
				\node[] at ([xshift=3mm,yshift=-4mm]c) {$\estage$};
				\draw[-, dashed,thick] (c) --  ([xshift=6mm]c);
				\draw[-, dashed,thick] ([xshift=6mm,yshift=8mm]c) --  ([xshift=6mm,yshift=-8mm]c);
				%%%%%%%%%%%%%%%%%%%%%%%%%%%%%%%%

				%%%%%%%%%%%%%%%%%%%%%%%%%%%%%%%%
				\coordinate (c) at (1.5, -2.1);
				\node[draw, minimum height=18mm,minimum width=32mm, rounded corners=1mm,thick,anchor=west] at (c) (b7) {};
				\node[anchor=east] at ([xshift=11mm]c)  {$\pi$:};
				\node[draw, minimum height=5mm,minimum width=5mm] (n1) at ([xshift=13mm]c)  {\footnotesize $s_0$};
				\node[draw, minimum height=5mm,minimum width=5mm] (n2) at ([xshift=18mm]c)  {\footnotesize $s_1$};
				\node[draw, minimum height=5mm,minimum width=5mm] (n3) at ([xshift=23mm]c)  {\footnotesize $s_2$};
				\node[draw, minimum height=5mm,minimum width=5mm] (n4) at ([xshift=28mm]c)  {\footnotesize $s_0$};
				\draw[->,  line width=2pt] ($(n4) + (0, 0.45)$) -- (n4); 
				\draw[->, line width=2pt] ($(n2) + (0, -0.45)$) -- (n2); 
				
				\node[anchor=center] at ($(n4) + (0, 0.6)$) () {$\beta_1$};
				\node[anchor=center] at ($(n2) + (0, -0.65)$) () {$\beta_2$};
				
				\node[] at ([xshift=3mm,yshift=4mm]c) {$v_7$};
				\node[] at ([xshift=3mm,yshift=-4mm]c) {$\ustage$};
				\draw[-, dashed,thick] (c) --  ([xshift=6mm]c);
				\draw[-, dashed,thick] ([xshift=6mm,yshift=8mm]c) --  ([xshift=6mm,yshift=-8mm]c);
				%%%%%%%%%%%%%%%%%%%%%%%%%%%%%%%%
				
				%%%%%%%%%%%%%%%%%%%%%%%%%%%%%%%%
				\coordinate (c) at (-1.5, -2.1);
				\node[draw, minimum height=18mm,minimum width=27mm, rounded corners=1mm,thick,anchor=west] at (c) (b6) {};
				\node[anchor=east] at ([xshift=11mm]c)  {$\pi$:};
				\node[draw, minimum height=5mm,minimum width=5mm] (n1) at ([xshift=13mm]c)  {\footnotesize $s_1$};
				\node[draw, minimum height=5mm,minimum width=5mm] (n2) at ([xshift=18mm]c)  {\footnotesize $s_2$};
				\node[draw, minimum height=5mm,minimum width=5mm] (n3) at ([xshift=23mm]c)  {\footnotesize $s_0$};
				\draw[->,  thick] ($(n3) + (0, 0.45)$) -- (n3); 
				\draw[->, thick] ($(n1) + (0, -0.45)$) -- (n1); 
				
				\node[anchor=center] at ($(n3) + (0, 0.6)$) () {$\beta_1$};
				\node[anchor=center] at ($(n1) + (0, -0.65)$) () {$\beta_2$};
				
				\node[] at ([xshift=3mm,yshift=4mm]c) {$v_6$};
				\node[] at ([xshift=3mm,yshift=-4mm]c) {$\fstage$};
				\draw[-, dashed,thick] (c) --  ([xshift=6mm]c);
				\draw[-, dashed,thick] ([xshift=6mm,yshift=8mm]c) --  ([xshift=6mm,yshift=-8mm]c);
				%%%%%%%%%%%%%%%%%%%%%%%%%%%%%%%%

				%%%%%%%%%%%%%%%%%%%%%%%%%%%%%%%%
				\coordinate (c) at (-1.5, -4.2);
				\node[draw, minimum height=18mm,minimum width=27mm, rounded corners=1mm,thick,anchor=west, line width=2pt] at (c) (b11) {};
				\node[anchor=east] at ([xshift=11mm]c)  {$\pi$:};
				\node[draw, minimum height=5mm,minimum width=5mm] (n1) at ([xshift=13mm]c)  {\footnotesize $s_1$};
				\node[draw, minimum height=5mm,minimum width=5mm] (n2) at ([xshift=18mm]c)  {\footnotesize $s_2$};
				\node[draw, minimum height=5mm,minimum width=5mm] (n3) at ([xshift=23mm]c)  {\footnotesize $s_0$};
				\draw[->,  thick] ($(n3) + (0, 0.45)$) -- (n3); 
				\draw[->, thick] ($(n1) + (0, -0.45)$) -- (n1); 
				
				\node[anchor=center] at ($(n3) + (0, 0.6)$) () {$\beta_1$};
				\node[anchor=center] at ($(n1) + (0, -0.65)$) () {$\beta_2$};
				
				\node[] at ([xshift=3mm,yshift=4mm]c) {$v_{11}$};
				\node[] at ([xshift=3mm,yshift=-4mm]c) {$\estage$};
				\draw[-, dashed,thick] (c) --  ([xshift=6mm]c);
				\draw[-, dashed,thick] ([xshift=6mm,yshift=8mm]c) --  ([xshift=6mm,yshift=-8mm]c);
				%%%%%%%%%%%%%%%%%%%%%%%%%%%%%%%%
				
				%%%%%%%%%%%%%%%%%%%%%%%%%%%%%%%%
				\coordinate (c) at (1.5, -4.2);
				\node[draw, minimum height=18mm,minimum width=32mm, rounded corners=1mm,thick,anchor=west] at (c) (b12) {};
				\node[anchor=east] at ([xshift=11mm]c)  {$\pi$:};
				\node[draw, minimum height=5mm,minimum width=5mm] (n1) at ([xshift=13mm]c)  {\footnotesize $s_1$};
				\node[draw, minimum height=5mm,minimum width=5mm] (n2) at ([xshift=18mm]c)  {\footnotesize $s_2$};
				\node[draw, minimum height=5mm,minimum width=5mm] (n3) at ([xshift=23mm]c)  {\footnotesize $s_0$};
				\node[draw, minimum height=5mm,minimum width=5mm] (n4) at ([xshift=28mm]c)  {\footnotesize $s_1$};
				\draw[->,  thick] ($(n3) + (0, 0.45)$) -- (n3); 
				\draw[->, line width=2pt] ($(n2) + (0, -0.45)$) -- (n2); 
				
				\node[anchor=center] at ($(n3) + (0, 0.6)$) () {$\beta_1$};
				\node[anchor=center] at ($(n2) + (0, -0.65)$) () {$\beta_2$};
				
				\node[] at ([xshift=3mm,yshift=4mm]c) {$v_{12}$};
				\node[] at ([xshift=3mm,yshift=-4mm]c) {$\ustage$};
				\draw[-, dashed,thick] (c) --  ([xshift=6mm]c);
				\draw[-, dashed,thick] ([xshift=6mm,yshift=8mm]c) --  ([xshift=6mm,yshift=-8mm]c);
				%%%%%%%%%%%%%%%%%%%%%%%%%%%%%%%%
				
				%%%%%%%%%%%%%%%%%%%%%%%%%%%%%%%%
				\coordinate (c) at (5.0, -4.2);
				\node[draw, minimum height=18mm,minimum width=22mm, rounded corners=1mm,thick,anchor=west] at (c) (b13) {};
				\node[anchor=east] at ([xshift=11mm]c)  {$\pi$:};
				\node[draw, minimum height=5mm,minimum width=5mm] (n1) at ([xshift=13mm]c)  {\footnotesize $s_2$};
				\node[draw, minimum height=5mm,minimum width=5mm] (n2) at ([xshift=18mm]c)  {\footnotesize $s_0$};
				\draw[->,  thick] ($(n2) + (0, 0.45)$) -- (n2); 
				\draw[->, thick] ($(n1) + (0, -0.45)$) -- (n1); 
				
				\node[anchor=center] at ($(n2) + (0, 0.6)$) () {$\beta_1$};
				\node[anchor=center] at ($(n1) + (0, -0.65)$) () {$\beta_2$};
				
				\node[] at ([xshift=3mm,yshift=4mm]c) {$v_{13}$};
				\node[] at ([xshift=3mm,yshift=-4mm]c) {$\fstage$};
				\draw[-, dashed,thick] (c) --  ([xshift=6mm]c);
				\draw[-, dashed,thick] ([xshift=6mm,yshift=8mm]c) --  ([xshift=6mm,yshift=-8mm]c);
				%%%%%%%%%%%%%%%%%%%%%%%%%%%%%%%%

				%%%%%%%%%%%%%%%%%%%%%%%%%%%%%%%%
				\coordinate (c) at (7.5, -4.2);
				\node[draw, minimum height=18mm,minimum width=22mm, rounded corners=1mm,thick,anchor=west, line width=2pt] at (c) (b14) {};
				\node[anchor=east] at ([xshift=11mm]c)  {$\pi$:};
				\node[draw, minimum height=5mm,minimum width=5mm] (n1) at ([xshift=13mm]c)  {\footnotesize $s_2$};
				\node[draw, minimum height=5mm,minimum width=5mm] (n2) at ([xshift=18mm]c)  {\footnotesize $s_0$};
				\draw[->,  thick] ($(n2) + (0, 0.45)$) -- (n2); 
				\draw[->, thick] ($(n1) + (0, -0.45)$) -- (n1); 
				
				\node[anchor=center] at ($(n2) + (0, 0.6)$) () {$\beta_1$};
				\node[anchor=center] at ($(n1) + (0, -0.65)$) () {$\beta_2$};
				
				\node[] at ([xshift=3mm,yshift=4mm]c) {$v_{14}$};
				\node[] at ([xshift=3mm,yshift=-4mm]c) {$\estage$};
				\draw[-, dashed,thick] (c) --  ([xshift=6mm]c);
				\draw[-, dashed,thick] ([xshift=6mm,yshift=8mm]c) --  ([xshift=6mm,yshift=-8mm]c);
				%%%%%%%%%%%%%%%%%%%%%%%%%%%%%%%%
				
				%%%%%%%%%%%%%%%%%%%%%%%%%%%%%%%%
				\coordinate (c) at (10.0, -4.2);
				\node[draw, minimum height=18mm,minimum width=27mm, rounded corners=1mm,thick,anchor=west] at (c) (b15) {};
				\node[anchor=east] at ([xshift=11mm]c)  {$\pi$:};
				\node[draw, minimum height=5mm,minimum width=5mm] (n1) at ([xshift=13mm]c)  {\footnotesize $s_2$};
				\node[draw, minimum height=5mm,minimum width=5mm] (n2) at ([xshift=18mm]c)  {\footnotesize $s_0$};
				\node[draw, minimum height=5mm,minimum width=5mm] (n3) at ([xshift=23mm]c)  {\footnotesize $s_1$};
				\draw[->,  line width=2pt] ($(n3) + (0, 0.45)$) -- (n3); 
				\draw[->, line width=2pt] ($(n2) + (0, -0.45)$) -- (n2); 
				
				\node[anchor=center] at ($(n3) + (0, 0.6)$) () {$\beta_1$};
				\node[anchor=center] at ($(n2) + (0, -0.65)$) () {$\beta_2$};
				
				\node[] at ([xshift=3mm,yshift=4mm]c) {$v_{15}$};
				\node[] at ([xshift=3mm,yshift=-4mm]c) {$\ustage$};
				\draw[-, dashed,thick] (c) --  ([xshift=6mm]c);
				\draw[-, dashed,thick] ([xshift=6mm,yshift=8mm]c) --  ([xshift=6mm,yshift=-8mm]c);
				%%%%%%%%%%%%%%%%%%%%%%%%%%%%%%%%
				
				\coordinate[] (c1) at (14, -4.2);
				
				\coordinate[] (c2) at (14, 1.2);
				
				\coordinate[] (c3) at (9.6, 1.2);
				
				\draw[->, thick] (b0) -- (b1);
				\draw[->, thick] (b1) -- (b2);
				\draw[->, thick] (b2) -- (b3);
				\draw[->, thick] (b3) -- (b4);
				\draw[->, thick] (b4) -- (b5);
				\draw[->, thick] (b5) -- (b10);
				\draw[->, thick] (b10) -- (b9);
				\draw[->, thick] (b9) -- (b8);
				\draw[->, thick] (b8) -- (b7);
				\draw[->, thick] (b7) -- (b6);
				\draw[->, thick] (b6) -- (b11);
				\draw[->, thick] (b11) -- (b12);
				\draw[->, thick] (b12) -- (b13);
				\draw[->, thick] (b13) -- (b14);
				\draw[->, thick] (b14) -- (b15);
				
				\draw[-, thick] (b15) -- (c1);
				\draw[-, thick] (c1) -- (c2);
				\draw[-, thick] (c2) -- (c3);
				\draw[->,  thick] (c3) -- (b4);

				\draw[->, very thick] (-0.5, 0) -- (b0);
			\end{tikzpicture}
		}
		
		\subcaption{}\label{fig:strat2}
	\end{subfigure}
	
	\caption{In \Cref{fig:stutter-system}, we depict a transition system over system variables $\calX = \{a\}$. In \Cref{fig:strat2}, we give a winning strategy for the verifier (see \Cref{ex:windowSize} for details).
		Similar to \Cref{fig:strat-asyn}, we visualize the stuttering pointers using arrows; in this case, one for stuttering $\beta_1$ and one for stuttering $\beta_2$.
	}
\end{figure}

In the example in \Cref{sec:overview}, we dealt with a unique stuttering for each trace variable, so a window of size $2$ suffices to accommodate all stutterings.	
This changes if we deal with multiple stutterings on the same trace variable.
Consider the following illustrative example. 

\begin{example}\label{ex:windowSize}
	Consider the trivial transition system in \Cref{fig:stutter-system}, generating the unique trace $([a \mapsto 0][a \mapsto 1][a \mapsto 1])^\omega$ and let $\varphi := \exists \pi. \exists \beta_1 \stutter \pi \ldot \exists \beta_2 \stutter \pi\ldot \ltlN \ltlG (a_{\beta_1} \neq a_{\beta_2})$ specify that (on the unique trace), there exist two stutterings that always differ in the value of $a$. 
	It is not hard to see that the verifier loses $\game{\calT}{\varphi}{1}$ and $\game{\calT}{\varphi}{2}$.
	However, the verifier wins $\game{\calT}{\varphi}{3}$, showing that $\calT \models \varphi$ (cf.~\Cref{theo:soudness}).
	We depict a winning strategy for $\game{\calT}{\varphi}{3}$ in \Cref{fig:strat2}. 
	Whenever the game is in the $\ustage$-stage, the value of variable $a$ differs at both stutteringh pointers. 
	To ensure this, both stuttering may diverge for up to $3$ positions (e.g., in vertex $v_7$).
	\demo
\end{example}

\monotone*
\begin{proof}
	We show both part separately:
	\begin{itemize}
		\item
		
		Assume $\sigma$ is a winning strategy for the verifier in $\calG_{\calT, \varphi, Z}$. 
		We can easily extend $\sigma$ to a winning strategy in $\calG_{\calT, \varphi, Z'}$: We simply ignore the larger window size, i.e., the strategy only considers the smaller window of size $Z$ and ignores all states beyond $Z$ for it's decision. 
		
		\item
		We can extend the ideas from \Cref{ex:windowSize}.
		For each $Z \geq 1$, define $\calT_Z := (S, s_0, \kappa, \ell)$, where $S = \{s_0, \ldots, s_Z\}$, $\kappa(s_i) := \{s_{(i+1 )\% (Z + 1)}\}$, and $\ell(s_1) = \cdots = \ell(s_{Z-1}) := [a \mapsto 1], \ell(s_0) := [a \mapsto 0]$.
		Note that $\calT_Z$ generates the unique trace $(\{a \mapsto 0\}\{a \mapsto 1\}^Z)^\omega$.
		For example, $\calT_2$ is depicted in \Cref{fig:stutter-system}. 
		Consider the formula $\varphi := \exists \pi. \exists \beta_1 \stutter \pi\ldot \exists \beta_2 \stutter \pi\ldot \ltlN \ltlG (a_{\beta_1} \neq a_{\beta_2})$.
		It is easy to see that $\calT_Z \models \varphi$.
		However, similar to \Cref{ex:windowSize}, the verifier does not win $\calG_{\calT_Z, \varphi, Z}$ ($\neg \mathit{wins}_{\game{\calT_Z}{\varphi}{Z}}(\verifier, V_{\mathit{init}})$) but does win $\calG_{\calT_Z, \varphi, Z+1}$ ($\mathit{wins}_{\game{\calT_Z}{\varphi}{Z+1}}(\verifier, V_{\mathit{init}})$), as required.
	\end{itemize}
	
\end{proof}

\section{Completeness For Alternation-Free Formulas}\label{app:alt-free}

\altfree*
\begin{proof}
	The first direction follows directly from \Cref{theo:soudness}.
	For the other direction, assume that $\calT \models \varphi$. 
	We first consider the case where $\varphi$ is a $\exists^*$ formula.
	As $\varphi$ contains at most one stuttering per trace variable, we can assume, w.l.o.g., that is has the form 
	\begin{align*}
		\exists \pi_1\ldots \pi_n\ldot \exists \beta_1 \stutter \pi_1, \ldots, \beta_n \stutter \pi_n\ldot \psi
	\end{align*}
	i.e., $\beta_i$ is the unique stuttering on trace $\pi_i$.
	As $\calT \models \varphi$, we get witness traces $t_1, \ldots, t_n \in \traces{\calT}$ for $\pi_1, \ldots, \pi_n$ generated by paths $p_1, \ldots, p_n \in \paths{\calT}$ (i.e., $t_i = \ell_\calT(p_i)$).
	Moreover, by the semantics, there exists stuttered state sequences $p_1', \ldots, p_n' \in S_\calT^\omega$ such that $p_i' \stutter p_i$ for all $1 \leq i \leq n$ and
	\begin{align*}
		[\pi_1 \mapsto t_1, \ldots, \pi_n \mapsto t_n], \big[\beta_1 \mapsto \ell_\calT(p_1'), \ldots, \beta_n \mapsto \ell_\calT(p_n')\big], 0 \models_\calT \psi
	\end{align*} 
	The verifier can now easily win $\calG_{\calT, \varphi, 1}$ by simply appending states according to $p_1, \ldots, p_n$, and stuttering such that the current state at the $i$th visit to the $\ustage$-stage is  $p_1'(i), \ldots, p_n'(i)$. 
	This is similar to the strategy used in the proof of \Cref{theo:soudness}. 
	Note that the refuter has no influence in $\calG_{\calT, \varphi, 1}$, i.e., each vertex controlled by $\refuter$ has a unique successor vertex. 
	As $p_1', \ldots, p_n'$ are fair stutterings of $p_1, \ldots, p_n$ the verifier progresses $\beta_1, \ldots, \beta_n$ infinity often.
	As $p_1', \ldots, p_n'$	satisfy $\psi$, the resulting play is won by the verifier, so $\mathit{wins}_{\game{\calT}{\varphi}{1}}(\verifier, V_{\mathit{init}})$ as required. 
	
	The case where  $\varphi$ is a $\forall^*$ formula is similar. 
	For this, 
	\begin{align*}
		\varphi = \forall \pi_1\ldots \pi_n\ldot \forall \beta_1 \stutter \pi_1, \ldots, \beta_n \stutter \pi_n \ldot \psi
	\end{align*}
	Here, the verifier has no choice in $\calG_{\calT, \varphi, 1}$, i.e., each vertex controlled by $\verifier$ has a unique successor vertex. 
	However, the transitions rules of the refuter (cf.~\Cref{fig:rules-forall}) only allow him to use some path combination $p_1, \ldots, p_n$ in $\calT$ and stutter those paths. 
	The construction of $\psi_\mathit{mod}$ (cf.~\Cref{sub:formula-transformation}), ensures that the verifier wins if the stuttering in unfair (i.e., some path is stuttered forever). 
	If the stuttering is fair, all path combinations satisfy $\psi$ as we assumed $\calT \models \varphi$. 
\end{proof}

\section{Completeness For Terminating System}\label{app:lookahead}

\subsection{Proof of \Cref{prop:term}}

\term*
\begin{proof}
	The ``only if'' direction follows directly from \Cref{theo:soudness}.
	For the other direction, assume that $\calT \models \varphi$ and let $\varphi = \forall \pi_1\ldots \pi_n\ldot \exists \pi_{n+1}\ldots \pi_{n+m}\ldot \exists \beta_1 \stutter \pi_{l_1}  \ldots \beta_k \stutter \pi_{l_k}\ldot \psi$.
	We sketch a strategy for the verifier in $\game{\calT}{\varphi}{D}$, showing that $\mathit{wins}_{\game{\calT}{\varphi}{D}}(\verifier, V_{\mathit{init}})$.
	Each initial vertex in $\game{\calT}{\varphi}{D}$ assigns the (universally quantified) traces $\pi_1, \ldots, \pi_n$ to a state window of length $D$, i.e., a sequence of $D$ consecutive states in $\calT$. 
	Let $y_1, \ldots, y_n \in S_\calT^D$ be those initial path fragments. 
	By the assumption that $\calT$ is terminating with depth $D$, each prefix $y_i$ already defines a unique path (i.e., ends in a sink state). 
	Let $p_1, \ldots, p_n \in \paths{\calT}$ be the unique paths in $\calT$ obtained by extending $y_1, \ldots, y_n$. 
	Now instate the universally quantified trace variables $\pi_1, \ldots, \pi_n$ in $\varphi$ with traces $\ell_\calT(p_1), \ldots, \ell_\calT(p_n)$. 
	As $\calT \models \varphi$, we thus get witness traces $t_{n+1}, \ldots, t_{n+m} \in \traces{\calT}$ for the existentially quantified trace variables $\pi_{n+1}, \ldots, \pi_{n+m}$, and corresponding paths $p_{n+1}, \ldots, p_{n+m} \in \paths{\calT}$ (i.e., $\ell_\calT(p_i) = t_i$ for $n < i \leq n+m$).
	Moreover, for each $\beta_i \in \{\beta_1, \ldots, \beta_k\}$, there exists a stuttering $r_{\beta_i} \in S_\calT^\omega$ of path $p_{l_i}$ (i.e., $r_{\beta_i} \stutter p_{l_i}$), such that 
	\begin{align*}
		\big[\pi_i \mapsto t_i\big]_{\pi_i \in \{\pi_1, \ldots, \pi_{n+m} \}}, \big[\beta_i \mapsto \ell_\calT(r_{\beta_i})  \big]_{\beta_i \in \{\beta_1, \ldots, \beta_k\}}, 0 \models_\calT \psi.
	\end{align*}
	The strategy for the verifier in $\game{\calT}{\varphi}{D}$ now simply extends the state windows of $\pi_{n+1}, \ldots, \pi_{n+m}$ according to $p_{n+1}, \ldots, p_{n+m}$, and stutters $\beta_i$ according to $r_{\beta_i}$ (similar to the refuter in the proof of \Cref{theo:soudness}).
	Note that a window of size $D$ suffices, even though we deal with multiple stutterings on the same trace: After progressing a stuttering pair  $D$ times, the state never changes (as we reach a sink state after at most $D$ steps). 
	As long as some stuttering pointer has not reached a sink state, the verifier can thus stutter according to $\{r_{\beta_i}\}$ without letting any stuttering pointer grow strictly larger than $D$.
\end{proof}

\subsection{Clairvoyance By Window Size}

\begin{figure}[!t]
	\centering
	\scalebox{0.8}{
		\begin{tikzpicture}
			\node[circle,draw, thick,label=above:{\small$\{a = 0\}$}] at (0, 0) (n0) {$s_0$};
			
			\node[circle,draw,thick,label=below:{\small$\{a = 1\}$}] at (0, -1.5) (n1) {$s_1$};
			
			\draw[->, thick] (-0.6,0) -- (n0); 
			
			\draw[->, thick] (n0) edge[bend right] (n1);
			\draw[->, thick] (n1) edge[bend right] (n0);

			\draw [->, thick] (n0) edge[loop right] (n0);
			\draw [->, thick] (n1) edge[loop right] (n1);
		\end{tikzpicture}
	}
	
	\caption{}\label{fig:system-syn}
\end{figure}

Consider the simple system $\calT$ in \Cref{fig:system-syn} which generates traces over variable $a$. 
For $n \in \nat$, consider the following HyperLTL property: 
\begin{align*}
	\varphi_n := \forall \pi_1\ldot \exists \pi_2\ldot \exists \beta_1 \stutter \pi_1. \exists \beta_2 \stutter \pi_2 \ldot \ltlN\ltlG ((a_{\beta_2} = 0) \leftrightarrow \ltlN^n (a_{\beta_1} = 0))
\end{align*}
where we write $\ltlN^n$ for the $n$-fold application of $\ltlN$. 
This formula requires that for every trace $\pi_1$, some trace $\pi_2$ ``predicts $\pi_1$'', i.e., $a$ should hold in the $i$th step on $\pi_2$ iff it holds in the $(i+n)$th step on $\pi_1$. 
It is easy to see that $\calT \models \varphi_n$. 
However, for $n > 1$, $\verifier \not\models \game{\calT}{\varphi_n}{1}$, i.e., the verifier looses $\game{\calT}{\varphi_n}{1}$. 
The reason for this is simple: 
In our game, we construct $\pi_1, \pi_2$ step-wise, so the verifier has to decide on a successor state for the $\pi_2$ window, knowing only a prefix of $\pi_1$. 
However, to ensures that $\varphi_n $ holds, the verifier would need to know which state the refuter will pick $n$ steps into the future. 

As $\varphi_n$ contains a single stuttering per trace variable, it seems like there is no need to use a window size larger than $1$. 
However, it turns out that the window size has a secondary benefit (beyond accommodating multiple stutterings on the same trace): It provides clairvoyance to the verifier. 
In $\game{\calT}{\varphi}{Z}$, we always maintain a window of size $Z$ for the universally quantified traces (cf.~\Cref{fig:rules-update}).
The verifier thus does not only see the first positions in the state window, but also the next $Z$ steps. 
For property $\varphi_n$, knowing the next $Z = n$ steps is exactly the information that was missing:
It is easy to see that the verifier wins $\game{\calT}{\varphi_n}{n}$, i.e., knowing the next $n$ steps provides enough information to predict $\pi_1$.

\section{Completeness For Admissible Formulas}\label{app:admissible}

In this section, we prove \Cref{theo:admissible}:

\admissible*

We assume $\varphi = \forall \pi_1\ldots \forall \pi_n \ldot \exists \beta_1 \stutter \pi_1 \ldots \exists \beta_n \stutter \pi_n \ldot \psi$
In \Cref{sub:admissible}, we already provide a rough outline and sketch the strategy that wins $\calG_{\calT, \varphi, 1}$. 
We recall some of the definitions from \Cref{sub:admissible}:
Each vertex $v$ in $\game{\calT}{\varphi}{1}$ controlled by $\verifier$ has the form 
\begin{align*}
	v = \big\langle \estage, &\big[\pi_i \mapsto [s_i, s_i']\big]_{i=1}^n, \big[ \beta_i \mapsto 0\big]_{i=1}^n, \emptyset,  q \big\rangle,
\end{align*}
i.e., each trace $\pi_i$ is mapped to a length-$2$ state window $[s_i, s'_i]$, and the $\beta_i$-stuttering points to the $0$th position in this window. 
We assume that $v$ satisfies all coloring constraints in $\psi_\mathit{phase}$, i.e., for all $i < j$ and $a \in P_{i, j}$,  $a \in \ell_\calT(s_i)(a) = \ell_\calT(s_j)(a)$ (if this is not the case, $\psi_\mathit{phase}$ is already violated).
Call this assumption \textbf{(A)}.
For $M \subseteq \{1, \ldots, n\}$, we define states $\mathit{next}^M_v(1), \ldots, \mathit{next}^M_v(n) \in S_\calT$  by $\mathit{next}^M_v(i) = s_i$ if $i \not\in M$ and $\mathit{next}^M_v(i) = s'_i$ if $i \in M$.
A progress set $M$ is \emph{safe} if for every $i < j$, and every $a \in P_{i, j}$, $\ell_\calT(\mathit{next}^M_v(i))(a) = \ell_\calT(\mathit{next}^M_v(j))(a)$, i.e., $M$ ensures that all coloring constraints are satisfied \emph{locally} in the next step.

\union*
\begin{proof}
	To show that $M_1 \cup M_2$ is safe, let $i < j$ be arbitrary. 
	We claim that for every $a \in P_{i, j}$, $\ell_\calT(\mathit{next}^{M_1 \cup M_2}_v(i))(a) = \ell_\calT(\mathit{next}^{M_1 \cup M_2}_v(j))(a)$, i.e, $\mathit{next}^{M_1 \cup M_2}_v(i)$ and $\mathit{next}^{M_1 \cup M_2}_v(j)$ have the same $P_{i, j}$-color.
	We perform a case distinction on $(M_1 \cup M_2) \cap \{i, j\}$:
	\begin{itemize}[leftmargin=*]
		\item If $(M_1 \cup M_2) \cap \{i, j\} = \emptyset$: That is, none of the two copies are progressed in $M_1 \cup M_2$, so $\mathit{next}^{M_1 \cup M_2}_v(i) = s_i$ and $\mathit{next}^{M_1 \cup M_2}_v(j) = s_j$, and the claim follows directly from assumption \textbf{(A)}. 
		\item If $(M_1 \cup M_2) \cap \{i, j\} = \{i\}$ (or, analogously, $(M_1 \cup M_2) \cap \{i, j\} = \{j\}$):
		This implies that $\mathit{next}^{M_1 \cup M_2}_v(i) = s_i'$ and $\mathit{next}^{M_1 \cup M_2}_v(j) = s_j$.
		Now $(M_1 \cup M_2) \cap \{i, j\} = \{i\}$ implies that $M_1 \cap \{i, j\} = \{i\}$ or $M_2 \cap \{i, j\} = \{i\}$. 
		W.l.o.g., we assume the former, so $\mathit{next}^{M_1 \cup M_2}_v(i) = s_i' = \mathit{next}^{M_1}_v(i)$ and $\mathit{next}^{M_1 \cup M_2}_v(j) = s_j = \mathit{next}^{M_1}_v(j)$. 
		By the assumption that $M_1$ is a safe progress set, we get that $s_i'$ and $s_j$ have the same $P_{i, j}$-color.
		
		\item  If $(M_1 \cup M_2) \cap \{i, j\} = \{i, j\}$:
		This can occur in two cases (excluding symmetric cases): 
		\begin{itemize}[leftmargin=*]
			\item If $M_1 \cap \{i, j\} = \{i, j\}$ (or, analogously, $M_2 \cap \{i, j\} = \{i, j\}$): The claim then follows directly from  the assumption that $M_1$ is safe. 
			\item $M_1 \cap \{i, j\} = \{i\}$ and $M_2 \cap \{i, j\} = \{j\}$ (or, analogously, $M_1 \cap \{i, j\} = \{j\}$ and $M_2 \cap \{i, j\} = \{i\}$): 
			As $M_1$ is safe, we have that $\mathit{next}^{M_1}_v(i) = s_i'$ has the same $P_{i, j}$-color as $\mathit{next}^{M_1}_v(j) = s_j$. 
			Likewise, as $M_2$ is safe, we have that $\mathit{next}^{M_2}_v(i) = s_i$ has the same $P_{i, j}$-color as $\mathit{next}^{M_2}_v(j) = s_j'$. 
			Moreover, by \textbf{(A)} $s_i$ has the same $P_{i, j}$-color as $s_j$. 
			Together this implies that $\mathit{next}^{M_1 \cup M_2}_v(i) = s_i'$ has the same $P_{i, j}$-color as $\mathit{next}^{M_1 \cup M_2}_v(j) = s_j'$ as desired.
			\qedhere
		\end{itemize}
	\end{itemize}
\end{proof}

We can now formally define a winning strategy for the verifier. 
For simplicity, we sketch a non-positional strategy, i.e., a strategy that does not only depend on the current vertex but the previous finite play. 
As Büchi (and parity) games ar positionally determined \cite{martin1975borel}, this already implies the existence of a positional strategy (cf.~\Cref{sec:prelim}).

In the initial vertex in $\calG_{\calT, \varphi, 1}$, the verifier can observe the initial state window of all traces $\pi_1, \ldots, \pi_n$. 
This suffices to evaluate all state formulas in $\psi$ (as they only depend on the first state). 
Based on this evaluation, the verifier can decide if satisfying $\psi_\mathit{phase}$ is necessary to satisfy $\psi$. 
\begin{itemize}[leftmargin=*]
	\item If this is not the case, i.e., $\psi$ already holds based on the state formulas alone, the verifier repeatedly progresses all stutterings $\beta_1, \ldots, \beta_n$. This ensures some fair stuttering and satisfying $\psi_\mathit{phase}$ is not necessary anyway.
	\item If not, i.e., the verifier needs to satisfy $\psi_\mathit{phase}$ in order to fulfill $\psi$, it needs to progress such that the color sequences align. 
	As already outlined in \Cref{sub:admissible}, we archive this by using \Cref{lem:union} and always progressing as many traces as possible. 
	For each vertex $v$, we define $M^\mathit{max}_v$ as the unique maximal safe progress set in vertex $v$ (defined as the union of all safe progress sets), which is safe according to \Cref{lem:union}.
	For each vertex $v$, the verifier now progresses  all traces in $M^\mathit{max}_v$, i.e., we define 
	\begin{align*}
		\sigma^\mathit{max}(v) := \big\langle \ustage, \big[\pi_i \mapsto [s_i, s_i'] \big]_{i=1}^n, \big[\beta_i \mapsto \mathit{ite}(i \in M^\mathit{max}_v, 1, 0) \big]_{i=1}^n, \{\beta_i \mid i \in M_v^\mathit{max}\}, q \big\rangle.
	\end{align*}
	That is, the verifier always (locally) progresses as many stutterings as possible while ensuring that the coloring constraints in $\psi_\mathit{phase}$ holds in the next step.
\end{itemize}

\noindent
Let $\sigma^\mathit{max}$ be the resulting strategy for the verifier. 
It is not hard to see that $\sigma^\mathit{max}$ wins $\calG_{\calT, \varphi, 1}$:
Strategy $\sigma^\mathit{max}$ always progresses as many stutterings as possible in order to satisfy the coloring constraints in $\psi_\mathit{phase}$.
On any paths $p_1, \ldots, p_n$ for $\pi_1, \ldots, \pi_n$ that satisfy $\exists \beta_1 \stutter \pi_1, \ldots, \beta_n \stutter \pi_n \ldot \psi_\mathit{phase}$, strategy  $\sigma^\mathit{max}$ automatically aligns them such that $\psi_\mathit{phase}$ holds.
That is each trace is progressed as long as possible and if it cannot be progressed without violating $\psi_\mathit{phase}$, it automatically ``waits'' until it can change its color \emph{together} with the other traces. 
Consequently, we get: 

\begin{restatable}{lemma}{strat}
	By following strategy $\sigma^\mathit{max}$, the verifier wins $\game{\calT}{\varphi}{1}$. 
\end{restatable}

\fi

\end{document}
\endinput